\documentclass[journal]{IEEEtran}
\usepackage{amsmath}
\usepackage{amssymb}
\usepackage{cite}
\usepackage{graphicx}
\usepackage{epstopdf}
\usepackage{amsthm}
\usepackage{color}
\usepackage{algorithm}
\usepackage{algpseudocode}
\usepackage{amsmath}
\usepackage{cancel}
\usepackage{mathtools}
\newcommand{\bs}{\boldsymbol}
\newtheorem{theorem}{Theorem}

\newtheorem{remark}{Remark}
\DeclareMathAlphabet{\mathpzc}{OT1}{pzc}{m}{it}
\usepackage{tikz}

\begin{document}
\title{Zero-Forcing Based Downlink Virtual MIMO-NOMA Communications in IoT Networks}
%\title{Zero-Forcing NOMA Systems over Fully Correlated Rayleigh MIMO Channels}
\author{Zheng~Shi,
    Hong~Wang,
    Yaru~Fu,
    Guanghua~Yang,
    Shaodan~Ma,
    Fen~Hou,
    and Theodoros A. Tsiftsis
\thanks{Copyright (c) 20xx IEEE. Personal use of this material is permitted. However, permission to use this material for any other purposes must be obtained from the IEEE by sending a request to pubs-permissions@ieee.org.}
\thanks{Manuscript received July 6, 2019; revised October 1, 2019; revised November 11, 2019; accepted November 27, 2019. This work was supported in part by the National Key Research and Development Program of China under Grant 2017YFE0120600, in part by National Natural Science Foundation of China under Grants 61801192, 61801246 and 61601524, in part by Guangdong Basic and Applied Basic Research Foundation under Grant 2019A1515012136, in part by the Science and Technology Planning Project of Guangdong Province under Grants 2018B010114002 and 2019B010137006, in part by the Science and Technology Development Fund, Macau SAR (File no. 0036/2019/A1, File no. 037/2017/AMJ and File no. SKL-IOTSC2018-2020), in part by Open Research Foundation of National Mobile Communications Research Laboratory of Southeast University under Grant 2018D09, and in part by the Research Committee of University of Macau under Grants MYRG2018-00156-FST and MYRG2016-00171-FST. The associate editor coordinating the review of this paper and approving it for publication was Jun Wu. (Corresponding author: Guanghua Yang.)}%in part by joint fund by Ministry of Science and Technology of the China and Macau Science and Technology Development under grant 037/2017/AMJ
\thanks{Zheng Shi, Guanghua Yang and Theodoros A. Tsiftsis are with the Institute of Physical Internet and the School of Intelligent Systems Science and Engineering, Jinan University, Zhuhai 519070, China (e-mails: zhengshi@jnu.edu.cn, ghyang@jnu.edu.cn, theo\_tsiftsis@jnu.edu.cn).}
\thanks{H. Wang is with the School of Communication and Information Engineering, Nanjing University of Posts and Telecommunications, Nanjing 210003, China, and also with the National Mobile Communications Research Laboratory, Southeast University, Nanjing 210096, China (e-mail: wanghong@njupt.edu.cn).}
\thanks{Yaru~Fu is with the Department of Information Systems Technology and Design, Singapore University of Technology and Design, Singapore 487372 (email:yaru\_fu@sutd.edu.sg).}
\thanks{Shaodan~Ma and Fen~Hou are with the State Key Laboratory of Internet of Things for Smart City and the Department of Electrical and Computer Engineering, University of Macau, Macao S.A.R. 999078, China (e-mails: shaodanma@um.edu.mo, fenhou@um.edu.mo).}
}
\maketitle
\begin{abstract}

To support massive connectivity and boost spectral efficiency for internet of things (IoT), a downlink scheme combining virtual multiple-input multiple-output (MIMO) and non-orthogonal multiple access (NOMA) is proposed. All the single-antenna IoT devices in each cluster cooperate with each other to establish a virtual MIMO entity, and multiple independent data streams are requested by each cluster. NOMA is employed to superimpose all the requested data streams, and each cluster leverages zero-forcing detection to de-multiplex the input data streams. Only statistical channel state information (CSI) is available at base station to avoid the waste of the energy and bandwidth on frequent CSI estimations. The outage probability and goodput of the virtual MIMO-NOMA system are thoroughly investigated by considering Kronecker model, which embraces both the transmit and receive correlations. Furthermore, the asymptotic results facilitate not only the exploration of physical insights but also the goodput maximization. In particular, the asymptotic outage expressions provide quantitative impacts of various system parameters and enable the investigation of diversity-multiplexing tradeoff (DMT). Moreover, power allocation coefficients and/or transmission rates can be properly chosen to achieve the maximal goodput. By favor of Karush-Kuhn-Tucker conditions, the goodput maximization problems can be solved in closed-form, with which the joint power and rate selection is realized by using alternately iterating optimization. Besides, the optimization algorithms tend to allocate more power to clusters under unfavorable channel conditions and support clusters with higher transmission rate under benign channel conditions.

\end{abstract}
\IEEEpeerreviewmaketitle
\hyphenation{HARQ}
\section{Introduction}\label{sec:int}
% essentially consider a special implementation of NOMA only a pair of users
%became a of increasing importance paradigm to
%hybrid MA system
\IEEEPARstart{W}{ith} the phenomenal growth of the internet of things (IoT) devices in 5G scenarios (including smart city, connected health, industrial internet, vehicle network), a huge number of wireless-enabled sensors are broadly deployed. A great deal of sensing data is collected sporadically, and then forwarded to the cloud via connected IoT devices, cellular networks and internet \cite{han2019energy,hao2018green}. It has been predicted in Ericsson's report that the number of short-range and cellular IoT devices will reach 17.8 billion and 4.1 billion by 2024, respectively \cite{ericsson2019ericsson}. Moreover, the visual networking index (VNI) released by Cisco has forecasted that machine to machine (M2M) connections in IoT networks will grow from 6.1 billion in 2017 to 14.6 billion by 2022 \cite{cisco2018cisco}. The proliferation of IoT devices will pose the necessity of massive wireless connections to confront the ongoing paradigm of 5G and beyond. %The rapid growth of machine type communications (MTC)
% forwarded/reported to the application server in the cloud
%Internet of Things (IoT) is no longer a phenomenon, but it has become a prevalent system in which people, processes, data, and things connect to the Internet and each other. Globally, M2M connections will grow 2.4-fold, from 6.1 billion in 2017 to 14.6 billion by 2022 (Figure 10). There will be 1.8 M2M connections for each member of the global population by 2022.
%Ericsson predicts that around 3.5 billion cellular IoT devices, i.e. MTC devices (MTCDs), will be widely deployed by 2023--Ericsson,“Ericsson mobility report,” https://www.ericsson.com/en/mobility-report, Jun. 2018.
In addition, the non-orthogonal multiple access (NOMA) has been recognized as a promising technology enabler for IoT to offer massive connectivity over limited wireless resources \cite{duan2019resource,alkhawatrah2019buffer,ding2015application,al2019efficient}. The essence of NOMA is to serve multiple users over the same frequency/time/code resource by means of superposition coding and successive interference cancellation (SIC). In particular, NOMA is capable of attaining an improved spectral efficiency over its conventional counterpart, orthogonal multiple access (OMA). In addition, NOMA can also strike a balanced tradeoff between the system throughput and users' fairness. All these intrinsic characteristics of NOMA are beneficial to the development and deployment of IoT. Moreover, in order to get rid of the high complexity and severe error propagation of SIC, the combination of NOMA and OMA technologies has been recognized as a feasible solution for resource-limited IoT networks \cite{liu2016cooperative,ding2015application}. Particularly, the hybrid multiple access (MA) scheme first groups users into pairs to perform NOMA, then multiplexes different pairs over orthogonal resources.

%We totally agree with the reviewer that the complexity of SIC scales with the number of superposed signals and severe error propagation would take place in incorrect SIC decoding. Hence, most of existing literature frequently assume that users are grouped into pairs (two-layer signals, i.e., $K=2$) to carry out NOMA. However, as we think if $K=2$ is assumed in this paper, it might limit its further applications to future potential scenarios. Therefore, our analysis examines a general case by encompassing the case of $K=2$ as a special one. Additionally, to address the above-mentioned implementation issues, hybrid multiple access (MA) has been recognized as a potential solution, which combines NOMA with conventional orthogonal multiple access (OMA) technologies \cite{liu2016cooperative}, \cite{ding2015application}.  More specifically, the users are first clustered into pairs to perform NOMA, and the conventional OMA is used to multiplex different pairs over orthogonal resources (e.g., time/frequency/code).

%IoT applications requires NOMA to support massive connectivity...

%enhanced technology like MIMO is attractive solution to...

Despite the prominent advantages of NOMA, other enhanced technologies are also needed to further boost the throughput and reliability of IoT networks. Note that multiple-input-multiple-output (MIMO) can exploit extra degrees of freedom \cite{lv2018degrees,wu2019efficient}, and thus the integration of NOMA with MIMO has became an appealing solution for both massive connectivity support and capacity provision recently \cite{dai2015non}. %MIMO-NOMA systems have gained noticeable research interests recently.
\subsection{Related Works}
%To be specific,

The maximization of the ergodic capacity was studied for MIMO-NOMA system with a pair of users in \cite{sun2015ergodic}. To accommodate as many pairs of users as possible, the concept of signal alignment was devised to cancel out inter-pair interference within MIMO-NOMA framework in \cite{liu2016capacity,ding2016general}. By using the signal alignment, the signal receptions between different pairs of users can be handled independently. Hence, the investigations of MIMO-NOMA in \cite{sun2015ergodic,ding2016general,liu2016capacity} were essentially specific to single/multiple pairs of NOMA users. To support more than two users in each NOMA group, multi-user MIMO-NOMA transmission strategies were developed in \cite{wang2018novel,wang2019power} by leveraging group interference cancellation in conjunction with minimum mean square error (MMSE) detector. To meet dynamic QoS requirements in IoT communications, novel precoding and detection strategies were proposed for MIMO-NOMA systems to differ users' effective channel gains in \cite{ding2016mimo}. A coordinated multipoint (CoMP)-enabled NOMA scheme was proposed to attain an improved spectral efficiency in \cite{eryani2019generalized}, where distributed single-antenna base stations (BSs) collaborate together to form a virtual MIMO. Nonetheless, most of the existing works assumed frequent feedback of channel state information (CSI) at the transmitter. This assumption leads to high energy consumption, signaling and computational overhead, which is apparently impractical for bandwidth-limited, energy-constrained and latency-sensitive IoT networks \cite{ding2014performance,yang2016performance}. In \cite{ding2016design,ding2016application}, more effective CSI feedback was considered into NOMA-MIMO schemes. For further suppression of system overhead, the sender was assumed to have only statistical knowledge of CSI \cite{gong2018antenna,gong2019application,tong2019mimo,cui2018outage,shi2018performance}. To be specific, the rate and outage performance was analyzed for single-stream MIMO-NOMA systems in \cite{gong2018antenna}. In regard to the similar systems, a power allocation scheme was further proposed for maximizing the ergodic sum rate in \cite{gong2019application}. %, based on which the impact of user clustering was analyzed. but with only statistical characteristic of channel matrix at the transmitter  but with only statistical characteristic of channel matrix at the transmitter
Furthermore, a linear MIMO-NOMA model with complex-valued power coefficients was derived by considering multi-stream transmissions in \cite{tong2019mimo}. Unlike \cite{gong2018antenna,gong2019application,tong2019mimo} that assume independent fading channels, the joint design of power allocation and receive beamforming was developed for single-stream MIMO-NOMA systems with receive correlation in \cite{cui2018outage}. Whereas, the outage performance of multi-stream MIMO-NOMA systems was examined in \cite{shi2018performance} by considering transmit correlation.

IoT devices usually equip with a single antenna for the sake of limited size, low cost and energy saving, which in turn hampers the exploitation of spatial multiplexing gain from MIMO technology. To combat this issue, virtual MIMO technology has emerged as an attractive candidate by coordinating multiple nodes to form one virtual antenna array entity \cite{dohler2002space,li2017performance}. The realization of virtual MIMO brings new challenges of data sharing, time/frequency synchronization, user grouping/clustering, and CSI estimation \cite{chang2016low,alexiou2014wireless}. Nevertheless, it was disclosed in \cite{cui2004energy,jayaweera2006virtual} that the virtual MIMO system surpasses the simple single-input single-output (SISO) system in terms of the consumed energy and delay. In \cite{rafique2013performance}, the virtual MIMO was proved to be a viable framework to improve the data rate and energy efficiency. Furthermore, experimental measurements were carried out for the virtual MIMO to confirm its throughput enhancement in indoor-to-outdoor scenarios in \cite{dey2019throughput}. In order to prolong the lifetime of IoT networks as well as improve its capacity, the application of virtual MIMO to IoT has received an ever-increasing research attention recently\cite{feng2018uav,kisseleff2016distributed,youssef2018joint,dey2019virtual}. To name a few, in \cite{feng2018uav}, a multi-antenna unmanned aerial vehicle (UAV) serves as an on-demand flexible platform, which communicates with a cluster of single-antenna IoT devices through a virtual MIMO link. In \cite{kisseleff2016distributed}, a virtual MIMO approach was invented for the delivery of sensor data to a distant sink node in wireless body area networks (WBAN). %Virtual MIMO was further combined with low complex low density parity check (LDPC) code to reduce power consumption and improve the bit error rate (BER) performance for WBAN in \cite{youssef2018joint}.

\subsection{Scope of the Paper}
To offer massive number of connections and implement MIMO transmissions in IoT scenarios, this paper proposes a novel virtual MIMO-NOMA transmission scheme for IoT communications. More specifically, downlink communications between a BS and multiple clusters are taken into account, wherein all the single-antenna IoT devices in each cluster collaborates with each other to create a virtual MIMO entity, and a number of independent data streams are requested by each cluster. On one hand, NOMA technique is employed at the BS to simultaneously serve all the clusters in power-domain. On the other hand, each cluster capitalizes on zero-forcing (ZF) detection to de-multiplex all the desired input data streams. To avoid the waste of the energy and bandwidth on frequent channel estimations, the BS is assumed to have only statistical knowledge of CSI. An in-depth analysis of the outage and goodput performance of the virtual MIMO-NOMA system is carried out, where a more practical channel model, i.e., Kronecker model, is adopted to account for both the transmit and receive correlations. More importantly, the asymptotic outage probability is derived in a compact and simple form, which offers not only insightful conclusions but also a simple analytical treatment for the optimal system design. %Moreover, the power allocation coefficients and/or transmission rates can be properly chosen to maximize the goodput of the virtual MIMO-NOMA system.
With the asymptotic results, the maximization of the goodput is enabled with closed-form solution. Moreover, a modified version of the proposed scheme, i.e., virtual MIMO-hybrid MA, is devised to balance the tradeoff between two conflicting goals of enhancing spectral efficiency and saving computation/energy resources for IoT networks. Numerical verifications are finally presented. %

\subsection{Contributions}
The main contributions of this paper can be summarized as follows:
\begin{enumerate}
  \item The outage probability and goodput of the downlink virtual MIMO-NOMA-aided IoT communication system are scrutinized in this paper. To characterize the impact of spatial correlation between antennas in realistic propagation environments, Kronecker model is employed to capture both the transmit and receive correlations. However, Kronecker channel model poses significant challenges to the performance evaluation of the proposed scheme. As the most fundamental performance metric, the outage probability can be accurately approximated in a compact form by invoking the methodologies of moment generating function (MGF) and multivariate analysis.
  \item The asymptotic outage probability is then derived in closed-form to ease the extraction of insightful results. In particular, the asymptotic outage probability quantifies the impacts of spatial correlation, transmission rates and power allocation coefficients. In addition, the diversity order is $\mathpzc d = N_r-M+1$, where $M$ and $N_r$ stand for the number of data streams requested by each cluster and the number of IoT devices within each cluster, respectively.
  \item Furthermore, diversity-multiplexing tradeoff (DMT) is studied to figure out the multiplexing benefit of MIMO-NOMA channels.
  \item The asymptotic outage probability also enables the maximization of goodput in closed-form by favor of Karush-Kuhn-Tucker (KKT) conditions. Three optimization algorithms are thereon developed for the proper selection of the transmission powers and/or rates. It is revealed that the developed algorithms incline to allocate more power to clusters under unfavorable channel conditions and support clusters with higher transmission rate under benign channel conditions. %Additionally, the numerical results show that the proposed joint power and rate optimization algorithm achieves a considerable performance gain over the other two algorithms relying solely on power or rate optimization.
\end{enumerate}

\subsection{Outline and Notations}
The remainder of the paper is outlined as follows. The virtual MIMO-NOMA system model is presented together with two fundamental performance metrics in Section \ref{sec:sys_mod}. Section \ref{sec:per_ana} conducts the exact and asymptotic analyses of the outage probability. With the analytical results, the goodput maximization is then enabled in Section \ref{sec:goodput_max}. The numerical results are shown for verifications in Section \ref{sec:num}.

\emph{Notations:} We shall use the following notations throughout the paper. Bold uppercase and lowercase letters are used to denote matrices and vectors, respectively. ${\bf X}^{\mathrm{T}}$, ${\bf X}^{\mathrm{H}}$, ${\bf X}^{-1}$ and ${\bf X}^{1/2}$ denote the transpose, conjugate transpose, matrix inverse and Hermitian square root of matrix ${\bf X}$, respectively. ${\bf X}^\dag$ represents the Moore-Penrose pseudo-inverse of $\bf X$ and ${\bf X}^\dag = {\left( {{{\bf{X}}}^{\rm{H}}{{\bf{X}}}} \right)^{ - 1}}{{\bf{X}}}^{\rm{H}}$. $\mathrm{vec}$, $\rm{tr}$, $\mathrm{det}$ and $\mathrm{diag}$ are the operators of vectorization, trace, determinant and diagonalization, respectively. %$\Delta \left( {\bf X} \right)$ refers to the Vandermonde determinant of the eigenvalues of matrix ${\bf X}$.
$[{\bf X}]_{ij}$ is the $(i,j)$-th element of $\bf X$. $\mathbf{0}_n$ and $\mathbf{I}_n$ stand for $1 \times n$ all-zero vector and $n \times n$ identity matrix, respectively. %$\mathbb C$ and complex numbers and
$\mathbb C^{m\times n}$ denotes the set of $m\times n$-dimensional complex matrices. The symbol ${\rm i}=\sqrt{-1}$ is the imaginary unit. ${{\mathrm{E}}_A}$ denotes the operator of the expectation taking over random variable $A$. %$\Re{(s)}$ denotes the real part of the complex number $s$.
$o(\cdot)$ denotes little-O notation. $(\cdot)_n$ represents Pochhammer symbol. %$|S|$ refers to the cardinality of set $S$.
$\bigcup{(\cdot)}$ stands for the union operation.
The symbol $\simeq$ means ``asymptotically equal to''. $\propto$ stands for ``directly proportional to''. $\Gamma(x)$ denotes Gamma function. $G^{m,n}_{p,q}(\left. \cdot \right|x)$ stands for Meijer G-function. $\gamma(s,x)$ denotes the lower incomplete gamma function. Any other notations will be defined in the place where they take place.
%non-orthogonal multiple access (NOMA)
%
%To suppress the multiuser interference,
%
%% We may consider distributed-MIMO NOMA with zero-forcing!
%% \cite{bashar2017performance}: D-MIMO sufficiently spaced antenna
%
%multiplex multiple users
%1、C. Li, B. Xia, Q. Jiang, Y. Yao, and G. Yang, "Achievable Rate of the Multi-User Two-Way Full-Duplex Relay System," IEEE Transactions on Vehicular Technology, to be published, doi: 10.1109/TVT.2018.2789843.
%2. K. Xiao, B. Xia, Z. Chen, J. Wang, D. Chen, and S. Ma, "On Optimizing Multicarrier-Low Density Codebook for GMAC with Finite Alphabet Inputs," IEEE Communications Letters, Vol. 21, no. 8, pp. 1811-1814, Aug. 2017
%3、K. Xiao, B. Xia, Z. Chen, B. Xiao, D. Chen, and S. Ma, "On Capacity-based Codebook Design and Advanced Decoding for Sparse Code Multiple Access Systems," IEEE Transactions on Wireless Communications, to be published, doi: 10.1109/TWC.2018.2816929
%\cite{li2018achievable}
%\cite{xiao2017optimizing}
%\cite{xiao2018capacity}
\section{System Model and Performance Metrics}\label{sec:sys_mod}
\begin{figure}[!t]
  \centering
  \includegraphics[width=3.2in]{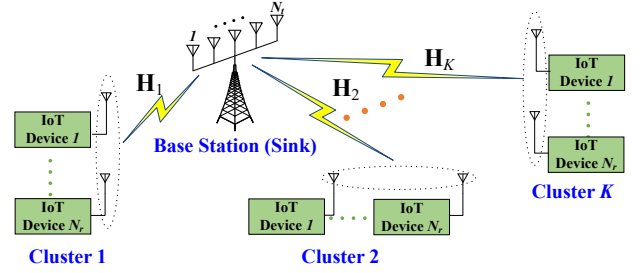}
  \caption{A downlink virtual MIMO-NOMA system model.}\label{fig:sys}
%\hrulefill
%\vspace*{4pt}
\end{figure}
%In this paper, a downlink single-cell network is considered, where the users are uniformly distributed in a disk with radius $D$ and the base station is located at the center. Specifically, the users within the disc are assumed to follow a PPP with intensity $\lambda$. Moreover, the base station is equipped with $N_t$ antennas and each user is equipped with $N_r$ antennas. We assume that the base station delivers $M (\le \min{(N_t,N_r)})$ independent data streams to multiple users at the same frequency/time/code via NOMA transmission. In each transmission, the number of users is assumed to be $K$. Then the received signal at the $k$-th nearest user to the base station is given by
\subsection{System Model}
As shown in Fig. \ref{fig:sys}, this paper considers a downlink virtual MIMO-NOMA system with one BS/central sink node serving $K$ clusters in IoT networks. Herein, we assume that the BS is equipped with $N_t$ antennas. Besides, the $K$ clusters are predefined to ease management, and each cluster consists of $N_r$ neighbouring single-antenna IoT devices\footnote{In general, the pairing/clustering process could be either predefined or dynamically scheduled. In particular, similarly to \cite{feng2018uav}, the predefined clustering divides all the IoT devices into multiple clusters, each with $N_r$ adjacent devices. It is noteworthy that our subsequent analysis is also applicable to the case of clusters with unequal sizes. Whereas, with regard to the dynamic clustering, the BS first sends a message to notify all the IoT devices within the cell so as to trigger the clustering process via a self-organizing protocol. Upon receiving the notification message, each device periodically broadcasts a life messages (LifeMsg) containing its unique identity, and meanwhile searches its neighbouring devices within a certain range by overhearing their LifeMsgs. After the completion of the search, all the neighbouring devices collaborate and share their antennas with each other to form a cluster. Each cluster then reports its relative position (or path loss) and statistical knowledge of CSI (e.g., channel correlations, average channel gains) to the BS. Once the BS gathers all the clusters' reports, it determines the decoding order of SIC for these clusters according to their relative positions, and then the transmission rates and powers are properly devised based on the channel statistics \cite{choi2016harq,liu2016cooperative,chingoska2016resource}. The necessary information (e.g., decoding order) for detections is thereafter fed back from the BS to the clusters. Nonetheless, the further investigation of implementing the pairing/clustering is out of the scope of this paper. For more details, interested readers are referred to \cite{ding2016application,ding2015application,ding2015impact}.}\cite{feng2018uav}. These devices share their antennas with each other through an error-free intra-cluster interface, in order to create a virtual MIMO system\cite{buratti2011multihop,ding2016transmission}. We can easily extend our design to the general case with more than one antennas mounted on the IoT devices. This is based on the fact that an IoT device with multiple antennas can be viewed as multiple single-antenna IoT devices. %It is worth mentioning that the $K$ clusters are predefined to facilitate analysis \cite{feng2018uav}. %the number of antennas mounted on the BS is denoted by , while each cluster in IoT cluster equips with one single antenna.  and , respectively
Moreover, we assume that $M (\le \min{(N_t,N_r)})$ independent data streams are requested by each cluster. In doing so, a power-domain NOMA is used to multiplex these clusters over the same time-frequency resource with the help of superposition coding at the BS. %the BS delivers In each transmission, the number of users is assumed to be $K$.
Thus, the received signal at the $k$-th nearest cluster to the BS is given by
\begin{equation}
{{\bf{y}}_k} = {\sqrt{\ell(d_k)}}{{\bf{H}}_{k}}{{\bf{V}}}{{\bf{s}}} + {{\bf{n}}_k},
\end{equation}
where ${{\bf{H}}_{k}} \in {\mathbb C}^{N_r\times N_t}$ denotes the channel response matrix from the BS to cluster $k$; ${\bf V}=({\bf v}_1,\cdots,{\bf v}_M) \in {\mathbb C}^{N_t \times M}$ represents the transmit beamforming matrix at the BS, i.e., $\left\| {{{\bf{v}}_m}} \right\|=1$ for $m=1,\cdots,M$; ${{\bf{s}}} = ({s}_1,\cdots,{s}_M)^{\rm T}$ is the superimposed symbol vector at the BS, and the entries of the transmit symbol vector are drawn independently from a complex symmetric circularly normal distributions with zero-mean and variance of $P$, i.e., ${{\bf{s}}} \sim {\cal C N}({\bf 0}, P{\bf I}_M)$, and $P$ denotes the total transmit power of each data stream; ${{\bf{n_k}}}$ is the additive Gaussian white noise (AWGN) vector with variance $\sigma^2$, i.e., ${{\bf{n_k}}} \sim {\cal C N}({\bf 0}, \sigma^2{\bf I}_{N_r})$. A composite channel model is applied to encompass both small-scale Rayleigh fading and large-scale path loss. In particular, the path loss is modelled by using Friis equation, such that $\ell(d_k)=\mathcal K {d_k}^{-\alpha}$ with path loss exponent $\alpha \ge 2$, where $d_k$ denotes the distance between the BS and cluster $k$, $\mathcal K $ represents the free space power received at the reference distance $d_0=1$m. It is noticeable that all the IoT devices within each cluster are assumed to suffer from almost identical large-scale fadings because of a sufficiently long distance between the BS and clusters relative to the cluster size \cite{qu2010cooperative}. %, $\mathcal K=\left(\frac{\lambda}{4\pi d_0}\right)^2$ and $\lambda$ stands for the singal wavelength in meters.
Without loss of generality, the distances between the BS and clusters are sorted in an increasing order such that $d_1 \le d_2 \le \cdots \le d_K$.

Due to the close spacing between adjacent antenna elements no matter at the BS or within each cluster, the resultant mutual antenna coupling brings about the spatial correlation between elements in realistic propagation environments\cite{shin2008mimo}. Particularly in downlink IoT communications, high correlation between the transmit antennas may occur in the absence of local scattering \cite{tian2017overlapping}. Furthermore, in IoT network, devices or sensors located indoors (e.g., airports, malls, house and buildings) usually suffer from shadowing effects. The shadowing induced by the same obstacle (tree/building) will incur high correlation between propagation paths\cite{qu2010cooperative}. Overlooking the existence of spatial correlation will overestimate the system performance and offer misleading design guide. In order to capture the negative impact of the antenna correlation at transceivers, Kronecker channel model is widely adopted by decomposing the spatial correlation into two independent components, i.e., transmit and receive correlations\cite{martin2004asymptotic}. According to the Kronecker model, the channel matrix between the BS and cluster $k$, ${{\bf{H}}_{k}}$, is represented by \cite{larsson2008space,paulraj2003introduction}
\begin{equation}\label{eq:corr_mod}
{{\bf{H}}_{k}} = {{\bf R}_r}^{1/2} {\bf H}_w {{\bf R}_t}^{1/2},
\end{equation}
where ${\bf R}_t$ and ${\bf R}_r$ denote transmit- and receive-side covariance matrices, respectively, the elements of matrix ${\bf H}_w$ are independent and identically distributed (i.i.d.) complex zero-mean unit-variance Gaussian random variables. Thus, the distribution of ${\bf H}_k$ follows as ${\rm vec}\left({\bf H}_k\right) \sim {\cal C N}({\bf 0}_{N_tN_r},  {{\bf R}_t}^{\rm T}\otimes {\bf R}_{{{r} }})$. It is worthwhile to point out that ${\bf R}_t$ and ${\bf R}_r$ are positive semi-definite Hermitian matrices. The validity of Kronecker model has been corroborated regardless of antenna configurations and intra-array spacing, provided that the transmit and receive angular power profiles are separable \cite{clerckx2013mimo}. Besides, in order to conserve energy, relieve signaling overhead and avoid frequent CSI feedback in large-scale energy- and bandwidth-limited IoT networks, we assume that only statistical CSI is available to the BS \cite{ding2014performance,yang2016performance}. % but the knowledge of CSI is perfectly known at  side.
%In wireless transmission, high correlation can be induced between propagation paths by shadowing if they are blocked by the same obstacle, such as a tree or a building
%mutual antenna coupling and close spacing between adjacent elements [10]
%Both the  are considered in this paper.
%More specifically, ${{\bf{H}}_{k}}$ is a complex circularly symmetric Gaussian matrix with zero mean, transmit-side covariance matrix ${\bf R}_t$ and receive-side covariance matrix ${\bf R}_r$. Hence, ${{\bf{H}}_{k}}$ can be written as , where the notation $(\cdot)^{\rm T}$ refers to the transpose of a matrix

Since NOMA is employed to simultaneously accommodate multiple clusters by sharing the same physical resources, the information bearing vector of $M$ data streams can be constructed by using superposition coding as follows
\begin{equation}
{\bf{s}} = \sqrt P \sum\limits_{k = 1}^K {{\rm diag}\left\{\sqrt{\zeta  _{1,k}},\cdots,\sqrt{\zeta  _{M,k}}\right\}{{\bf{x}}_k}} ,
\end{equation}
where ${{\bf{x}}_k}=(x_{1,k},\cdots,x_{M,k})^{\rm T}$ is the signal vector intended for cluster $k$, $\zeta _{m,k}$ stands for the power allocation coefficient of signal $x_{m,k}$, which is subject to the normalization constraint $\sum\nolimits_{k = 1}^K {{\zeta _{m,k}}}  = {1}/{M}$ for all $ m \in [1,M]$, to ensure the fairness among different data streams. On the other hand, the application of NOMA can guarantee clusters' fairness by taking advantage of the difference in channel qualities between the BS and clusters\cite{ding2016application}.

The linear ZF detection is a vital technology to offer low complexity and implementation cost for IoT communications\cite{ding2019success,ding2016mimo}. Hence, we assume that the received signal is reconstructed by using a ZF detector. Specifically, the ZF estimate of the transmitted symbol streams at cluster $k$ can be obtained by multiplying ${\bf y}_k$ with $({\sqrt{\ell(d_k)}}{{{{\bf{H}}_k}{\bf{V}}} })^\dag$ as%pseudo-inverse of ${\sqrt{\ell(d_k)}}{{{{\bf{H}}_k}{\bf{V}}} }$ such that
\begin{align}
{\bf{\hat s}} %&= \left({\sqrt{\ell(d_k)}}{{{{\bf{H}}_k}{\bf{V}}} }\right)^\dag{{\bf y}_k} \notag\\
&= {\bf{s}} + \frac{1}{\sqrt{\ell(d_k)}}\left({{{{\bf{H}}_k}{\bf{V}}} }\right)^\dag{{\bf{n}}_k}.
\end{align}
With the linear ZF receiver, the decodings for $M$ data streams can be independently performed. This eventually induces a very low decoding complexity.

After ZF detection, the SIC as one of the key components of NOMA technique is applied to separate superimposed signals sequentially. To simplify the design of power allocation coefficient and better utilize NOMA scheme, the decoding order at receivers is determined according to the decreasing order of transmission distances, i.e., $d_K \ge  \cdots \ge  d_2 \ge d_1$\cite{shi2018cooperative,liu2016cooperative}. %\cite{choi2016harq}, i.e., $d_K \ge  \cdots \ge  d_2 \ge d_1$
In particular, cluster $k$ will detect the message of cluster $i$ prior to decoding its own message ${\bf x}_k$ such that $i > k$, and then subtracting the signal ${\bf x}_{i}$ from its observation successively. In the meantime, the signal ${\bf x}_{i}$ intended for cluster $i$ ($ < k$) will be regarded as interference at cluster $k$. Therefore, for the $k$-th nearest cluster to the BS, the instantaneous received signal-to-interference-plus-noise ratio (SINR) to detect the signal $x_{m,i}$ that is extracted from data stream $m$ is given by
\begin{align}\label{eqn:sinr_x_im}
{\gamma _{m,k \to i}}%\notag\\
&= \frac{{P{\zeta _{m,i}}}}{{P\sum\limits_{l = 1}^{i - 1} {{\zeta _{m,l}}}  + \frac{\sigma ^2}{\ell ({d_k})}{{\left[ {{{\left( {{{\bf{H}}_k}{\bf{V}}} \right)}^\dag }{{\left( {{{\left( {{{\bf{H}}_k}{\bf{V}}} \right)}^\dag }} \right)}^{\rm{H}}}} \right]}_{mm}}}} \notag\\
&=\frac{{{{\bar \gamma }}{\ell ({d_k})}{\zeta _{m,i}}}}{{{{\bar \gamma }}{\ell ({d_k})}\sum\limits_{l = 1}^{i - 1} {{\zeta _{m,l}}}  + {{\left[ {{{\left( {{{\left( {{{\bf{H}}_k}{\bf{V}}} \right)}^{\rm{H}}}{{\bf{H}}_k}{\bf{V}}} \right)}^{ - 1}}} \right]}_{mm}}}},
\end{align}
where $ i \ge k$ and ${{\bar \gamma }} = {{P}}/{{{\sigma ^2}}}$ denotes the average transmit SNR.

\subsection{Performance Metrics}
%The most concerned problem is the error probability, which refers to the outage probability under the assumption of Gaussian codes and typical set decoding.

The absence of perfect instantaneous CSI at the BS undoubtedly results in decoding errors at receivers. In this paper, we are concerned with the error/outage probability. More specifically, the event that any cluster fails to decode any of its own data streams can be deemed as an error/outage. By following the proposed scheme, the outage event takes place if the cluster (e.g., cluster $k$) is unable to either subtract the interfering streams (e.g., $x_{m,i}$, $i>k$) or recover its own data stream (e.g., $x_{m,k}$). From an information-theoretical perspective, the outage probability of cluster $k$ for decoding its own data stream $m$ is therefore expressed as%By assuming applying Gaussian codes and typical set decoding can be measured by using the concept of
\begin{align}\label{eqn:out_pro_def}
p_{m,k}^{out} = \Pr \left[ {\bigcup\limits_{i = k}^K {{{\log }_2}\left( {1 + {\gamma _{m,k \to i}}} \right) < {R_{m,i}}} } \right],%\\
%& = \Pr \left[ {\bigcup\limits_{i = k}^K {{{\left[ {{{\left( {{{\left( {{{\bf{H}}_k}{\bf{V}}} \right)}^{\rm{H}}}{{\bf{H}}_k}{\bf{V}}} \right)}^{ - 1}}} \right]}_{mm}} > \frac{{{{\bar \gamma }}{\alpha _i}}}{{{2^{{R_{m,i}}}} - 1}} - \left( {{{\bar \gamma }}\sum\limits_{l = 1}^{i - 1} {{\alpha _l}} } \right)} } \right]
%= \Pr \left[ {{{\left[ {{{\left( {{{\left( {{{\bf{H}}_k}{\bf{V}}} \right)}^{\rm{H}}}{{\bf{H}}_k}{\bf{V}}} \right)}^{ - 1}}} \right]}_{mm}} > {{{\bar \gamma }{\theta _{m,k}}}{\ell ({d_k})}}} \right],
%=\Pr \left[ {\frac{1}{{{{\left[ {{{\left( {{{\left( {{{\bf{H}}_k}{\bf{V}}} \right)}^{\rm{H}}}{{\bf{H}}_k}{\bf{V}}} \right)}^{ - 1}}} \right]}_{mm}}}} < \frac{1}{{{{\bar \gamma }}{\theta _{m,k}}\ell ({d_k})}}} \right],
\end{align}
where ${R_{m,i}}$ is the target transmission rate for data stream $x_{m,k}$.
%, i.e., $x_{m,k}$,
% if and only if the mutual information for $x$
%From an information-theoretical perspective,
%This paper focuses on the outage probability is the most concerned performance metric

Since the occurrence of decoding failures will yield the deterioration of the target achievable rate, the long-term average throughput should be precisely estimated. Towards this goal, the concept of goodput is commonly used to measure the number of information bits successfully delivered per transmission\cite{rui2008combined}. The goodput of the virtual MIMO-NOMA system can be expressed in terms of outage probabilities as
\begin{equation}\label{eqn:goodput_def}
{T_g} = \sum\limits_{m = 1}^M {\sum\limits_{k = 1}^K {\left( {1 - p_{m,k}^{out}} \right){R_{m,k}}} }.
\end{equation}
Unambiguously, the outage probability plays a crucial role in performance evaluation. For the sake of thorough investigation, we have to derive a closed-form and tractable expression for the outage probability first. With the analytical result, we are concerned with the maximization of the goodput for the virtual MIMO-NOMA system in this paper. %For example, the goodput can be maximized by properly choosing the power allocation coefficients and/or transmission rates.% in Section \ref{sec:goodput_max}. % \ref{sec:per_ana}
\section{Analysis of Outage Probability}\label{sec:per_ana}
Although the outage probability of the MIMO-NOMA system under semi-correlated MIMO fading channels has been analyzed in \cite{shi2018performance}, the analytical results are inapplicable to that under full-correlated MIMO fading channels, as shown in \eqref{eq:corr_mod}. The Kronecker correlation model in \eqref{eq:corr_mod} significantly challenges the exact analysis of the outage probability. Moreover, the asymptotic analysis of the outage probability is also conducted to gain more helpful insights.
\subsection{Exact Analysis}
%Following the proposed virtual MIMO-NOMA scheme, given the distance $d_k$ and the total number of clusters $K$, the probability of that cluster $k$ fails to decode its own message for data stream $m$ can be expressed as
By putting \eqref{eqn:sinr_x_im} into \eqref{eqn:out_pro_def}, the probability of the event that cluster $k$ fails to decode its own data stream $m$, $p_{m,k}^{out}$, can be simplified as
\begin{align}\label{eqn:out_pro_def_rew}
p_{m,k}^{out}&= \notag\\
%= \Pr \left[ {\bigcup\limits_{i = k}^K {{{\log }_2}\left( {1 + {\gamma _{m,k \to i}}} \right) < {R_{m,i}}} } \right]\\
%& = \Pr \left[ {\bigcup\limits_{i = k}^K {{{\left[ {{{\left( {{{\left( {{{\bf{H}}_k}{\bf{V}}} \right)}^{\rm{H}}}{{\bf{H}}_k}{\bf{V}}} \right)}^{ - 1}}} \right]}_{mm}} > \frac{{{{\bar \gamma }}{\alpha _i}}}{{{2^{{R_{m,i}}}} - 1}} - \left( {{{\bar \gamma }}\sum\limits_{l = 1}^{i - 1} {{\alpha _l}} } \right)} } \right]
%= \Pr \left[ {{{\left[ {{{\left( {{{\left( {{{\bf{H}}_k}{\bf{V}}} \right)}^{\rm{H}}}{{\bf{H}}_k}{\bf{V}}} \right)}^{ - 1}}} \right]}_{mm}} > {{{\bar \gamma }{\theta _{m,k}}}{\ell ({d_k})}}} \right],
%&\Pr \left[ {\frac{1}{{{{\left[ {{{\left( {{{\left( {{{\bf{H}}_k}{\bf{V}}} \right)}^{\rm{H}}}{{\bf{H}}_k}{\bf{V}}} \right)}^{ - 1}}} \right]}_{mm}}}} < \frac{1}{{{{\bar \gamma }}{\theta _{m,k}}\ell ({d_k})}}} \right],
&\Pr \left[ {{{{{\left[ {{{\left( {{{\left( {{{\bf{H}}_k}{\bf{V}}} \right)}^{\rm{H}}}{{\bf{H}}_k}{\bf{V}}} \right)}^{ - 1}}} \right]}_{mm}}}} > {{{{\bar \gamma }}{\theta _{m,k}}\ell ({d_k})}}} \right],
\end{align}
where %${R_{m,i}}$ is the predefined transmission rate,
${\theta _{m,k}} = \min \left\{ {\left. {{{{\zeta _{m,i}}}}/{({{2^{{R_{m,i}}}} - 1})} - \sum\nolimits_{l = 1}^{i - 1} {{\zeta _{m,l}}} } \right|k \le i \le K} \right\}$. It is worth noting that the condition of ${\theta _{m,k}}>0$ should be always satisfied to successfully eliminate inter-cluster interference while performing SIC for NOMA. Hereby, the constraints $\left\{{\theta _{m,k}}>0,k\in[1,K]\right\}$ lead to $\left\{{R_{m,i}} < {\log _2}\left( {1 + {{{\zeta _{m,i}}}}/{{\sum\nolimits_{l = 1}^{i - 1} {{\zeta _{m,l}}} }}} \right),\,i\in[2,K]\right\}$.

By defining ${\bf{Z}} = {{{\bf{H}}_k}{\bf{V}}}$, the elements of ${\bf{Z}}$ are circularly symmetric complex Gaussian distributed as ${\rm vec}\left({\bf Z}\right) \sim {\cal C N}({\bf 0},  {{{\bf{R}}_{t'}}}^{\rm T}\otimes {\bf R}_{{{r} }})$, where ${{\bf{R}}_{t'}} = {{\bf{V}}^{\rm{H}}}{{\bf{R}}_t}{\bf{V}}$. Then, ${\bf{Z}}$ can be rewritten as ${\bf{Z}} = {{\bf{R}}_r}^{1/2}{{\bf{Z}}_w}{{\bf{R}}_{t'}}^{1/2}$, where ${{\bf{Z}}_w}$ denotes a $N_r \times M$ complex Gaussian matrix with i.i.d. ${\cal C N}(0, 1)$ entries. With the above definitions, $p_{m{},k}^{out}$ can be rewritten as
\begin{equation}\label{eqn:out_1_def}
p_{m{},k}^{out} = %\Pr \left[ {\frac{1}{{{{\left[ {{{\left( {{{\bf{Z}}^{\rm{H}}}{\bf{Z}}} \right)}^{ - 1}}} \right]}_{11}}}} < \frac{1}{{{{\bar \gamma }}{\theta _{1,k}}\ell ({d_k})}}} \right].
\Pr \left[ {{{{{\left[ {{{\left( {{{\bf{Z}}^{\rm{H}}}{\bf{Z}}} \right)}^{ - 1}}} \right]}_{mm}}}} > {{{{\bar \gamma }}{\theta _{m,k}}\ell ({d_k})}}} \right].
\end{equation}
Unfortunately, there are only a few relevant research results concerning \eqref{eqn:out_1_def} so far. More specifically, by solely considering the transmit correlation, \cite[Theorem 3.2.10]{muirhead2009aspects} can be employed to get \eqref{eqn:out_1_def} in closed-form \cite{gore2002transmit,ding2016design,shi2018performance}. However, the introduction of the receive correlation will exceedingly complicate the performance analysis due to the intractability of tackling the integral involving hypergeometric functions of matrix argument. By employing the vector-matrix partition, M. Kiessling \emph{et al.} in \cite{kiessling2003analytical} investigated the symbol error rate of a $2\times 2$ MIMO case with correlations occurring at both transmit- and receive-sides. However, there is still no general result available for deriving \eqref{eqn:out_1_def}. Motivated by this, the outage probability for arbitrary data stream $m$, $p_{m,k}^{out}$, can be accurately approximated as
%Substituting \eqref{eqn:cdf_F_mean_re1} and (\ref{eqn:cdf_X_k}) into \eqref{eqn:p_fxk}, the approximate expression of $p_{m=1,k}^{out}$ can be derived. In the same fashion, by combining \eqref{eqn:out_1_xdef} and \eqref{eqn:cdf_F_mean_re1}, the outage probability for arbitrary data stream $m$, $p_{m,k}^{out}$, can be generalized as
%(\ref{eqn:out_1_xdeffina}), at the top of the next page.
%\begin{figure*}[!t]
\begin{align}\label{eqn:out_1_xdeffina}
p_{m,k}^{out} \approx  \tilde F \left( \frac{{{{\left[ {{{\bf{R}}_{t'}}^{ - 1}} \right]}_{mm}}}}{{{{\bar \gamma }}{\theta _{m,k}}\ell ({d_k})}} \right),
%\Pr \left[ {{X_k} < \frac{{{{\left[ {{{\bf{R}}_{t'}}^{ - 1}} \right]}_{11}}}}{{{{\bar \gamma }}{\theta _{1,k}}\ell ({d_k})}}} \right]
%\frac{{\Gamma \left( {{N_r}} \right)}}{{\Gamma \left( {{N_r} - M + 1} \right)}}\frac{1}{{\det \left( {\left( {{\lambda _i}^{{N_r} - j}} \right)} \right)}}\times\\
%\left| {\begin{array}{*{20}{c}}
%{{{\left( {{\lambda _1}} \right)}^{{N_r} - 1}}G_{2,3}^{2,1}\left( {\left. {\begin{array}{*{20}{c}}
%{1,{N_r}}\\
%{1,{N_r} - M + 1,0}
%\end{array}} \right|\frac{{{{\left[ {{{\bf{R}}_{t'}}^{ - 1}} \right]}_{mm}}}}{{{\lambda _1}{{\bar \gamma }}{\theta _{m,k}}\ell ({d_k})}}} \right)}&{{\lambda _1}^{{N_r} - 2}}& \cdots &1\\
%{{{\left( {{\lambda _2}} \right)}^{{N_r} - 1}}G_{2,3}^{2,1}\left( {\left. {\begin{array}{*{20}{c}}
%{1,{N_r}}\\
%{1,{N_r} - M + 1,0}
%\end{array}} \right|\frac{{{{\left[ {{{\bf{R}}_{t'}}^{ - 1}} \right]}_{mm}}}}{{{\lambda _2}{{\bar \gamma }}{\theta _{m,k}}\ell ({d_k})}}} \right)}&{{\lambda _2}^{{N_r} - 2}}& \cdots &1\\
% \vdots & \vdots & \ddots & \vdots \\
%{{{\left( {{\lambda _{{N_r}}}} \right)}^{{N_r} - 1}}G_{2,3}^{2,1}\left( {\left. {\begin{array}{*{20}{c}}
%{1,{N_r}}\\
%{1,{N_r} - M + 1,0}
%\end{array}} \right|\frac{{{{\left[ {{{\bf{R}}_{t'}}^{ - 1}} \right]}_{mm}}}}{{{\lambda _{{N_r}}}{{\bar \gamma }}{\theta _{m,k}}\ell ({d_k})}}} \right)}&{{\lambda _{{N_r}}}^{{N_r} - 2}}& \cdots &1
%\end{array}} \right|.
\end{align}
where $\tilde F \left( x \right)$ is given by (\ref{eqn:cdf_X_k}) at the top of the next page, ${\lambda _1}, \cdots ,{\lambda _{N_r}}$ denote the $N_r$ eigenvalues of ${{\bf{R}}_r}$%, $\Gamma(\cdot)$ denotes Gamma function, $G^{m,n}_{p,q}(\left. \cdot \right|\cdot)$ stands for Meijer G-function \cite{gradshteyn1965table}
and the proof is detailed in Appendix \ref{app:F_xk}.
\begin{figure*}[!t]
\begin{equation}\label{eqn:cdf_X_k}
\tilde F \left( x \right)=
\frac{{\Gamma \left( {{N_r}} \right)\det \left( {{{\left\{ {{\lambda _i}^{{N_r} - 1}G_{2,3}^{2,1}\left( {\left. {\begin{array}{*{20}{c}}
{1,{N_r}}\\
{1,{N_r} - M + 1,0}
\end{array}} \right|\frac{x}{{{\lambda _i}}}} \right)} \right\}}_{\scriptstyle1 \le i \le {N_r}\hfill\atop
\scriptstyle j = 1\hfill}},{{\left\{ {{\lambda _i}^{{N_r} - j}} \right\}}_{\scriptstyle1 \le i \le {N_r}\hfill\atop
\scriptstyle2 \le j \le {N_r}\hfill}}} \right)}}{{\Gamma \left( {{N_r} - M + 1} \right)\det \left( {\left\{ {{\lambda _i}^{{N_r} - j}} \right\}_{1 \le i,j \le {N_r}}} \right)}}.
%\frac{{\Gamma \left( {{N_r}} \right)}}{{\Gamma \left( {{N_r} - M + 1} \right)}} \frac{1}{{\det \left( {\left( {{\lambda _i}^{{N_r} - j}} \right)} \right)}}
%\left| {\begin{array}{*{20}{c}}
%{{{\left( {{\lambda _1}} \right)}^{{N_r} - 1}}G_{2,3}^{2,1}\left( {\left. {\begin{array}{*{20}{c}}
%{1,{N_r}}\\
%{1,{N_r} - M + 1,0}
%\end{array}} \right|\frac{x}{{{\lambda _1}}}} \right)}&{{\lambda _1}^{{N_r} - 2}}& \cdots &1\\
%{{{\left( {{\lambda _2}} \right)}^{{N_r} - 1}}G_{2,3}^{2,1}\left( {\left. {\begin{array}{*{20}{c}}
%{1,{N_r}}\\
%{1,{N_r} - M + 1,0}
%\end{array}} \right|\frac{x}{{{\lambda _{\rm{2}}}}}} \right)}&{{\lambda _2}^{{N_r} - 2}}& \cdots &1\\
% \vdots & \vdots & \ddots & \vdots \\
%{{{\left( {{\lambda _{{N_r}}}} \right)}^{{N_r} - 1}}G_{2,3}^{2,1}\left( {\left. {\begin{array}{*{20}{c}}
%{1,{N_r}}\\
%{1,{N_r} - M + 1,0}
%\end{array}} \right|\frac{x}{{{\lambda _{{N_r}}}}}} \right)}&{{\lambda _{{N_r}}}^{{N_r} - 2}}& \cdots &1
%\end{array}} \right|.
\end{equation}
%\hrulefill
%\vspace*{4pt}
\end{figure*}
It it noticeable that the approximation in \eqref{eqn:out_1_xdeffina} will become an equality whenever ${ {{\bf{R}}_r}}=\lambda_0{\bf I}_{N_r}$. Clearly, the approximation in \eqref{eqn:out_1_xdeffina} would become accurate if fading channels experience a low correlation at cluster's side. In practice, this assumption makes sense because IoT devices are placed sufficiently far apart from each other relative to the close-spacing between transmit antennas at the base station. Moreover, \eqref{eqn:out_1_xdeffina} can be further simplified by considering the following two special cases.
\subsubsection{$M=1$}
Suppose that only a single data stream is requested by each cluster in downlink IoT communications, i.e., $M=1$, the Meijer-G function in \eqref{eqn:cdf_X_k} reduces to
\begin{align}\label{eqn:meijer_G}
G_{2,3}^{2,1}\left( {\left. {\begin{array}{*{20}{c}}
{1,{N_r}}\\
{1,{N_r} - M + 1,0}
\end{array}} \right|x} \right) &= G_{1,2}^{1,1}\left( {\left. {\begin{array}{*{20}{c}}
1\\
{1,0}
\end{array}} \right|x} \right) \notag\\
&= \gamma \left( {1,x} \right) = 1 - {e^{ - x}},
\end{align}
where the first equality holds by using the definition of Meijer G-function \cite{gradshteyn1965table}, and the second equality holds by using \cite[eq.06.07.26.0006.01]{wolframe2010math}. Hereby, the outage probability can be rewritten as \eqref{eqn:out_m_M1_fina}, shown at the top of the following page.
\begin{figure*}[!t]
\begin{equation}\label{eqn:out_m_M1_fina}
p_{m,k}^{out} = \frac{{\det \left( {{{\left\{ {{\lambda _i}^{{N_r} - 1}\left( {1 - {e^{ - \frac{{{{\left[ {{{\bf{R}}_{t'}}^{ - 1}} \right]}_{mm}}}}{{{\lambda _i}\bar \gamma {\theta _{m,k}}\ell ({d_k})}}}}} \right)} \right\}}_{\scriptstyle1 \le i \le {N_r}\hfill\atop
\scriptstyle j = 1\hfill}},{{\left\{ {{\lambda _i}^{{N_r} - j}} \right\}}_{\scriptstyle1 \le i \le {N_r}\hfill\atop
\scriptstyle2 \le j \le {N_r}\hfill}}} \right)}}{{\det \left( {\left\{ {{\lambda _i}^{{N_r} - j}} \right\}_{1 \le i,j \le {N_r}}} \right)}},
%\frac{1}{{\det \left( {\left( {{\lambda _i}^{{N_r} - j}} \right)} \right)}}\left| {\begin{array}{*{20}{c}}
%{{{\left( {{\lambda _1}} \right)}^{{N_r} - 1}}\left( {1 - {e^{ - \frac{{{{\left[ {{{\bf{R}}_{t'}}^{ - 1}} \right]}_{mm}}}}{{{\lambda _1}{{\bar \gamma }}{\theta _{m,k}}\ell ({d_k})}}}}} \right)}&{{\lambda _1}^{{N_r} - 2}}& \cdots &1\\
%{{{\left( {{\lambda _2}} \right)}^{{N_r} - 1}}\left( {1 - {e^{ - \frac{{{{\left[ {{{\bf{R}}_{t'}}^{ - 1}} \right]}_{mm}}}}{{{\lambda _2}{{\bar \gamma }}{\theta _{m,k}}\ell ({d_k})}}}}} \right)}&{{\lambda _2}^{{N_r} - 2}}& \cdots &1\\
% \vdots & \vdots & \ddots & \vdots \\
%{{{\left( {{\lambda _{{N_r}}}} \right)}^{{N_r} - 1}}\left( {1 - {e^{ - \frac{{{{\left[ {{{\bf{R}}_{t'}}^{ - 1}} \right]}_{mm}}}}{{{\lambda _{{N_r}}}{{\bar \gamma }}{\theta _{m,k}}\ell ({d_k})}}}}} \right)}&{{\lambda _{{N_r}}}^{{N_r} - 2}}& \cdots &1
%\end{array}} \right|.
\end{equation}
\hrulefill
%\vspace*{50pt}
\end{figure*}

\subsubsection{$\lambda_1=\cdots=\lambda_{N_r}=\lambda_0$}
If all the eigenvalues of ${{\bf{R}}_r}$ are identical, i.e., $\lambda_1=\cdots=\lambda_{N_r}=\lambda_0$, the approximation in \eqref{eqn:out_1_xdeffina} turns out to be an equality. By substituting \eqref{eqn:cdf_X_k} into \eqref{eqn:out_1_xdeffina}, we discover that both the denominator and numerator of \eqref{eqn:cdf_X_k} equal to zero because of identical rows in determinant. Hence, as proved in Appendix \ref{app:der_meijerG}, $p_{m,k}^{out}$ can be simplified as %(\ref{eqn:out_mKconst_corr}) at the top of the page after the next page.
%\begin{figure*}[!t]
%\begin{align}\label{eqn:out_mKconst_corr}
%p_{m,k}^{out}&= \frac{1}{{\prod\limits_{j = 1}^{{N_r} - 1} {j!} }}\frac{{\Gamma \left( {{N_r}} \right)}}{{\Gamma \left( {{N_r} - M + 1} \right)}}\left| {\begin{array}{*{20}{c}}
%{\frac{{{\partial ^{{N_r} - 1}}{{{{\lambda _1}}}^{{N_r} - 1}}G_{2,3}^{2,1}\left( {\left. {\begin{array}{*{20}{c}}
%{1,{N_r}}\\
%{1,{N_r} - M + 1,0}
%\end{array}} \right|\frac{{{{\left[ {{{\bf{R}}_{t'}}^{ - 1}} \right]}_{mm}}}}{{{\lambda _1}{{\bar \gamma }}{\theta _{m,k}}\ell ({d_k})}}} \right)}}{{\partial {\lambda _1}^{{N_r} - 1}}}}& \cdots &{\ddots}&{}\\
%0&{\frac{{{\partial ^{{N_r} - 2}}{\lambda _{{2}}}^{{N_r} - 2}}}{{\partial {\lambda _2}^{{N_r} - 2}}}}&{\cdots}&{}\\
% \vdots & \vdots & \ddots &{}\\
%0&0& \cdots &1
%\end{array}} \right|.
%\end{align}
%%\hrulefill
%%\vspace*{4pt}
%\end{figure*}
%Thus we have
%(\ref{eqn:out_mKconst_corr_spe}) at the top of the next page.
%\begin{figure*}[!t]
\begin{align}\label{eqn:out_mKconst_corr_spefina}
p_{m,k}^{out}&=\frac{{\gamma \left( {{N_r} - M + 1,\frac{{{{\left[ {{{\bf{R}}_{t'}}^{ - 1}} \right]}_{mm}}}}{{{\lambda _0}{{\bar \gamma }}{\theta _{m,k}}\ell ({d_k})}}} \right)}}{{\Gamma \left( {{N_r} - M + 1} \right)}},
\end{align}
which is consistent with the result in \cite{shi2018performance}.

\subsection{Asymptotic Analysis}
To extract more meaningful results from the exact analysis and facilitate the optimal system design as well, this paper focuses on the asymptotic behaviour of the outage probability under high SNR, i.e., $\bar \gamma \to \infty$. To this end, it is imperative to carry out the asymptotic analysis of $\tilde F \left( x \right)$ as $x\to 0$. As proved in Appendix \ref{app:F_tilde_exp}, $\tilde F \left( x \right)$ can be expanded as%split into two cases by considering whether $M=1$ or not, which affects the application of the Residue theorem to Meijer-G function in (\ref{eqn:out_1_xdeffina}).
\begin{align}\label{eqn:meijer_G_simple}
\tilde F &\left( x \right) = \frac{{\Gamma \left( {{N_r}} \right)}}{{\Gamma \left( {{N_r} - M + 1} \right)\det \left( {\left\{ {{\lambda _i}^{{N_r} - j}} \right\}} \right)}}\notag\\
&\times\sum\limits_{\vartheta  = {N_r} - M + 1}^{{N_r} - 1} {\frac{{{{\left( { - 1} \right)}^{2\vartheta  - {N_r} + M}}{{x}^\vartheta }}}{{\vartheta !\left( {\vartheta  - \left( {{N_r} - M + 1} \right)} \right)!\left( {{N_r} - 1 - \vartheta } \right)!}}} \notag\\
& \times \det \left( \left\{{{{{\lambda _i}} }^{{N_r} - \vartheta  - 1}}\ln {{\lambda _i}},\left. {{{ {{\lambda _i}}}^{{N_r} - j}}} \right|_{j = 2}^{{N_r}}\right\} \right)\notag\\
& + \frac{{\Gamma \left( {{N_r}} \right)}}{{\Gamma \left( {{N_r} - M + 1} \right)\det \left( {\left\{ {{\lambda _i}^{{N_r} - j}} \right\}} \right)}}\notag\\
&\times\sum\limits_{\vartheta  = {N_r}}^\infty  {\frac{{{{\left( { - 1} \right)}^{\vartheta  + M - 2}}\left( {\vartheta  - {N_r}} \right)!}}{{\vartheta !\left( {\vartheta  - {N_r} + M - 1} \right)!}}{x^\vartheta }} \notag\\
& \times \det \left( \left\{{{{{\lambda _i}}}^{{N_r} - \vartheta  - 1}},\left. {{{ {{\lambda _i}}}^{{N_r} - j}}} \right|_{j = 2}^{{N_r}}\right\} \right).
\end{align}
With the series expansion of $\tilde F \left( x \right)$ in \eqref{eqn:meijer_G_simple}, the asymptotic analysis of the outage probability is thus empowered. The first summation of (\ref{eqn:meijer_G_simple}) in the right-hand side of equation exists if and only if $M>1$, thus the
asymptotic analysis of $p_{m,k}^{out}$ should be separated into two cases, i.e., $M=1$ and $M>1$.
%To this end, the asymptotic analysis of the outage probability is split into two cases by considering whether $M=1$ or not, which affects the application of the Residue theorem to Meijer-G function in (\ref{eqn:out_1_xdeffina}).
\subsubsection{$M=1$}
If only a single data stream in each cluster is assumed, i.e., $M=1$, the asymptotic expression of $p_{m,k}^{out}$ as $\bar \gamma \to \infty$ can be obtained by putting \eqref{eqn:meijer_G_simple} into \eqref{eqn:out_1_xdeffina} as
\begin{align}\label{eqn:out_asy_M_1}
p_{m,k}^{out} &\simeq \frac{1}{{\det \left( {\left\{ {{\lambda _i}^{{N_r} - j}} \right\}} \right)}}\sum\limits_{\vartheta  = {N_r}}^\infty  {\frac{{{{\left( { - 1} \right)}^{\vartheta  - 1}}}}{{\vartheta !}}{{\left( {\frac{{{{\left[ {{{\bf{R}}_{t'}}^{ - 1}} \right]}_{mm}}}}{{\bar \gamma {\theta _{m,k}}\ell ({d_k})}}} \right)}^\vartheta }} \notag \\
&\quad \times\det \left( \left\{{{ {{\lambda _i}} }^{{N_r} - \vartheta  - 1}},\left. {{{ {{\lambda _i}}}^{{N_r} - j}}} \right|_{j = 2}^{{N_r}}\right\} \right)\notag \\
%&= \frac{{{{\left( { - 1} \right)}^{{N_r} - 1}}\det \left( \left\{{{{{\lambda _i}} }^{ - 1}},\left. {{{ {{\lambda _i}} }^{{N_r} - j}}} \right|_{j = 2}^{{N_r}}\right\} \right)}}{{{N_r}!\det \left( {\left\{ {{\lambda _i}^{{N_r} - j}} \right\}} \right)}}\notag\\
%&\times{\left( {\frac{{{{\left[ {{{\bf{R}}_{t'}}^{ - 1}} \right]}_{mm}}}}{{\bar \gamma {\theta _{m,k}}\ell ({d_k})}}} \right)^{{N_r}}} + o\left( {{{\bar \gamma }^{ - {N_r}}}} \right)\notag\\
&= \frac{1}{{{N_r}!\det \left( {{{\bf{R}}_r}} \right)}}{\left( {\frac{{{{\left[ {{{\bf{R}}_{t'}}^{ - 1}} \right]}_{mm}}}}{{\bar \gamma {\theta _{m,k}}\ell ({d_k})}}} \right)^{{N_r}}} + o\left( {{{\bar \gamma }^{ - {N_r}}}} \right),
\end{align}
where the last step holds by using %$o(\cdot)$ stands for the little-O notation and
\begin{multline}\label{eqn:det_idt}
\det \left( \left\{{{ {{\lambda _i}}}^{ - 1}},\left. {{{ {{\lambda _i}} }^{{N_r} - j}}} \right|_{j = 2}^{{N_r}}\right\} \right) = \\
{\left( { - 1} \right)^{{N_r} - 1}}\det {\left( {{{\bf{R}}_r}} \right)^{ - 1}}\det \left( {\left\{ {{\lambda _i}^{{N_r} - j}} \right\}} \right).
\end{multline}
\subsubsection{$M>1$}
Likewise, if $M>1$, by ignoring high order terms $o\left( {{{\bar \gamma }^{ - \left( {{N_r} - M + 1} \right)}}} \right)$ in the right-hand side of \eqref{eqn:meijer_G_simple}, $p_{m,k}^{out}$ is asymptotic to % (\ref{eqn:out_asy_Mg1}) at the top of the next page.
%\begin{figure*}[!t]
\begin{align}\label{eqn:out_asy_Mg1}
p_{m,k}^{out} &\simeq \frac{{{{\left( { - 1} \right)}^{{N_r} - M}}\left( {{N_r} - 1} \right)!}}{{\left( {{N_r} - M} \right)!\left( {{N_r} - M + 1} \right)!\left( {M - 2} \right)!
}}\notag\\
&\times \frac{\det \left( \left\{{{{{\lambda _i}}}^{M - 2}}\ln{{\lambda _i}},\left. {{{{{\lambda _i}}}^{{N_r} - j}}} \right|_{j = 2}^{{N_r}}\right\} \right)}{\det \left( \left\{{{\lambda _i}^{{N_r} - j}}\right\} \right)}\notag\\
&\times{\left( {\frac{{{{\left[ {{{\bf{R}}_{t'}}^{ - 1}} \right]}_{mm}}}}{{\bar \gamma {\theta _{m,k}}\ell ({d_k})}}} \right)^{{N_r} - M + 1}} + o\left( {{{\bar \gamma }^{ - \left( {{N_r} - M + 1} \right)}}} \right).
\end{align}
%\hrulefill
%\vspace*{4pt}
%\end{figure*}

Clearly from both \eqref{eqn:out_asy_M_1} and \eqref{eqn:out_asy_Mg1}, the impacts of spatial correlation, transmission rates, power allocation coefficients and the numbers of IoT devices and data streams on the outage performance can be quantified. Moreover, it is revealed that the diversity order achieved by the virtual MIMO-NOMA system at fixed rates is $\mathpzc d = N_r-M+1$, where the diversity order is defined as $\mathpzc d = \mathop {\lim }\nolimits_{\bar \gamma  \to \infty } {{\log {p_{m,k}^{out}}}}/{{\log \bar\gamma }}$. The diversity order reflects the decreasing slope of the outage probability relative to the transmit SNR. It is apparent that increasing the number of IoT devices within the cluster is beneficial to the improvement in the reception reliability. However, massive IoT devices in the cluster entails prohibitively high energy consumption and operating expenditure of channel estimation and signaling synchronization, which is unfavorable for IoT networks.

\subsection{Diversity-Multiplexing Tradeoff (DMT)}\label{sec:dmt}
To study the DMT in MIMO-NOMA channels, we assume that the power allocation coefficients scale as ${\zeta _{m,k}} \simeq {c_{m,k}}{{\bar \gamma }^{ - {\upsilon _{m,k}}}}$, where ${c_{m,k}}$ is a constant and ${\upsilon _{m,k}}\ge 0$. In particular, the power allocation scheme can be treated as a fixed one if ${\upsilon _{m,k}}=0$ for all $m\in [1,M]$ and $k\in [1,K]$. To reap the multiplexing benefit of MIMO-NOMA channels, let the transmission rate for the data stream of $x_{m,k}$ scale as ${R_{m,k}} = {r_{m,k}}{\log _2}\bar \gamma $, where ${r_{m,k}}\ge 0$ corresponds to multiplexing gain \cite{zheng2003diversity}. On the scale of interest, by following the above settings together with ${\theta _{m,k}}>0$ and $\sum\nolimits_{k = 1}^K {{\zeta _{m,k}}}  = {1}/{M}$, the conditions of ${\upsilon _{m,1}} \ge  \cdots  \ge {\upsilon _{m,K}} = 0$ and ${r_{m,i}} \le {\upsilon _{m,i - 1}} - {\upsilon _{m,i}}$ for $2 \le i \le K$ must be satisfied. Accordingly, on the basis of \eqref{eqn:out_asy_M_1} and \eqref{eqn:out_asy_Mg1}, the outage probability $p_{m,k}^{out}$ at target transmission rates ${R_{m,1}},\cdots,{R_{m,K}}$ is proportional to
\begin{align}\label{eqn:out_prop}
p_{m,k}^{out} &\propto {\left( {{{\bar \gamma {\theta _{m,k}}}}} \right)^{-({N_r} - M + 1)}} \notag\\
&\propto {\left( {{{\bar \gamma \min \left\{ {{c_{m,i}}{{\bar \gamma }^{ - {\upsilon _{m,i}} - {r_{m,i}}}},k \le i \le K} \right\}}}} \right)^{-({N_r} - M + 1)}}\notag\\
&\propto{{\bar \gamma }^{-\left( {{N_r} - M + 1} \right)\left( {1 - {\upsilon _{m,k}} - {r_{m,k}}} \right)}},
\end{align}
where the last step holds by using the inequality ${r_{m,i}} + {\upsilon _{m,i}} \le {\upsilon _{m,i - 1}} \le {r_{m,i - 1}} + {\upsilon _{m,i - 1}}$. Therefore, \eqref{eqn:out_prop} gives rise to a DMT for the data stream of $x_{m,k}$ as
\begin{align}\label{eqn:diver_gainDMT}
&d_{m,k}^*(\{r_{m,k}\}_{k\in[1,K]}) = -\lim\limits_{\bar \gamma \to \infty} \frac{\log p_{m,k}^{out}(\{{r_{m,k}}{\log _2}\bar \gamma\}_{k\in[1,K]})}{\log \bar \gamma} \notag\\
&=\left( {{N_r} - M + 1} \right)\left( {1 - {\upsilon _{m,k}} - {r_{m,k}}} \right),
\end{align}
where $d_{m,k}^*(\{r_{m,k}\}_{k\in[1,K]})$ represents the diversity gain for $x_{m,k}$, and the diversity gain at fixed rates typically refers to the diversity order. By recalling ${r_{m,k}} + {\upsilon _{m,k}} \le {r_{m,k - 1}} + {\upsilon _{m,k- 1}}$ and ${\upsilon _{m,K}}=0$, the diversity gains obey the following relationship
\begin{align}\label{eqn:diver_gain_relat}
&d_{m,1}^*(\{r_{m,k}\}_{k\in[1,K]})  \le d_{m,2}^*(\{r_{m,k}\}_{k\in[1,K]})\le \cdots \notag\\
& \le d_{m,K}^*(\{r_{m,k}\}_{k\in[1,K]})= ({{N_r} - M + 1})(1-{r_{m,K}}).
\end{align}
By taking the fixed power allocation scheme (i.e., ${\upsilon _{m,k}}=0$ for all $k \in [1,K]$) as an example, the condition of ${r_{m,i}} \le {\upsilon _{m,i - 1}} - {\upsilon _{m,i}}$ for $2 \le i \le K$ leads to vanishing multiplexing gains, i.e., ${r_{m,2}}=\cdots={r_{m,K}}=0$. This consequently yields the DMT of $d_{m,1}^*(\{r_{m,k}\}_{k\in[1,K]})= \left( {{N_r} - M + 1} \right)\left( {1  - {r_{m,1}}} \right)$ and $d_{m,2}^*(\{r_{m,k}\}_{k\in[1,K]})=\cdots= d_{m,K}^*(\{r_{m,k}\}_{k\in[1,K]})={{N_r} - M + 1}$.
%the diversity gain for the data stream of $x_{m,k}$ is achieved at multiplexing gain ${r_{m,k}}$ by

\section{Goodput Maximization}\label{sec:goodput_max}
%The goodput of the NOMA system is defined as
%\begin{equation}\label{eqn:goodput_def}
%{T_g} = \sum\limits_{m = 1}^M {\sum\limits_{k = 1}^K {\left( {1 - p_{m,k}^{out}} \right){R_{m,k}}} },
%\end{equation}
The optimal design of MIMO-NOMA system has been intensively studied by assuming perfect instantaneous CSI at the transmitter \cite{ding2016general,ding2016application}. However, this assumption limits their widespread applications into resource-constrained IoT scenarios, where only the statistical knowledge of CSI is usually known to the transmitter. Nevertheless, the optimal system settings were seldom discussed in the circumstances, because the optimal design using the exact outage probability would be unmanageable. To overcome this drawback, we resort to the tractable expression of the asymptotic outage probability to ease the system optimization. To exemplify this, the goodput maximization is considered next.

%With the asymptotic results, the optimal system design can be enabled with closed-form solution.

%The second key question is the optimal system design for the proposed scheme in the absence of perfect channel state information (CSI) at the transmitter. Whereas, in the literature, the parameters of the MIMO-NOMA system were optimally devised mostly based on the assumption of perfect/effective instantaneous CSI at the transmitter. To name a few, by imposing a fixed/dynamic quality of service (QoS) constraint on the target signal-to-interference-plus-noise-ratio (SINR), the power allocation coefficients were determined  by assuming effective CSI feedback in \cite{ding2016general}. Bearing the similar idea in mind, different power allocation strategies were examined for the MIMO-NOMA system based on signal alignment in \cite{}. However, the assumption of perfect CSI leads to high energy consumption, signaling and computational overhead, which is apparently impractical for bandwidth-limited, energy-constrained and latency-sensitive IoT networks.

Since reception failures are inevitable in the absence of perfect CSI in practice, the goodput is a pivotal performance metric to evaluate the achievable long term average throughput. It is clear that the goodput depends on the preset system configurations, such as transmission rates and power allocation coefficients. Accordingly, these adjustable parameters should be optimally devised by utilizing the channel statistical information systematically. In the sequel, the power allocation coefficients and/or transmission rates are properly chosen to maximize the goodput for the proposed scheme. %Moreover, the goodput maximization is taken as an example to justify the significance of the analytical results.

%are inevitable
%
%\eqref{eqn:goodput_def},
%
% Specifically, the goodput is maximized through optimally choosing the power allocation coefficients and/or transmission rates.

To enable the optimization with low computational complexity and close-form solution, the manageable expression of the asymptotic outage probability is adopted herein. For notational convenience, the asymptotic expression of the outage probability is unified by combining \eqref{eqn:out_asy_M_1} and \eqref{eqn:out_asy_Mg1} as
\begin{equation}\label{eqn:out_asy_gen}
p_{m,k}^{out} \simeq {\phi _{m,k}}{\theta _{m,k}}^{ - \left( {{N_r} - M + 1} \right)},
\end{equation}
where ${\phi _{m,k}}$ is given by (\ref{eqn:phi_def1}), shown at the top of the next page.
\begin{figure*}[!t]
\begin{equation}\label{eqn:phi_def1}
{\phi _{m,k}} = \left\{ {\begin{array}{*{20}{c}}
{\frac{1}{{{N_r}!\det \left( {{{\bf{R}}_r}} \right)}}{{\left( {\frac{{{{\left[ {{{\bf{R}}_{t'}}^{ - 1}} \right]}_{mm}}}}{{\bar \gamma \ell ({d_k})}}} \right)}^{{N_r}}},}&{M = 1}\\
{\frac{{{{\left( { - 1} \right)}^{{N_r} - M}}\left( {{N_r} - 1} \right)!\det \left( \left\{{{ {{\lambda _i}} }^{M - 2}}\ln {{\lambda _i}},\left. {{{{{\lambda _i}} }^{{N_r} - j}}} \right|_{j = 2}^{{N_r}}\right\} \right)}}{{\left( {{N_r} - M} \right)!\left( {{N_r} - M + 1} \right)!\left( {M - 2} \right)!\det \left( {\left\{ {{\lambda _i}^{{N_r} - j}} \right\}} \right)}}{{\left( {\frac{{{{\left[ {{{\bf{R}}_{t'}}^{ - 1}} \right]}_{mm}}}}{{\bar \gamma \ell ({d_k})}}} \right)}^{{N_r} - M + 1}},}&{M > 1}
\end{array}} \right.,
\end{equation}
%\hrulefill
%\vspace*{4pt}
\end{figure*}

\subsection{Optimal Power Allocation}\label{sec:optpowe}
Given the transmission rates, the problem of goodput maximization via power allocation is casted as
%Therefore, the optimization problem is casted as
\begin{equation}\label{eqn:good_max_math}
\begin{array}{*{20}{c l}}
{\mathop {\max }\limits_{{\left\{ {{\zeta _{m,k}}} \right\}_{ m \in [1,M]\hfill\atop
 k \in \left[ {1,K} \right]\hfill}}} }&{{T_g}}\\
{{\rm{s}}{\rm{.t}}{\rm{.}}}&{\sum\nolimits_{k = 1}^K {{\zeta _{m,k}}}  = \frac{1}{M},\,m \in [1,M],}\\
%{}&{p_{m,k}^{out}\le \varepsilon},m\in [1,M], k\in [1,K],\\
{}&{{\zeta _{m,k}} > 0,m\in [1,M], k\in [1,K],}\\
{}&{{{\log }_2}\left( {1 + \frac{{{\zeta _{m,k}}}}{{\sum\nolimits_{l = 1}^{k - 1} {{\zeta _{m,l}}} }}} \right) > {R_{m,k}}},\\
{}&{\quad \quad \quad \quad m\in [1,M] , k \in \left[ {2,K} \right]},\\
\end{array}
\end{equation}
where the first and last constraints are used to warrant the data streams' fairness and NOMA transmissions, respectively.

By introducing the auxiliary variables, the optimization problem can be reformulated as (\ref{eqn:opt_refor}) at the top of the next page.
\begin{figure*}[!t]
\begin{equation}\label{eqn:opt_refor}
\begin{array}{*{20}{c l}}
{\mathop {\max }\limits_{{\left\{ {{\zeta _{m,k}}} \right\}_{ m \in [1,M]\hfill\atop
 k \in \left[ {1,K} \right]\hfill}},{\left\{ {{\theta _{m,k}}} \right\}_{ m \in \left[ {1,M} \right]\hfill\atop
 k \in \left[ {1,K} \right]\hfill}}} }&{\sum\nolimits_{m = 1}^M {\sum\nolimits_{k = 1}^K {\left( {1 - {\phi _{m,k}}{\theta _{m,k}}^{ - \left( {{N_r} - M + 1} \right)}} \right){R_{m,k}}} } }\\
{{\rm{s}}{\rm{.t}}{\rm{.}}}&{\sum\nolimits_{k = 1}^K {{\zeta _{m,k}}}  = \frac{1}{M},\,m \in [1,M],}\\
{}&{{\zeta _{m,k}} > 0,m\in [1,M], k\in [1,K],}\\
%{}&{{\phi _{m,k}}{\theta _{m,k}}^{ - \left( {{N_r} - M + 1} \right)}\le \varepsilon},m\in [1,M], k\in [1,K],\\
{}&{{{\log }_2}\left( {1 + \frac{{{\zeta _{m,k}}}}{{\sum\nolimits_{l = 1}^{k - 1} {{\zeta _{m,l}}} }}} \right) > {R_{m,k}},{\mkern 1mu} k \in \left[ {2,K} \right],}\\
{}&{{\theta _{m,k}} = \min \left\{ {\left. {\frac{{{\zeta _{m,i}}}}{{{2^{{R_{m,i}}}} - 1}} - \sum\nolimits_{l = 1}^{i - 1} {{\zeta _{m,l}}} } \right|k \le i \le K} \right\}}.
\end{array}
\end{equation}
%\hrulefill
%\vspace*{4pt}
\end{figure*}
Notice that the objective function is an increasing function of $\theta_{m,k}$, the maximum goodput is achieved at the upper bound of $\theta_{m,k}$. The last constraint can thus be rewritten as ${{\theta _{m,k}} \le {{{\zeta _{m,i}}}}/{({{2^{{R_{m,i}}}} - 1})} - \sum\nolimits_{l = 1}^{i - 1} {{\zeta _{m,l}}} }$ equivalently. As a consequence, (\ref{eqn:opt_refor}) can be converted into a convex optimization problem with linear constraints, as shown in (\ref{eqn:opt_refor_rew}) at the top of the following page.
\begin{figure*}[!t]
\begin{equation}\label{eqn:opt_refor_rew}
\begin{array}{*{20}{c l}}
{\mathop {\min }\limits_{{\left\{ {{\zeta _{m,k}}} \right\}_{ m \in [1,M]\hfill\atop
 k \in \left[ {1,K} \right]\hfill}},{\left\{ {{\theta _{m,k}}} \right\}_{ m \in \left[ {1,M} \right]\hfill\atop
 k \in \left[ {1,K} \right]\hfill}}} }&{\sum\nolimits_{m = 1}^M {\sum\nolimits_{k = 1}^K {{\phi _{m,k}}{R_{m,k}}{\theta _{m,k}}^{ - \left( {{N_r} - M + 1} \right)}} } }\\
{{\rm{s}}{\rm{.t}}{\rm{.}}}&{\sum\nolimits_{k = 1}^K {{\zeta _{m,k}}}  = \frac{1}{M},\,m \in [1,M],}\\
{}&{{\zeta _{m,k}} > 0,m\in [1,M], k\in [1,K],}\\
{}&{{\theta _{m,k}} > 0,m\in [1,M], k\in [1,K],}\\
%{}&{{\theta _{m,k}} \ge \sqrt[{{N_r} - M + 1}]{{\frac{{{\phi _{m,k}}}}{\varepsilon }}},m\in [1,M], k\in [1,K],}\\
{}&{{\theta _{m,k}} \le \frac{{{\zeta _{m,i}}}}{{{2^{{R_{m,i}}}} - 1}} - \sum\nolimits_{l = 1}^{i - 1} {{\zeta _{m,l}}} ,m\in [1,M], k \in \left[ {1,K} \right],i \in \left[ {k,K} \right]}.
\end{array}
\end{equation}
\hrulefill
\vspace*{4pt}
\end{figure*}
Alternatively, the optimization problem can be solved by using KKT conditions. The optimal solution to (\ref{eqn:opt_refor_rew}) is given by the following theorem.
\begin{theorem}\label{the:pow}
The optimal power allocation coefficients ${{{{ \bs\zeta }}}_m^*} = {\left( {{{ \zeta }_{m,1}^*}, \cdots ,{{ \zeta }_{m,K}^*}} \right)^{\rm{T}}}$ are given by
  \begin{equation}\label{eqn:zeta_fina}
{{{\bs\zeta }}_m^*} = \frac{{{{\bf{L}}_m}^{ - 1}{{\bf{b}}_m}}}{{M{{\bf{1}}_K}^{\rm{T}}{{\bf{L}}_m}^{ - 1}{{\bf{b}}_m}}},
\end{equation}
where ${\bf 1}_K$ is an $K\times 1$ vector of ones, and
\begin{multline}\label{eqn:b_def}
{{\bf{b}}_m} =
\left( {{{\left( {\frac{{\left( {{N_r} - M + 1} \right){\phi _{m,1}}{R_{m,1}}}}{{{\bf e}_1}^{\rm T} {{\bf{U}}_m}^{ - 1}{{\bf{1}}_K}}} \right)}^{\frac{1}{{{N_r} - M + 2}}}}, \cdots ,} \right.\\
{\left. {{{\left( {\frac{{\left( {{N_r} - M + 1} \right){\phi _{m,K}}{R_{m,K}}}}{{{\bf e}_K}^{\rm T} {{\bf{U}}_m}^{ - 1}{{\bf{1}}_K}}} \right)}^{\frac{1}{{{N_r} - M + 2}}}}} \right)^{\rm{T}}},%\left( {\begin{array}{*{20}{c}}
%{{{\left( {\frac{{{2^{{R_{m,1}}}} - 1}}{{\left( {{N_r} - M + 1} \right){\phi _{m,1}}{R_{m,1}}}}} \right)}^{{N_r} - M + 2}}}\\
%{{{\left( {\frac{{{2^{{R_{m,2}}}} - 1}}{{\left( {{N_r} - M + 1} \right){\phi _{m,2}}{R_{m,2}}}}} \right)}^{{N_r} - M + 2}}}\\
% \vdots \\
%{{{\left( {\frac{{{2^{{R_{m,K}}}} - 1}}{{\left( {{N_r} - M + 1} \right){\phi _{m,K}}{R_{m,K}}}}} \right)}^{{N_r} - M + 2}}}
%\end{array}} \right).
\end{multline}
\begin{align}\label{eqn:L_def}
%{{\bf{L}}_m} = \left( {\begin{array}{*{20}{c}}
%{\frac{1}{{{2^{{R_{m,1}}}} - 1}}}&0& \cdots &0\\
%{ - 1}&{\frac{1}{{{2^{{R_{m,2}}}} - 1}}}& \cdots &0\\
% \vdots & \vdots & \ddots & \vdots \\
%{ - 1}&{ - 1}& \cdots &{\frac{1}{{{2^{{R_{m,K}}}} - 1}}}
%\end{array}} \right),
{{\bf{U}}_m} &= {{\bf{L}}_m}^{\rm{T}} \notag\\
&= \left( {\begin{array}{*{20}{c}}
{\frac{1}{{{2^{{R_{m,1}}}} - 1}}}&{ - 1}& \cdots &{ - 1}\\
0&{\frac{1}{{{2^{{R_{m,2}}}} - 1}}}& \cdots &{ - 1}\\
 \vdots & \vdots & \ddots & \vdots \\
0&0& \cdots &{\frac{1}{{{2^{{R_{m,K}}}} - 1}}}
\end{array}} \right),
\end{align}
and ${\bf e}_k$ denotes a column vector with one as its $k$-th element and zeros elsewhere.
Moreover, ${\theta _{m,k}^*}$ is given by
\begin{align}\label{eqn:theta_fin}
{\theta _{m,k}^*} &= \frac{{\zeta _{m,k}^*}}{{{2^{{R_{m,k}}}} - 1}} - \sum\limits_{l = 1}^{k - 1} {\zeta _{m,l}^*}
%= {{\bf{a}}_{m,k}}^{\rm{T}}{{{\bs\zeta }}_m} \notag\\
= \frac{{{{\bf{a}}_{m,k}}^{\rm{T}}{{\bf{L}}_m}^{ - 1}{{\bf{b}}_m}}}{{M{{\bf{1}}_K}^{\rm{T}}{{\bf{L}}_m}^{ - 1}{{\bf{b}}_m}}},
\end{align}
where ${{\bf{a}}_{m,k}} = {\left( {\underbrace { - 1, \cdots , - 1}_{k - 1},\frac{1}{{{2^{{R_{m,k}}}} - 1}},\underbrace {0, \cdots ,0}_{K - k}} \right)}^{\rm{T}}$.%$\theta _{m,1}^* < \cdots < \theta _{m,K}^*$ and
%\begin{equation}\label{eqn:am_def}
%{{\bf{a}}_{m,k}} = {\left( {\underbrace { - 1, \cdots , - 1}_{k - 1},\frac{1}{{{2^{{R_{m,k}}}} - 1}},\underbrace {0, \cdots ,0}_{K - k}} \right)}^{\rm{T}}.
%\end{equation}
\end{theorem}
\begin{proof}
  Please refer to Appendix \ref{app:pow}.
\end{proof}

%With (\ref{eqn:theta_fina}), ${\theta _{m,k}^*}$ can be expressed as (\ref{eqn:theta_fin})

By putting (\ref{eqn:theta_fin}) into (\ref{eqn:out_asy_gen}), $p_{m,k}^{out}$ can be obtained as
\begin{equation}\label{eqn:out_tar}
p_{m,k}^{out} \simeq {M^{{N_r} - M + 1}}{\phi _{m,k}}{\left( {\frac{{{{\bf{1}}_K}^{\rm{T}}{{\bf{L}}_m}^{ - 1}{{\bf{b}}_m}}}{{{{\bf{a}}_{m,k}}^{\rm{T}}{{\bf{L}}_m}^{ - 1}{{\bf{b}}_m}}}} \right)^{{N_r} - M + 1}}.
\end{equation}
Then plugging (\ref{eqn:out_tar}) into (\ref{eqn:goodput_def}) yields the asymptotic goodput.

Based on Theorem \ref{the:pow}, we arrive at the following remark regarding the power allocation coefficients.
\begin{remark}\label{the:rem}
The coefficients $\theta _{m,1}^*,\cdots,\theta _{m,K}^*$ follow the ascending order of their indices such that $\theta _{m,1}^* < \cdots < \theta _{m,K}^*$, which results in
\begin{equation}\label{eqn:normal_power}
\frac{{\zeta _{m,1}^*}}{{{2^{{R_{m,1}}}} - 1}} < \frac{{\zeta _{m,2}^*}}{{{2^{{R_{m,2}}}} - 1}} <  \cdots  < \frac{{\zeta _{m,K}^*}}{{{2^{{R_{m,K}}}} - 1}},
\end{equation}
where ${{\zeta _{m,k}^*}}/{({{2^{{R_{m,k}}}} - 1})}$ is the ratio of the power allocation coefficient to the received SNR threshold, which characterizes the cluster fairness. Accordingly, the result in (\ref{eqn:normal_power}) indicates that more transmission power is allocated to cluster under worse channel condition. This demonstrates that NOMA scheme can maintain the fairness among clusters while achieving the maximum goodput. On the other hand, this is consistent with the result of the DMT studied in Section \ref{sec:dmt} if we substitute the settings of ${\zeta _{m,k}}\simeq {c_{m,k}}{{\bar \gamma }^{ - {\upsilon _{m,k}}}}$ and ${R_{m,k}}= {r_{m,k}}{\log _2}\bar \gamma $ into \eqref{eqn:normal_power}.
%evaluates the average transmit energy per second per hertz of cluster $k$ on data stream $m$.
\end{remark}
\begin{proof}
  Please refer to Appendix \ref{app:rem}.
\end{proof}

\begin{remark}\label{the:remout}
The asymptotic outage probability is lower bounded as
\begin{equation}\label{eqn:out_opt_lower}
    p_{m,k}^{out\_asy} \ge {M^{{N_r} - M + 1}}{\phi _{m,k}}{\left( {{2^{{R_{m,k}}}} - 1} \right)^{{N_r} - M + 1}}.
\end{equation}
The right hand side of the inequality (\ref{eqn:out_opt_lower}) is equal to the asymptotic outage probability of cluster $k$ under OMA transmissions, in which all the power is allocated to cluster $k$'s data stream $m$.
\end{remark}
\begin{proof}
  Please refer to Appendix \ref{app:remout}.
\end{proof}
Roughly speaking from Remark \ref{the:remout}, the outage probability is an increasing function of transmission rate $R_{m,k}$. With the definition of the goodput, the transmission rate has two opposite effects on the goodput. Specifically, on one hand, the increase of transmission rate indicates more information bits delivered to the cluster. On the other hand, it results in high outage probability. Accordingly, the goodput approaches to zero no matter when $R_{m,k}$ tends to zero or infinity. All in all, the transmission rate should be well designed to attain the maximum goodput. However, by putting (\ref{eqn:out_tar}) into (\ref{eqn:goodput_def}), the complex expression of the goodput makes it fairly impossible to obtain the optimal transmission rate in closed-form, and even intractable to numerically solve the problem. Hence, a suboptimal algorithm is developed for the joint optimization of transmission powers and rates, which is deferred to Section \ref{sec:opt_joint}. %However, this is beyond the scope of this paper.
\subsection{Optimal Rate Selection}\label{sec:opt_rate}
Furthermore, for fixed values of the power allocation coefficients ${\zeta _{m,k}}$, the problem of goodput maximization through optimal rate selection while implementing NOMA transmissions is formulated as
\begin{equation}\label{eqn:good_max_mathrate}
\begin{array}{*{20}{c l}}
{\mathop {\max }\limits_{{\left\{ {{R _{m,k}}} \right\}_{ m \in [1,M]\hfill\atop
 k \in \left[ {1,K} \right]\hfill}}} }&{{T_g}}\\
{{\rm{s}}{\rm{.t}}{\rm{.}}}&{{R _{m,k}} > 0,m\in [1,M], k\in [1,K],}\\
{}&{{{\log }_2}\left( {1 + \frac{{{\zeta _{m,k}}}}{{\sum\nolimits_{l = 1}^{k - 1} {{\zeta _{m,l}}} }}} \right) > {R_{m,k}}},\\
{}&{\quad \quad \quad \quad m\in [1,M] , k \in \left[ {2,K} \right]}.
\end{array}
\end{equation}
Analogously, by introducing the auxiliary variables $\theta_{m,k}$, (\ref{eqn:good_max_mathrate}) can be further reformulated as (\ref{eqn:gp_max_rate}), shown at the top of the next page.
\begin{figure*}[!t]
\begin{equation}\label{eqn:gp_max_rate}
\begin{array}{*{20}{c}}
{\mathop {\max }\limits_{{{\left\{ {{R_{m,k}}} \right\}}_{m \in [1,M]\hfill\atop
k \in \left[ {1,K} \right]\hfill}},{{\left\{ {{\theta _{m,k}}} \right\}}_{m \in \left[ {1,M} \right]\hfill\atop
k \in \left[ {1,K} \right]\hfill}}} }&{\sum\nolimits_{m \in \left[ {1,M} \right]\hfill\atop
k \in \left[ {1,K} \right]\hfill} {\left( {1 - {\phi _{m,k}}{\theta _{m,k}}^{ - \left( {{N_r} - M + 1} \right)}} \right){R_{m,k}}} }\\
{{\rm{s}}.{\rm{t}}.}&{{R_{m,k}} > 0,m \in [1,M],k \in [1,K],}\\
{}&{{\theta _{m,k}} > 0,m \in [1,M],k \in [1,K],}\\
{}&{{R_{m,i}} \le {{\log }_2}\left( {1 + \frac{{{\zeta _{m,i}}}}{{{\theta _{m,k}} + \sum\nolimits_{l = 1}^{i - 1} {{\zeta _{m,l}}} }}} \right),m \in [1,M],k \in \left[ {1,K} \right],i \in \left[ {k,K} \right].}
\end{array}
\end{equation}
\hrulefill
%\vspace*{4pt}
\end{figure*}
By means of KKT conditions, the optimal solution to (\ref{eqn:gp_max_rate}) is obtained in the following theorem.
\begin{theorem}\label{the:rate_opt}
  The optimal transmission rate is given by
\begin{equation}\label{eqn:R_opt_theta}
R_{m,k}^* = {\log _2}\left( {1 + \frac{{{\zeta _{m,k}}}}{{\theta _{m,k}^* + \sum\nolimits_{l = 1}^{k - 1} {{\zeta _{m,l}}} }}} \right),
\end{equation}
where $\theta _{m,k}^*$ is the zero point of $\varrho(x) - \ell(x) = 0$, and
\begin{equation}\label{eqn:varrho_def}
\varrho(x) = \frac{{{\zeta _{m,k}}x ^{{N_r} - M + 2}\left( {1 - \phi _{m,k}x ^{ - \left( {{N_r} - M + 1} \right)}} \right)}}{{\left( {{N_r} - M + 1} \right)\phi _{m,k}{{\ln }}\left( {1 + \frac{{{\zeta _{m,k}}}}{{x + \sum\nolimits_{l = 1}^{k - 1} {{\zeta _{m,l}}} }}} \right)}},
\end{equation}
\begin{equation}\label{eqn:ell_def}
\ell(x) = \left( {x + \sum\nolimits_{l = 1}^{k - 1} {{\zeta _{m,l}}} } \right)\left( {x + \sum\nolimits_{l = 1}^k {{\zeta _{m,l}}} } \right).
\end{equation}
Besides, the optimal rate satisfies the following relationship
\begin{equation}\label{eqn:normal_rate}
\frac{{\zeta _{m,1}}}{{{2^{{R_{m,1}^*}}} - 1}} < \frac{{\zeta _{m,2}}}{{{2^{{R_{m,2}^*}}} - 1}} <  \cdots  < \frac{{\zeta _{m,K}}}{{{2^{{R_{m,K}^*}}} - 1}},
\end{equation}
which is in perfect accordance with the analysis of the DMT in Section \ref{sec:dmt}.
\end{theorem}
\begin{proof}
  Please see Appendix \ref{app:rate_opt}.
\end{proof}
\subsection{Joint Power Allocation and Rate Selection}\label{sec:opt_joint}
At last, we consider the maximization of goodput by jointly optimize power allocation coefficients and transmission rates. Unfortunately, it is virtually impossible to derive closed-form solution for this joint optimization problem. %the closed-form/numerical solution is found for no matter optimal power allocation problem or optimal rate selection problem, the closed-form solution cannot be derived, Different from Sections \ref{sec:optpowe} and \ref{sec:opt_rate}
It has been elaborated in Section \ref{sec:optpowe} that the non-linear fractional form of the outage probability in \eqref{eqn:out_tar} hinders further optimization of transmission rate. %Although we are unable to get the optimal solution
Nevertheless, the power allocation coefficients and transmission rates can be sub-optimally designed by using alternately iterating optimization, which is based on both the results in Sections \ref{sec:optpowe} and \ref{sec:opt_rate} \cite{bezdek2002some}. More specifically, by fixing the transmission rates, the power allocation coefficients are optimized with \eqref{eqn:zeta_fina}. Then by fixing the optimal power allocation, the optimal transmission rate can be obtained with \eqref{eqn:R_opt_theta}. We successively conduct the above two steps, until a negligible difference between the optimal objective values of two consecutive iterations occurs. The pseudocode of the joint power and rate optimization is provided in Algorithm \ref{alg:rs}. The local convergence of the proposed suboptimal algorithm can be elucidated as follows. For ease of exposition, we denote by ${T_g}(\left(R_{m,k}\right)_{M\times K}, \left({{{\zeta_{m,k}}}}\right)_{M\times K})$ the goodput corresponding to the fixed transmissions rates $\left(R_{m,k}\right)_{M\times K}$ and powers $\left({{{\zeta_{m,k}}}}\right)_{M\times K}$. According to Algorithm \ref{alg:rs}, the optimal values of the goodput computed by any two consecutive iterative steps are in an increasing order as
\begin{align}\label{eqn:alg_conv}
  &{T_g}\left(\left(R_{m,k}^{(\imath)}\right)_{M\times K}, \left({{{\zeta_{m,k}^{(\imath)}}}}\right)_{M\times K}\right)\notag\\
   & \le {T_g}\left(\left(R_{m,k}^{(\imath)}\right)_{M\times K}, \left({{{\zeta_{m,k}^{(\imath+1)}}}}\right)_{M\times K}\right)\notag\\
   & \le {T_g}\left(\left(R_{m,k}^{(\imath+1)}\right)_{M\times K}, \left({{{\zeta_{m,k}^{(\imath+1)}}}}\right)_{M\times K}\right),\, \imath\ge 0,
\end{align}
where $\imath$ denotes the iteration number, and the inequalities hold thanks to the global optimality of both the power allocation and rate selection algorithms developed in Section \ref{sec:optpowe} and \ref{sec:opt_rate}, respectively. Moreover, by putting \eqref{eqn:out_opt_lower} into \eqref{eqn:goodput_def}, $T_g$ is upper bounded, because
\begin{align}\label{eqn:tg_upper}
&{T_g} \le \notag\\
&\sum\limits_{m = 1}^M {\sum\limits_{k = 1}^K {\left( {1 - {M^{{N_r} - M + 1}}{\phi _{m,k}}{{\left( {{2^{{R_{m,k}}}} - 1} \right)}^{{N_r} - M + 1}}} \right){R_{m,k}}} },
\end{align}
where the right-hand side of \eqref{eqn:tg_upper} is obviously upper bounded owing to its concavity with respect to $({R_{m,k}})_{M\times K}$. Accordingly, the monotonic and bounded sequence of $\{{T_g}((R_{m,k}^{(\imath)})_{M\times K}, ({{{\zeta_{m,k}^{(\imath)}}}})_{M\times K})\}$ is convergent. This consequently substantiates the convergence of Algorithm \ref{alg:rs}.
\begin{algorithm}[h]
   \caption{Suboptimal Joint Power Allocation and Rate Selection Algorithm}\label{alg:rs}
    \begin{algorithmic}[1]
        \State Generate the initial rates $\left(R_{m,k}^{(0)}\right)_{M\times K}$
        \State Set $\imath = 0$ and ${T_{g}^{*}}^{(\imath)}=-1$
        \Repeat
        \State $\imath = \imath+1$
        %\For{$j = 1$ to $K$}
            \State Given $\left(R_{m,k}^{(\imath-1)}\right)_{M\times K}$, update $\left({{{\zeta_{m,k} ^{(\imath)}}}}\right)_{M\times K}$ with (\ref{eqn:zeta_fina})
            \State Given $\left({{{\zeta_{m,k} ^{(\imath)}}}}\right)_{M\times K}$, update $\left(R_{m,k}^{(\imath)}\right)_{M\times K}$ with (\ref{eqn:R_opt_theta})
        %\EndFor
        \State ${T_{g}^{*}}^{(\imath)} = {T_g}((R_{m,k}^{(\imath)})_{M\times K}, ({{{\zeta_{m,k}^{(\imath)}}}})_{M\times K})$
        \Until{satisfying termination criterion ${T_{g}^{*}}^{(\imath)} - {T_{g}^{*}}^{(\imath-1)} < \epsilon$}
        \State $\left(\zeta_{m,k}^{*}\right)_{M\times K} = \left(\zeta_{m,k}^{(\imath)}\right)_{M\times K}$
        \State $\left(R_{m,k}^{*}\right)_{M\times K} = \left(R_{m,k}^{(\imath+1)}\right)_{M\times K}$
        %\State \textcolor[rgb]{1.00,0.00,0.00}{${T_{g}^{*}} = {T_{g}^{*}}^{(\imath)}$}
        \end{algorithmic}
\end{algorithm}

%the optimal solution is determined by
%\eqref{eqn:zeta_fina} and
%This is due to the fact that the expression of the goodput is very complex by putting (\ref{eqn:out_tar}) into (\ref{eqn:goodput_def}), also makes it fairly intractable to numerically solve the problem.

From both (\ref{eqn:normal_power}) and (\ref{eqn:normal_rate}), the proposed sub-optimal algorithm would prefer to allocate more transmission power to clusters under bad channel conditions and support them with higher transmission rate under good channel conditions.

\section{Numerical Results and Discussions}\label{sec:num}
In this section, numerical results are presented for verifications and discussions. For illustration, the parameters of the virtual MIMO-NOMA system are set as follows unless otherwise specified. %The path loss exponent and the reference received power are set as $\alpha=3$ and $\mathcal K=1$, respectively.
All the data streams are assumed to have the same transmission rate $R_{m,i}$, i.e., $R_{m,i}=R$. According to the prerequisite stated in \eqref{eqn:out_pro_def_rew}, the power allocation coefficients can be devised as \cite{shi2018performance} to implement NOMA such that
\begin{align}\label{eqn:power_allo_defaul_sim}
&{\zeta _{m,k}} = \left( {\frac{1}{M} - \sum\limits_{l = k + 1}^K {{\zeta _{m,l}}} } \right)\left( {1 - \varepsilon {2^{ - R}}} \right), k = 2,3,...,K,\notag\\
&{\zeta _{m,1}} = \frac{1}{M} - \sum\limits_{l = 2}^K {{\zeta _{m,l}}},
\end{align}
where $\varepsilon$ lies within $\left[ {0,1} \right]$. The transmit beamforming matrix $\bf V$ is set with ones on its principal diagonal and zeros elsewhere. The antenna correlation is modelled by using the exponential correlation structure as ${\bf R}_t = \left( [\rho^{|i-j|}]_{1 \le i,j\le N_t}\right)$ and ${\bf R}_r = \left( [\rho^{|i-j|}]_{1 \le i,j\le N_r}\right)$ \cite{li2017performance}. The system parameters $N_t$, $N_r$, $M$, $\varepsilon$, $\alpha$, $\mathcal K$, $\rho$, $\bar \gamma$ and $R$ are assumed to be $3$, $3$, $3$, $0.7$, 3, 1, $0.5$, $70$dB and $2$bps/Hz, respectively. We assume that the base station is placed at the centre of the cell and $K=4$ clusters are served. The cluster centers are located at positions $(10{\rm m},0)$, $(0,20{\rm m})$, $(0,-30{\rm m})$ and $(-40{\rm m},0)$. %The AWGN is assumed to have a unity variance, i.e., $\sigma^2=1$.
\subsection{Numerical Verifications}
Fig. \ref{fig:out_ver} depicts the outage probability $p_{1,1}^{out}$ against the transmit SNR $\bar \gamma$. As expected from Fig. \ref{fig:out_ver}, the Monte Carlo simulation results confirm the validity of the analytical results obtained from \eqref{eqn:out_1_xdeffina}, and there is a perfect match between the analytical and asymptotic results in high SNR regime. Besides, it is also shown in Fig. \ref{fig:out_ver} that the decreasing slope of the outage curve increases with the number of data streams $M$, but is independent of the path loss exponent $\alpha$. This is consistent with the analysis of the diversity order, i.e., $\mathpzc d=N_r-M+1$. Whereas, since $\alpha$ characterizes the severity of path loss, the decrease of $\alpha$ is favorable for decreasing the outage probability. Moreover, the outage probability increases with $M$ given a fixed value of $\bar \gamma$, because the total power constraint results in the reduction of the transmit power allocated to each data stream as $M$ increases.
\begin{figure}
  \centering
  \includegraphics[width=3in]{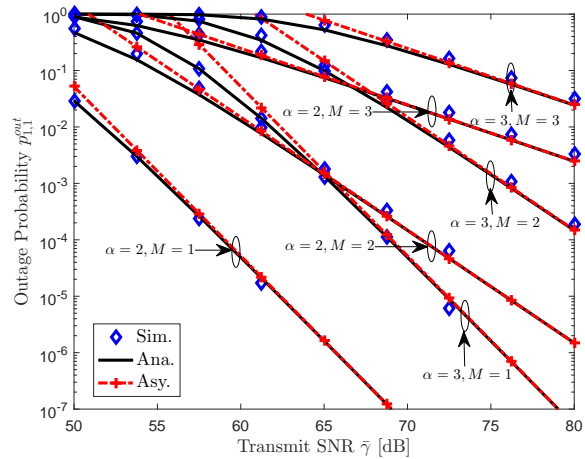}
  \caption{The outage probability $p_{1,1}^{out}$ versus the transmit SNR $\bar \gamma$.}\label{fig:out_ver}
\end{figure}

Fig. \ref{fig:gp_ver} plots the goodput $T_g$ versus the transmit SNR $\bar \gamma$. It is observed that the theoretical results coincide well with the simulation ones. It is clear from Fig. \ref{fig:gp_ver} that less severe fading, i.e., small $\alpha$, improves the goodput. However, unlike Fig. \ref{fig:out_ver}, the virtual MIMO-NOMA system may benefit from the increase of $M$ in terms of the goodput, albeit degrading the outage performance. Furthermore, it can be seen from Fig. \ref{fig:gp_ver} that the goodput is upper bounded as SNR increases because of the prescribed target transmission rates $R_{m,k}$. Hence, the observations imply the necessity of optimally designing the transmission rate to maximize the goodput, which will be illustrated in Section \ref{sec:gp_num}.
\begin{figure}
  \centering
  \includegraphics[width=3in]{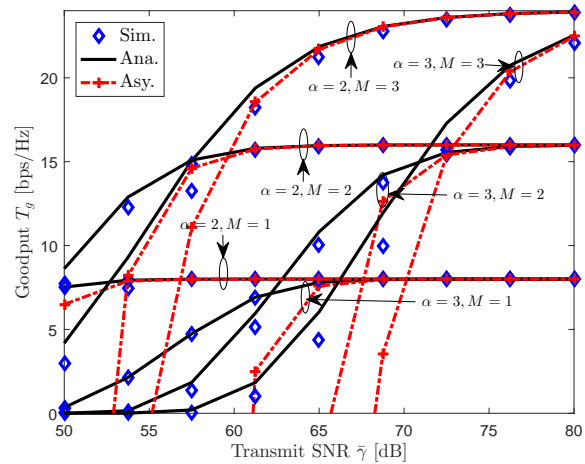}
  \caption{The goodput $T_g$ versus the transmit SNR $\bar \gamma$.}\label{fig:gp_ver}
\end{figure}

\subsection{Impact of Spatial Correlation}
The effects of the spatial correlation on the outage probability and goodput are investigated in Figs. \ref{fig:corr_out} and \ref{fig:corr_gp}, respectively. From both figures, the theoretical and simulation results are in perfect agreement under low spatial correlation due to the assumption made in \eqref{eqn:cdf_F_mean_re1}. It can also be observed from Figs. \ref{fig:corr_out} and \ref{fig:corr_gp} that the spatial correlation adversely impacts the outage probability and goodput. In particular, the spatial correlation would severely impair the performance of the virtual MIMO-NOMA system if $\rho>0.5$, which conforms to the conclusion drawn in \cite[Fig. 3.5]{goldsmith2005wireless}.
\begin{figure}
  \centering
  \includegraphics[width=3in]{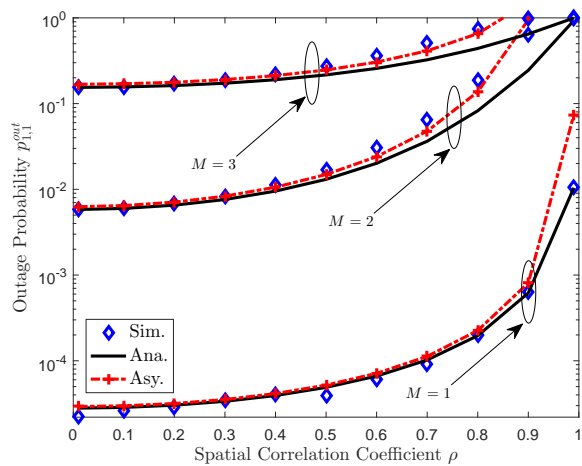}
  \caption{The impact of the spatial correlation on the outage probability $p_{1,1}^{out}$.}\label{fig:corr_out}
\end{figure}
\begin{figure}
  \centering
  \includegraphics[width=3in]{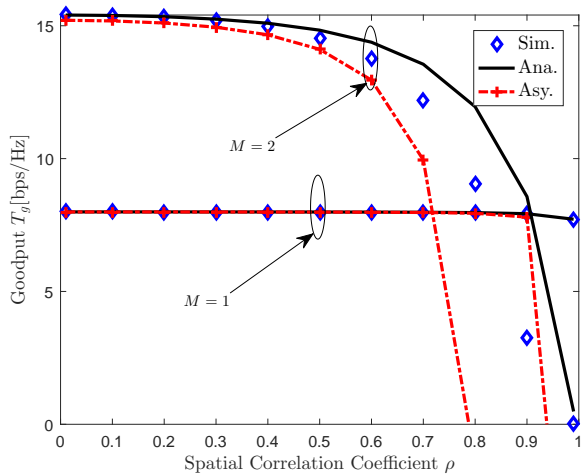}
  \caption{The impact of the spatial correlation on the goodput $T_g$.}\label{fig:corr_gp}
\end{figure}

\subsection{Comparison with Five Baseline Schemes}\label{sec:comp_num}
To illustrate the advantages of the proposed scheme, we choose five reference schemes for benchmarking, including the NOMA scheme (labelled as ``NOMA'') \cite{saito2013system}, the conventional OMA scheme (labelled as ``OMA'') that serves as a baseline in \cite{ding2016application}, the virtual MIMO scheme (labelled as ``Virtual MIMO'') \cite{dohler2002space}, the hybrid MA scheme (labelled as ``Hybrid MA'') in \cite{liu2016cooperative}, the virtual MIMO-hybrid MA scheme (labelled as ``Virtual MIMO-Hybrid MA''). Herein, the virtual MIMO-hybrid MA scheme is a modified version of the proposed scheme, which groups virtual MIMO entities into pairs to conduct NOMA and then accommodates different pairs with OMA. More specifically, we assume for the NOMA, OMA and hybrid MA schemes that each IoT device requests one data stream from the BS. With regard to both the proposed and virtual MIMO-based schemes, this amounts to the case of three data streams requested by each cluster, i.e., $M=3$. For fair comparison, time-division multiple access (TDMA) is considered for the OMA, virtual MIMO, hybrid MA and virtual MIMO-hybrid MA schemes to serve all the IoT devices/clusters/pairs with equal-length time slots. In addition, for both the NOMA and hybrid MA schemes, each IoT device can be treated as a special cluster with only a single IoT device. It is worthwhile to note that our analytical results also apply to the virtual MIMO-hybrid MA scheme as long as user pairing is determined. This is because each NOMA pair can be tackled independently. For illustration, an extreme pairing policy in \cite{ding2016application} is considered for both hybrid MA-based schemes, where the device/virtual MIMO entity having the best channel condition is paired with the one having the worst channel condition, and the rest of devices/entities are paired in the same manner.

Fig. \ref{fig:comp} compares the goodputs of the proposed scheme and five baseline schemes. It is obviously observed that the proposed scheme performs much better than the five baseline schemes in high SNR regime. However, as opposed to both the OMA and virtual MIMO schemes, the superior performance of the proposed scheme cannot be perfectly exhibited at low SNR. This is due to the fact that the significant advantages of NOMA are based on the premise of high SNR\cite{ding2015impact}. Whereas, this result no longer holds at low SNR. It is worth noting that the condition of high SNR is commonly satisfied to achieve a very low outage (e.g., $10^{-5}$) in 5G systems \cite{bennis2018ultrareliable}. For instance, by considering a carrier bandwidth of 100MHz, thermal noise level of -174dBm/Hz and maximum output power budget of 24dBm, the transmit SNR is calculated as $118$dB \cite{dahlman20185g}. It is also observed from Fig. \ref{fig:comp} that the proposed scheme generally outperforms the NOMA scheme in terms of the goodput. For example, by aiming at $T_g=6$bps/Hz, the required transmit SNRs for the proposed scheme and NOMA scheme must be set as 65dB and 77dB, respectively. Henceforth, the introduction of the virtual MIMO yields a reduction of the transmit SNR by 12dB, which remarkably conserves energy with guaranteed QoS for IoT networks. %NOMA that benefits from the difference among fading channels is based on the premise of high SNR.
%This is due to the fact that the benefit of NOMA that originates from the difference among fading channels is based on the premise of high SNR. .
%, which consequently degrades the system performance. The similar results between the hybrid MA and virtual MIMO-hybrid MA schemes can be observed.
In addition, it is shown that the proposed scheme is a little bit inferior to the NOMA scheme under low SNR because of the negative effect of antenna correlations. The similar observations can be found coincidentally between the OMA and virtual MIMO schemes. Besides, it can be seen that both hybrid MA-based schemes achieve a higher goodput than the proposed scheme at low SNR. Nevertheless, the maximum achievable goodput of the hybrid MA scheme in \cite{liu2016cooperative} is tremendously restricted. In contrast, the virtual MIMO-hybrid MA scheme is able to strike a balance between boosting spectral efficiency and reducing system overhead for resource-limited IoT networks.

 %To exemplify this, we consider an extreme clustering policy in \cite{ding2016application}, where the user having the best channel condition is paired with the one having the worst channel condition, and the rest of users are paired in the same manner. By using time division MA to serve different pairs, the virtual MIMO-hybrid MA system thus consists of two independent virtual MIMO-NOMA pairs. Accordingly, by setting $K=2$ in this paper, our analytical results are applicable to the performance evaluation and optimal design of the virtual MIMO-hybrid MA system. For instance, Fig. 6 (as copied in Fig. \ref{fig:comp} in the response to the reviewer's next comment) has been redrawn by incorporating the goodput curve of the virtual MIMO-hybrid MA scheme for comparison. It is observed that the virtual MIMO-hybrid MA scheme achieves a superior performance over the virtual MIMO-NOMA scheme at low SNR in terms of the goodput, while the relationship between their goodputs reverses at high SNR. Consequently, the virtual MIMO-hybrid MA scheme turns out to be very appealing for energy-constrained IoT networks.

\begin{figure}
  \centering
  \includegraphics[width=3in]{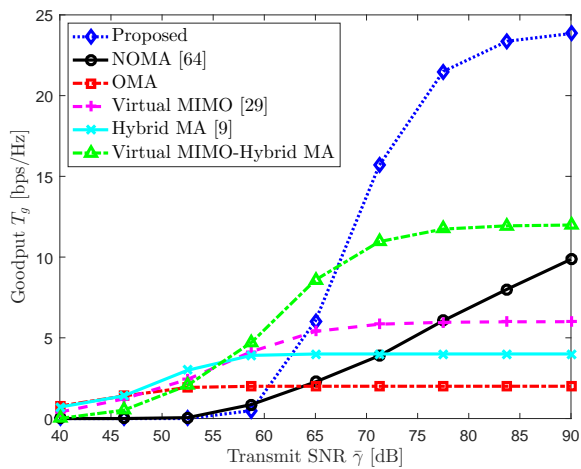}
  \caption{The comparison of the goodput of the proposed scheme and five baseline schemes with $N_t=N_r=M=3$, $\varepsilon=0.7$, $\rho=0.5$ and $R=2$bps/Hz.}\label{fig:comp}
\end{figure}

\subsection{Goodput Maximization}\label{sec:gp_num}
In this section, Figs. \ref{fig:opt_gp_power} and \ref{fig:opt_gp_rate} are plotted to demonstrate the necessity of the goodput maximization through power allocation and/or rate selection. For illustrative purposes, we assume that the base station only serves two clusters with coordinates $(10{\rm m},0)$ and $(0,20{\rm m})$, and a single data stream is delivered to each cluster, i.e., $M=1$. Accordingly, if the values of the transmission rates $R_{1,1}$ and $R_{1,2}$ are fixed, the power allocation coefficients $\zeta_{1,1}$ and $\zeta_{1,2}$ can be optimally chosen as studied in Section \ref{sec:optpowe}. With the optimal power allocation, a three-dimensional figure is displayed in Fig. \ref{fig:opt_gp_power} to describe the relationship between the maximum goodput and the transmission rates. It is clearly observed that the transmission rates drastically influences the maximum achievable goodput. Moreover, in order to reach the maximum goodput, Fig. \ref{fig:opt_gp_power} exhibits that the near cluster is inclined to have a higher transmission rate than the distant cluster. It therefore justifies the result in Section \ref{sec:opt_joint} that the proposed joint optimization tends to provide the cluster under better channel condition with a higher data rate. Whereas, the difference between the optimal transmission rates diminishes with the increase of the transmit SNR.
\begin{figure}
  \centering
  \includegraphics[width=3in]{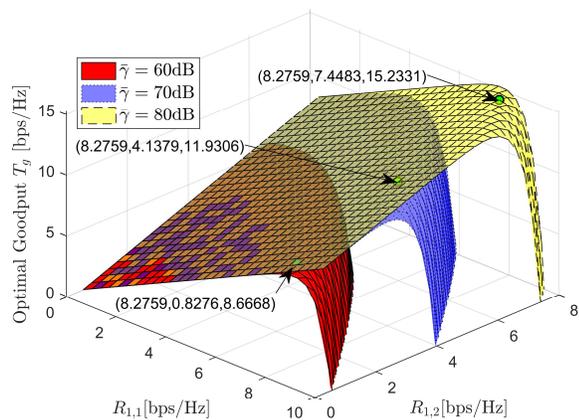}
  \caption{The optimal goodput versus the transmission rates.}\label{fig:opt_gp_power}
\end{figure}

% in Section \ref{sec:opt_rate}
%Furthermore, given either the power allocation coefficients, % $\zeta_{1,1}$ and $\zeta_{1,2}$ ($\zeta_{1,1}+\zeta_{1,2}=1$),
%the optimal transmission rate can be determined by Theorem \ref{the:rate_opt}.
With the optimal rate selection, Fig. \ref{fig:opt_gp_rate} plots the maximum goodput versus $\zeta_{1,1}$, where $\zeta_{1,2}=1-\zeta_{1,1}$. It can seen from Fig. \ref{fig:opt_gp_rate} that the power allocation coefficients should be properly selected to attain the maximum achievable goodput. In particular, if $\zeta_{1,1}=0$ and $\zeta_{1,2}=1$, NOMA scheme degenerates to its counterpart, i.e., OMA. Apparently, neither $\zeta_{1,1}=0$ nor $\zeta_{1,2}=1$ will yield the maximum goodput. Conversely, the optimal power allocation coefficient often lies in between $[0,1]$. In addition, it is shown that more transmission power is allocated to the farther user. It further substantiates the result in Section \ref{sec:opt_joint} that the joint power and rate optimization prefers to assign more transmission power to the cluster under worse channel condition. The difference between the optimal values of the power allocation coefficients becomes conspicuous in high SNR regime. This is because more transmission power is demanded to support the increasing data rate for the cluster under poor channel condition at high SNR, as elucidated in Fig. \ref{fig:opt_gp_power}.
\begin{figure}
  \centering
  \includegraphics[width=3in]{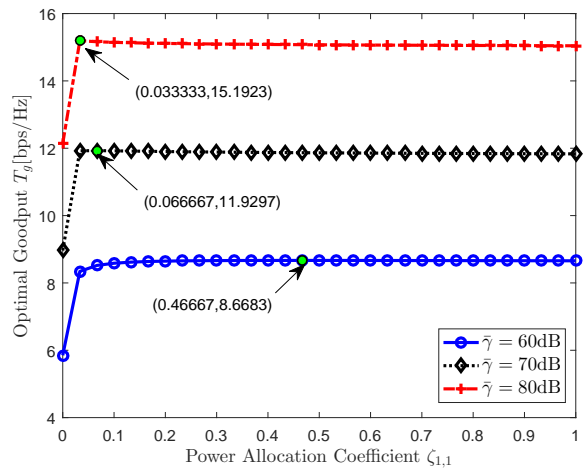}
  \caption{The optimal goodput versus the power allocation coefficient.}\label{fig:opt_gp_rate}
\end{figure}

Fig. \ref{fig:opt_gp_joint} testifies the convergence behavior of the suboptimal joint power allocation and rate selection algorithm, namely Algorithm \ref{alg:rs}. As seen from Fig. \ref{fig:opt_gp_joint}, the monotonic convergence of Algorithm \ref{alg:rs} validates the correctness of the convergence analysis in Section \ref{sec:opt_joint}. Moreover, it is shown that the goodput performance almost converges to the upper bound within 15 iterations, which demonstrates the effectiveness of Algorithm \ref{alg:rs}. In addition, it is as expected that the optimal goodput grows with the transmit SNR $\bar \gamma$.
%To show the advantage of the suboptimal joint power allocation and rate selection algorithm (labeled as ``Joint Opt.''), the three proposed algorithms for the goodput maximization are compared in Fig. \ref{fig:opt_gp_joint}. The label ``Opt. Power'' corresponds to the result of the optimal power allocation algorithm with $R_{m,k}=R=2$bps/Hz. On the other hand, the result of the optimal rate selection algorithm is labeled as ``Opt. Rate''. For a fair comparison with ``Opt. Power'', the power allocation coefficients with regard to ``Opt. Rate'' are set as \eqref{eqn:power_allo_defaul_sim}. Clearly from Fig. \ref{fig:opt_gp_joint}, the proposed joint power and rate optimization algorithm offers a significant performance gain over the other two algorithms. Besides, it is as expected that the optimal rate selection is superior to the optimal power allocation especially under high SNR, because the goodput is upper bounded by assuming fixed transmission rates in ``Opt. Power''.

%The iteration stops until the difference between two adjacent optimal goodput is less than $10^{-5}$.
\begin{figure}
  \centering
  \includegraphics[width=3in]{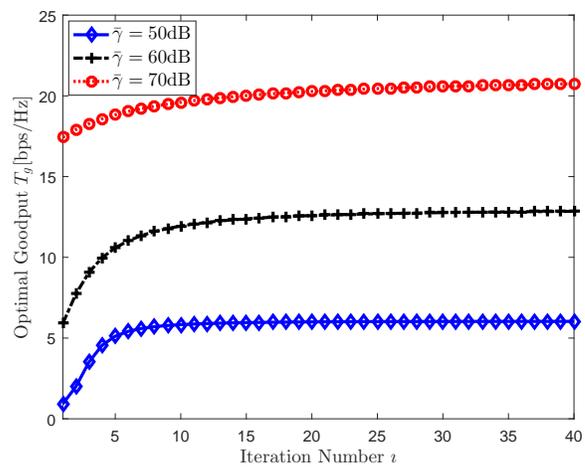}
  \caption{The convergence behavior of Algorithm \ref{alg:rs} with the maximum allowable stopping criteria $\epsilon=10^{-5}$.}\label{fig:opt_gp_joint}
\end{figure}

\section{Conclusion}

This paper has proposed a downlink virtual MIMO-NOMA scheme in IoT networks. %All the single-antenna IoT devices in each cluster cooperate with each other to form a virtual MIMO entity, and multiple independent data streams are requested by each cluster. NOMA technique is employed at the base station to superimpose all the requested data streams in power-domain, and each cluster takes advantage of zero-forcing detection to de-multiplex the input data streams.
We have assumed only statistical CSI available at the base station to avoid frequent instantaneous SINR reporting and signaling overhead so as to save the energy and bandwidth for IoT networks. The outage probability and goodput of the virtual MIMO-NOMA system have been thoroughly examined by using Kronecker correlation model. The asymptotic results not only have provided physical insights but also have paved a way for the goodput maximization. In particular, the asymptotic outage expressions quantify the impacts of various system parameters and make it possible to examine the DMT of MIMO-NOMA channels. Additionally, the power allocation coefficients and/or transmission rates have been properly chosen to reach the maximum goodput. By virtue of KKT conditions, the optimal solutions to the goodput maximization problems have been derived in closed-form. %The numerical analysis has revealed that the joint power and rate optimization algorithm achieves a noticeable goodput gain over the other two algorithms relying solely on power or rate optimization.
Besides, it has been found that the proposed algorithms tend to allocate more power to clusters under unfavorable channel conditions and support clusters with higher transmission rate under benign channel conditions.

\appendices
\section{Proof of \eqref{eqn:out_1_xdeffina}}\label{app:F_xk}
In what follows we specialize to the derivation of $p_{m=1,k}^{out}$ first. The outage probability for arbitrary data stream $m$ can be got in an analogous way. By partitioning ${\bf{Z}}$ as ${\bf{Z}} = [{{{\bf{z}}_1}},{{\bf{\tilde Z}}} ]$ along with matrix inversion in block form in \cite[Exercise 5.16]{abadir2005matrix}, $p_{m{{ = }}1,k}^{out}$ can be further expressed as %As in \cite{kiessling2003analytical}
\begin{align}\label{eqn:pout_1_part}
&p_{m=1,k}^{out} = \notag\\
&\quad \Pr \left[ {{{\bf{z}}_1}^{\rm{H}}\left( {{\bf{I}} - {\bf{\tilde Z}}{{\left( {{{{\bf{\tilde Z}}}^{\rm{H}}}{\bf{\tilde Z}}} \right)}^{ - 1}}{{{\bf{\tilde Z}}}^{\rm{H}}}} \right){{\bf{z}}_1} < \frac{1}{{{{\bar \gamma }}{\theta _{1,k}}\ell ({d_k})}}} \right].
\end{align}
%According to the partition, the conditional distribution of ${\bf{z}}_1$ given ${{\bf{\tilde Z}}}$ follows as ${\bf{z}}_1|{{\bf{\tilde Z}}} \sim {\cal C N}({\bf 0},  {\bf R}_{{{r} }}/{\left[{{\bf{R}}_{t'}}^{-1}\right]_{1,1}})$.
In analogous to \cite[eq. 15]{kiessling2003analytical}, $p_{m{{ = }}1,k}^{out}$ can be derived as (\ref{eqn:out_1_cond}), shown at the top of the next page,
\begin{figure*}[!t]
\begin{equation}\label{eqn:out_1_cond}
p_{m{{ = }}1,k}^{out} = \Pr \left[ {{{\bf{u}}^{\rm{H}}}{{\bf{R}}_r}^{1/2}\underbrace {\left( {{\bf{I}} - {{\bf{R}}_r}^{1/2}{{{\bf{\tilde Z}}}_w}{{\left( {{{{\bf{\tilde Z}}}_w}^{\rm{H}}{{\bf{R}}_r}{{{\bf{\tilde Z}}}_w}} \right)}^{ - 1}}{{{\bf{\tilde Z}}}_w}^{\rm{H}}{{\bf{R}}_r}^{1/2}} \right)}_{\bf{Q}}{{\bf{R}}_r}^{1/2}{\bf{u}} < \frac{{{{\left[ {{{\bf{R}}_{t'}}^{ - 1}} \right]}_{11}}}}{{{{\bar \gamma }}{\theta _{1,k}}\ell ({d_k})}}} \right],
\end{equation}
%\hrulefill
%\vspace*{4pt}
\end{figure*}
where ${{{\bf{\tilde Z}}}_w}\in {\mathbb C}^{N_r\times (M-1)}$ and ${\bf u}\in {\mathbb C}^{N_r\times 1}$ have i.i.d. zero-mean complex Gaussian elements with unity variance. By defining ${X_k} = {{\bf{u}}^{\rm{H}}}{{\bf{R}}_r}^{1/2}{\bf{Q}}{{\bf{R}}_r}^{1/2}{\bf{u}}$, it follows that
\begin{equation}\label{eqn:p_fxk}
p_{m{=}1,k}^{out} = F_{{{X_k}}} \left( {{{{\left[ {{{\bf{R}}_{t'}}^{ - 1}} \right]}_{11}}}}/{({{{\bar \gamma }}{\theta _{1,k}}\ell ({d_k})})} \right),
\end{equation}
where $F_{{{X_k}}} \left( x \right)$ is the cumulative distribution function (CDF) of $X_k$. Conditioned on ${{\bf{R}}_r}^{1/2}{\bf{Q}}{{\bf{R}}_r}^{1/2}$, we arrive at
\begin{align}\label{eqn:out_1_xdef}
&F_{{{X_k}}} \left( x \right) %={\Pr \left[ {{X_k} < \frac{{{{\left[ {{{\bf{R}}_{t'}}^{ - 1}} \right]}_{11}}}}{{{{\bar \gamma }}{\theta _{1,k}}\ell ({d_k})}}} \right]}\notag\\
= {{\mathbb{E}}_{{{{\bf{\tilde Z}}}_w}}}\left\{ {\Pr \left[ {\left. {{X_k} < x} \right|{{\bf{R}}_r}^{1/2}{\bf{Q}}{{\bf{R}}_r}^{1/2}} \right]} \right\} \notag\\
&= {{\mathbb{E}}_{{{{\bf{\tilde Z}}}_w}}}\left\{ {{F_{\left. {{X_k}} \right|{{\bf{R}}_r}^{1/2}{\bf{Q}}{{\bf{R}}_r}^{1/2}}}\left( {\left. x \right|{{\bf{R}}_r}^{1/2}{\bf{Q}}{{\bf{R}}_r}^{1/2}} \right)} \right\},
\end{align}
where ${F_{\left. {{X_k}} \right|{{\bf{R}}_r}^{1/2}{\bf{Q}}{{\bf{R}}_r}^{1/2}}}\left( x \right)$ denotes the conditional CDF of $X_k$ given ${{\bf{R}}_r}^{1/2}{\bf{Q}}{{\bf{R}}_r}^{1/2}$. To proceed, we recourse to the method of MGF to derive ${F_{\left. {{X_k}} \right|{{\bf{R}}_r}^{1/2}{\bf{Q}}{{\bf{R}}_r}^{1/2}}}\left( x \right)$.

Given ${{{\bf{R}}_r}^{1/2}{\bf{Q}}{{\bf{R}}_r}^{1/2}}$, $X_k$ is a Hermitian quadratic form in complex normal variables, and the MGF of $X_k$ conditioned on ${{{\bf{R}}_r}^{1/2}{\bf{Q}}{{\bf{R}}_r}^{1/2}}$ can be obtained by using Turin's result from \cite{turin1960characteristic} as
\begin{align}\label{eqn:cond_mgf}
{\mathcal M_{\left. {{X_k}} \right|{{\bf{R}}_r}^{1/2}{\bf{Q}}{{\bf{R}}_r}^{1/2}}}\left( s \right) = {{\det \left( {{\bf{I}} - s{{\bf{R}}_r}^{1/2}{\bf{Q}}{{\bf{R}}_r}^{1/2}} \right)}}^{-1}.
\end{align}
Then, the CDF ${F_{\left. {{X_k}} \right|{{\bf{R}}_r}^{1/2}{\bf{Q}}{{\bf{R}}_r}^{1/2}}}\left( y \right)$ can be derived by invoking inverse Laplace transform as\cite[eq.3.6.8]{debnath2010integral}
\begin{multline}\label{eqn:x_k_cdf_cond_inverse}
{F_{\left. {{X_k}} \right|{{\bf{R}}_r}^{1/2}{\bf{Q}}{{\bf{R}}_r}^{1/2}}}\left( x \right) \\
= {{\cal L}^{ - 1}}\left\{ \frac{{{\mathcal M_{\left. {{X_k}} \right|{{\bf{R}}_r}^{1/2}{\bf{Q}}{{\bf{R}}_r}^{1/2}}}}\left( { - s}\right)}{s}  \right\}\left( x \right).
\end{multline}
By substituting (\ref{eqn:cond_mgf}) into (\ref{eqn:x_k_cdf_cond_inverse}), it follows that
\begin{align}\label{eqn:cond_pdf_x_k_inv}
&{F_{\left. {{X_k}} \right|{{\bf{R}}_r}^{1/2}{\bf{Q}}{{\bf{R}}_r}^{1/2}}}\left( x \right) \notag\\
&\quad= \frac{1}{{2\pi {\rm i}}}\int\nolimits_{a - {\rm i}\infty }^{a + {\rm i}\infty } {\frac{{{e^{sx}}}}{{s\det \left( {{\bf{I}} + s{{\bf{R}}_r}^{1/2}{\bf{Q}}{{\bf{R}}_r}^{1/2}} \right)}}ds} \notag\\
 &\quad= \frac{1}{{2\pi {\rm i}}}\int\nolimits_{a - {\rm i}\infty }^{a + {\rm i}\infty } {\frac{{{e^{sx}}}}{s}{}_1{F_0}\left( {1; - s{{\bf{R}}_r}^{1/2}{\bf{Q}}{{\bf{R}}_r}^{1/2}} \right)ds},
\end{align}
where ${\rm i}=\sqrt{-1}$, the second step holds by using the identity $\det {\left( {{\bf{I}} - {\bf{X}}} \right)^{ - a}} = {}_1{F_0}\left( {a;{\bf{X}}} \right)$ from \cite{ratnarajah2003complex}, and ${}_p{F_q}\left( {{\bf{X}}} \right)$ denotes the hypergeometric function of one Hermitian matrix \cite{muirhead2009aspects}. Hence, the CDF of $X_k$, i.e., $F_{{{X_k}}} \left( x \right)$, can be obtained by taking expectation of (\ref{eqn:cond_pdf_x_k_inv}) over the distribution of ${{{\bf{\tilde Z}}}_w}$ as
\begin{align}\label{eqn:cdf_F_mean}
&F_{{{X_k}}} \left( x \right)= {{\mathbb{E}}_{{{{\bf{\tilde Z}}}_w}}}\left\{ {{F_{\left. {{X_k}} \right|{{\bf{R}}_r}^{1/2}{\bf{Q}}{{\bf{R}}_r}^{1/2}}}\left( {\left. x \right|{{\bf{R}}_r}^{1/2}{\bf{Q}}{{\bf{R}}_r}^{1/2}} \right)} \right\}\notag\\
&= \frac{1}{{2\pi {\rm i}}}\int\nolimits_{a - {\rm i}\infty }^{a + {\rm i}\infty } {\frac{{{e^{sx}}}}{s}{{\rm{E}}_{{{{\bf{\tilde Z}}}_w}}}\left\{ {{}_1{F_0}\left( {1; - s{{\bf{R}}_r}^{1/2}{\bf{Q}}{{\bf{R}}_r}^{1/2}} \right)} \right\}ds}\notag\\
&= \frac{1}{{2\pi {\rm i}}}\int\nolimits_{a - {\rm i}\infty }^{a + {\rm i}\infty } {\frac{{{e^{sx}}}}{s}{{\rm{E}}_{{{{\bf{\tilde Z}}}_w}}}\left\{ {{}_1{F_0}\left( {1; - s{{\bf{R}}_r}{\bf{Q}}} \right)} \right\}ds}.
\end{align}
where the second step holds by exchanging the operations of the integration and the expectation, and the last step holds by using ${}_p{F_q}\left( {{\bf{AB}}} \right)={}_p{F_q}\left( {{\bf{A}}^{1/2}{\bf{B}}{\bf{A}}^{1/2}} \right)$ for any positive definite matrix $\bf A$ and symmetric matrix $\bf B$ \cite[p237, eq. 17]{muirhead2009aspects}. Noticing that $\bf Q$ is an idempotent matrix  with $N_r-M+1$ eigenvalues of value 1 and all other eigenvalues of value 0. More specifically, by defining ${\bf Y} = {{\bf{R}}_r}^{1/2}{{{\bf{\tilde Z}}}_w}$, $\bf Y$ can be written as ${\bf Y}= {\bf{U\Sigma }}{{\bf{\Lambda }}^{\rm{H}}}$ based on singular value decomposition, where ${\bf{\Sigma }}= {\rm diag}\left\{\sqrt{{{\delta _i}}}\right\}\in {\mathbb R}_{+}^{N_r\times (M-1)}$, ${\bf U} \in {{\cal U}\left( {{N_r}} \right)}$, ${\bf \Lambda} \in {{\cal U}\left( {{M-1}} \right)}$ and ${{\cal U}\left( {{N}} \right)}$ stands for the group of unitary $N\times N$ matrices. Accordingly, $\bf Q$ can be rewritten as
\begin{equation}\label{eqn:Q_rew}
{\bf{Q}} = {\bf{U}}\underbrace {\left( {{\bf{I}} - {\bf{\Sigma }}{{\left( {{{\bf{\Sigma }}^{\rm{H}}}{\bf{\Sigma }}} \right)}^{ - 1}}{{\bf{\Sigma }}^{\rm{H}}}} \right)}_{\bf{V}}{{\bf{U}}^{\rm{H}}},
\end{equation}
where ${\bf{V}}$ is a diagonal matrix as
%Hence, the Hermitian matrix has the eigendecompostion ${\bf{Q}} = {{\bf{U}}}{\bf{V}}{\bf U}^{\rm{H}}$, where $\bf U$ is unitary, i.e., ${\bf U} {\bf U}^{\rm H} = {\bf U}^{\rm H}{\bf U} = \bf I$, and
\begin{equation}\label{eqn:V_def}
{\bf{V}} = {\rm{diag}}\left\{ {\underbrace {1, \cdots ,1}_{{N_r} - M + 1},\underbrace {0, \cdots ,0}_{M - 1}} \right\}.
\end{equation}
Thus, (\ref{eqn:cdf_F_mean}) can be further obtained as
\begin{align}\label{eqn:cdf_F_mean_re}
&F_{{{X_k}}} \left( x \right)= \notag\\
&\quad\frac{1}{{2\pi {\rm i}}}\int\nolimits_{a - {\rm i}\infty }^{a + {\rm i}\infty } {\frac{{{e^{sx}}}}{s}{{\rm{E}}_{\bf{U}}}\left\{ {{}_1{F_0}\left( {1; - s{{\bf{R}}_r}{{\bf{U}}}{\bf{V}}{\bf U}^{\rm{H}}} \right)} \right\}ds}.%\notag\\
%&\textcolor[rgb]{1.00,0.00,0.00}{\approx}\frac{1}{{2\pi {\rm i}}}\int\limits_{a - {\rm i}\infty }^{a + {\rm i}\infty }\frac{{{e^{sx}}}}{s} \times \notag\\
%&{\int\nolimits_{{\bf{U}} \in {\cal U}\left( {{N_r}} \right)} }{{}_1{F_0}\left( {1; - s{{\bf{R}}_r}^{1/2}{{\bf{U}}}{\bf{V}}{\bf U}^{\rm{H}}{{\bf{R}}_r}^{1/2}} \right)(d{\bf{U}})} ds .
\end{align}
%where ${{\cal U}\left( {{N}} \right)}$ stands for the group of unitary $N\times N$ matrices.
%\textcolor[rgb]{1.00,0.00,0.00}{ add a footnote to show the exact expression.}
To get \eqref{eqn:cdf_F_mean_re}, we have to obtain the distribution of random unitary matrix $\bf U$. With the definition of $\bf Y$ and the singular value decomposition, we arrive at \eqref{eqn:dY_ext}, as shown at the top of the following page,
\begin{figure*}[!t]
\begin{align}\label{eqn:dY_ext}
\left( {d{\bf{Y}}} \right) &= \det {\left( {{{\bf{R}}_r}} \right)^{M - 1}}\left( {d{{{\bf{\tilde Z}}}_w}} \right)\notag\\
&= \left( {M - 1} \right)!\prod\limits_{j = 1}^{M - 1} {\left( {{N_r} - j} \right)!\left( {M - 1 - j} \right)!} \prod\limits_{1 \le i < j \le M - 1} {{{\left( {{\delta _j} - {\delta _i}} \right)}^2}}\prod\limits_{i = 1}^{M - 1} {{\delta _i}^{{N_r} - M + 1}} d{\delta _1} \cdots d{\delta _{M - 1}}\left( {d{\bf{U}}} \right)\left( {d{\bf{\Lambda }}} \right),
\end{align}
\hrulefill
%\vspace*{4pt}
\end{figure*}
where the symbol $(d{\bf X})$ stands for the exterior product of the elements of $d{\bf X}$\cite{ghaderipoor2012application,muirhead2009aspects,mehta2004random}. In particular, $\left( {d{\bf{U}}} \right)$ and $\left( {d{\bf{\Lambda }}} \right)$ denote the standard Haar measures of ${{\cal U}\left( {{N_r}} \right)}$ and ${{\cal U}\left( {{M-1}} \right)}$, respectively. Hence, by combining \eqref{eqn:dY_ext} with the probability density function (PDF) of $\tilde {\bf Z}_w$, i.e., ${f_{{{{\bf{\tilde Z}}}_w}}}( {{{{\bf{\tilde Z}}}_w}} ) = {{{\pi ^{-{N_r}\left( {M - 1} \right)}}}}{\rm{etr}}( { - {{{\bf{\tilde Z}}}_w}{{{\bf{\tilde Z}}}_w}^{\rm{H}}} )$, the PDF of $\bf U$ is derived as
\begin{multline}\label{eqn:U_dist}
{f_{\bf{U}}}\left( {\bf{U}} \right) = \frac{{\left( {M - 1} \right)!\prod\nolimits_{j = 1}^{M - 1} {\left( {{N_r} - j} \right)!\left( {M - 1 - j} \right)!} }}{{\det {{\left( {{{\bf{R}}_r}} \right)}^{M - 1}}{\pi ^{{N_r}\left( {M - 1} \right)}}}}\times\\
\int\nolimits_{{\bf \Delta}  \in {\mathbb R}_ + ^{{N_r} \times \left( {M - 1} \right)}} {{\rm{etr}}\left( { - {{\bf{R}}_r}^{ - 1}{\bf{U}}{\bf \Delta} {{\bf{U}}^{\rm{H}}}} \right)\times}\\
\prod\limits_{1 \le i < j \le M - 1} {{{\left( {{\delta _j} - {\delta _i}} \right)}^2}} \prod\limits_{j = 1}^{M - 1} {{\delta _i}^{{N_r} - M + 1}} d{\delta _1} \cdots d{\delta _{M - 1}},
\end{multline}
 where ${\bf{\Delta }}= {\rm diag}\left\{{{{\delta _i}}}\right\}\in {\mathbb R}_{+}^{N_r\times (M-1)}$ and ${\rm{etr}}({\bf X})={\rm{exp}}({\rm tr}\{\bf X\})$. Unfortunately, the complex form of the distribution of $\bf U$ given by \eqref{eqn:U_dist} will render the derivation of \eqref{eqn:cdf_F_mean_re} cumbersome. Motivated by the finding that ${f_{\bf{U}}}\left( {\bf{U}} \right)=1$ if the receive antennas undergo uncorrelated fadings, i.e., ${ {{\bf{R}}_r}}={\bf I}_{N_r}$, \eqref{eqn:cdf_F_mean_re} can be approximated as
\begin{align}\label{eqn:cdf_F_mean_re1}
F_{{{X_k}}} \left( x \right)&%\notag\\
%&= \frac{1}{{2\pi {\rm i}}}\int\limits_{a - {\rm i}\infty }^{a + {\rm i}\infty } {\frac{{{e^{sx}}}}{s}{{\rm{E}}_{\bf{U}}}\left\{ {{}_1{F_0}\left( {1; - s{{\bf{R}}_r}^{1/2}{{\bf{U}}}{\bf{V}}{\bf U}^{\rm{H}}{{\bf{R}}_r}^{1/2}} \right)} \right\}ds}\notag\\
\approx \frac{1}{{2\pi {\rm i}}}\int\nolimits_{a - {\rm i}\infty }^{a + {\rm i}\infty }\frac{{{e^{sx}}}}{s}  \notag\\
&\times{\int\nolimits_{{\bf{U}} \in {\cal U}\left( {{N_r}} \right)} }{{}_1{F_0}\left( {1; - s{{\bf{R}}_r}{{\bf{U}}}{\bf{V}}{\bf U}^{\rm{H}}} \right)(d{\bf{U}})} ds\notag\\
&\triangleq \tilde F \left( x \right).
\end{align}
Afterwards, by identifying the integral in \eqref{eqn:cdf_F_mean_re1} with \cite[Theorem 7.3.3]{muirhead2009aspects}, $\tilde F \left( x \right)$ can be further rewritten as
\begin{align}\label{eqn:cdf_F_mean_re2}
\tilde F \left( x \right)&=\frac{1}{{2\pi {\rm i}}}\int\nolimits_{a - {\rm i}\infty }^{a + {\rm i}\infty } {\frac{{{e^{sx}}}}{s}} \notag\\
&\times\int\nolimits_{{\bf{U}} \in {\cal U}\left( {{N_r}} \right)} {{}_1{F_0}\left( {1; - s{{\bf{R}}_r}{{\bf{U}}}{\bf{V}}{\bf U}^{\rm{H}}} \right)(d{\bf{U}})} ds\notag\\
&= \frac{1}{{2\pi {\rm i}}}\int\nolimits_{a - {\rm i}\infty }^{a + {\rm i}\infty } {\frac{{{e^{sx}}}}{s}{}_1{F_0}^{\left( {{N_r}} \right)}\left( {1; - s{{\bf{R}}_r},{\bf{V}}} \right)ds},
\end{align}
where ${}_p{F_q}^{(m)}\left( {{\bf{X}}},{\bf Y} \right)$ denotes the hypergeometric function of two $m\times m$ Hermitian matrices \cite{muirhead2009aspects}. To derive (\ref{eqn:cdf_F_mean_re2}) in closed-form, ${}_1{F_0}^{\left( {{N_r}} \right)}\left( {1; - s{{\bf{R}}_r},{\bf{V}}} \right)$ %the hypergeometric function of two Hermitian matrices
is simplified as (\ref{eqn:F_two_matrix}), shown at the top of the next page,
\begin{figure*}[!t]
\begin{align}\label{eqn:F_two_matrix}
{}_1{F_0}^{\left( {{N_r}} \right)}\left( {1; - s{{\bf{R}}_r},{\bf{V}}} \right) &=
\frac{\det\left( {{{\left\{ {{{ {{\lambda _i}} }^{{N_r} - 1}}{{}_2{F_1}}\left( {1,{N_r} - M + 1;{N_r}; - s{\lambda _i}} \right)} \right\}}_{\scriptstyle1 \le i \le {N_r}\hfill\atop
\scriptstyle j = 1\hfill}},{{\left\{ {{\lambda _i}^{{N_r} - j}} \right\}}_{\scriptstyle1 \le i \le {N_r}\hfill\atop
\scriptstyle2 \le j \le {N_r}\hfill}}} \right)}{{\det \left( {{{\left\{ {{\lambda _i}^{{N_r} - j}} \right\}}_{1 \le i,j \le {N_r}}}} \right)}},%\notag\\
%&\times ,
%\frac{1}{{\det \left( {\left( {{\lambda _i}^{{N_r} - j}} \right)} \right)}}
%\left| {\begin{array}{*{20}{c}}
%{{{\left( {{\lambda _1}} \right)}^{{N_r} - 1}}{}_2{F_1}\left( {1,{N_r} - M + 1;{N_r}; - s{\lambda _1}} \right)}&{{\lambda _1}^{{N_r} - 2}}& \cdots &1\\
%{{{\left( {{\lambda _2}} \right)}^{{N_r} - 1}}{}_2{F_1}\left( {1,{N_r} - M + 1;{N_r}; - s{\lambda _2}} \right)}&{{\lambda _2}^{{N_r} - 2}}& \cdots &1\\
% \vdots & \vdots & \ddots & \vdots \\
%{{{\left( {{\lambda _{{N_r}}}} \right)}^{{N_r} - 1}}{}_2{F_1}\left( {1,{N_r} - M + 1;{N_r}; - s{\lambda _{{N_r}}}} \right)}&{{\lambda _{{N_r}}}^{{N_r} - 2}}& \cdots &1
%\end{array}} \right|.
\end{align}
%\hrulefill
%\vspace*{4pt}
\end{figure*}
where ${}_p{F_q}(a_1,\cdots,a_p;b_1,\cdots,b_q;x)$ denotes the hypergeometric function, ${\lambda _1}, \cdots ,{\lambda _{N_r}}$ denote the $N_r$ eigenvalues of ${{\bf{R}}_r}$, and the proof of (\ref{eqn:F_two_matrix}) is relegated to Appendix \ref{app:f_two}.

Substituting (\ref{eqn:F_two_matrix}) into (\ref{eqn:cdf_F_mean_re2}) leads to (\ref{eqn:cdf_F_mean_resub}), shown at the top of the next page.
\begin{figure*}[!t]
\begin{align}\label{eqn:cdf_F_mean_resub}
\tilde F \left( x \right)& =\frac{\det \left( {{{\left\{ {{{{{\lambda _i}}}^{{N_r} - 1}}\frac{1}{{2\pi {\rm{i}}}}\int\nolimits_{a - {\rm{i}}\infty }^{a + {\rm{i}}\infty } {\frac{{{e^{sx}}}}{s}{}_2{F_1}\left( {1,{N_r} - M + 1;{N_r}; - s{\lambda _i}} \right)ds} } \right\}}_{\scriptstyle1 \le i \le {N_r}\hfill\atop
\scriptstyle j = 1\hfill}},{{\left\{ {{\lambda _i}^{{N_r} - j}} \right\}}_{\scriptstyle1 \le i \le {N_r}\hfill\atop
\scriptstyle2 \le j \le {N_r}\hfill}}} \right)}{{\det \left( {\left\{ {{\lambda _i}^{{N_r} - j}} \right\}_{1 \le i,j \le {N_r}}} \right)}}.
%= \frac{1}{{\det \left( {\left( {{\lambda _i}^{{N_r} - j}} \right)} \right)}}\left| {\begin{array}{*{20}{c}}
%{{{\left( {{\lambda _1}} \right)}^{{N_r} - 1}}\frac{1}{{2\pi {\rm i}}}\int\limits_{a - {\rm i}\infty }^{a + {\rm i}\infty } {\frac{{{e^{sx}}}}{s}{}_2{F_1}\left( {1,{N_r} - M + 1;{N_r}; - s{\lambda _1}} \right)ds} }&{{\lambda _1}^{{N_r} - 2}}& \cdots &1\\
%{{{\left( {{\lambda _2}} \right)}^{{N_r} - 1}}\frac{1}{{2\pi {\rm i}}}\int\limits_{a - {\rm i}\infty }^{a + {\rm i}\infty } {\frac{{{e^{sx}}}}{s}{}_2{F_1}\left( {1,{N_r} - M + 1;{N_r}; - s{\lambda _2}} \right)ds} }&{{\lambda _2}^{{N_r} - 2}}& \cdots &1\\
% \vdots & \vdots & \ddots & \vdots \\
%{{{\left( {{\lambda _{{N_r}}}} \right)}^{{N_r} - 1}}\frac{1}{{2\pi {\rm i}}}\int\limits_{a - {\rm i}\infty }^{a + {\rm i}\infty } {\frac{{{e^{sx}}}}{s}{}_2{F_1}\left( {1,{N_r} - M + 1;{N_r}; - s{\lambda _{{N_r}}}} \right)ds} }&{{\lambda _{{N_r}}}^{{N_r} - 2}}& \cdots &1
%\end{array}} \right|,
\end{align}
\hrulefill
%\vspace*{4pt}
\end{figure*}
By using \cite[eq.9.34.8]{gradshteyn1965table}, the Gauss hypergeometric function can be expressed in terms of Meijer G-function as
\begin{multline}\label{eqn:hyper_meijerg}
{{}_2{F_1}\left( {1,{N_r} - M + 1;{N_r}; - s{\lambda _1}} \right)}=\\
\frac{{\Gamma \left( {{N_r}} \right)}}{{\Gamma \left( {{N_r} - M + 1} \right)}}G_{2,2}^{1,2}\left( {\left. {\begin{array}{*{20}{c}}
{0,M - {N_r}}\\
{0,1 - {N_r}}
\end{array}} \right|s{\lambda _1}} \right).
\end{multline}
where $\Gamma(x)$ denotes Gamma function. By inverting Laplace transform \cite[eq.07.34.22.0003.01]{wolframe2010math}, we reach
\begin{align}\label{eqn:inverse_hyper}
&\frac{1}{{2\pi {\rm i}}}\int\nolimits_{a - {\rm i}\infty }^{a + {\rm i}\infty } {\frac{{{e^{sx}}}}{s}{}_2{F_1}\left( {1,{N_r} - M + 1;{N_r}; - s{\lambda _i}} \right)ds} \notag\\
&= \frac{{\Gamma \left( {{N_r}} \right)}}{{\Gamma \left( {{N_r} - M + 1} \right)}}G_{3,2}^{1,2}\left( {\left. {\begin{array}{*{20}{c}}
{0,M - {N_r},1}\\
{0,1 - {N_r}}
\end{array}} \right|\frac{{{\lambda _i}}}{x}} \right) \notag\\
&= \frac{{\Gamma \left( {{N_r}} \right)}}{{\Gamma \left( {{N_r} - M + 1} \right)}}G_{2,3}^{2,1}\left( {\left. {\begin{array}{*{20}{c}}
{1,{N_r}}\\
{1,{N_r} - M + 1,0}
\end{array}} \right|\frac{x}{{{\lambda _i}}}} \right),
\end{align}
where the last step holds by using \cite[eq.9.31.2]{gradshteyn1965table}. Plugging (\ref{eqn:inverse_hyper}) into (\ref{eqn:cdf_F_mean_resub}), $\tilde F \left( x \right)$ can be finally derived as (\ref{eqn:cdf_X_k}). % at the top of the next page.
Substituting \eqref{eqn:cdf_F_mean_re1} and (\ref{eqn:cdf_X_k}) into \eqref{eqn:p_fxk}, the approximate expression of $p_{m=1,k}^{out}$ can be derived. In the similar fashion, the outage probability for arbitrary data stream $m$, $p_{m,k}^{out}$, can be obtained as \eqref{eqn:out_1_xdeffina}.

\section{Proof of (\ref{eqn:F_two_matrix})}\label{app:f_two}
With \cite[Definition 7.3.2]{muirhead2009aspects}, ${}_1{F_0}^{\left( {{N_r}} \right)}\left( {1; - s{{\bf{R}}_r},{\bf{V}}} \right)$ can be expanded as
\begin{equation}\label{eqn:hyper_gen_01_exp}
{}_1{F_0}^{\left( {{N_r}} \right)}\left( {1; - s{{\bf{R}}_r},{\bf{V}}} \right)=\sum\limits_{k = 0}^\infty  {\sum\limits_{\bs\kappa}  {{{[1]}_{\bs\kappa} }\frac{{{C_{\bs\kappa} }\left( { - s{{\bf{R}}_r}} \right){C_{\bs\kappa} }\left( {\bf{V}} \right)}}{{k!{C_{\bs\kappa} }\left( {\bf{I}} \right)}}} },
\end{equation}
where ${\bs\kappa}=(k_1,\cdots,k_{N_r})$ is a partition of an integer $k$ with $k_1\ge k_2\ge\cdots k_{N_r} \ge 0$ and $k_1+\cdots+k_{N_r}=k$, ${\left[ a \right]_{\bs\kappa} }$ is the complex multivariate hypergeometric coefficient and ${\left[ a \right]_{\bs\kappa} } = \prod\limits_{i = 1}^m {{{\left( {a - i + 1} \right)}_{{k_i}}}} $, ${\left( a \right)_k} = a\left( {a + 1} \right) \cdots \left( {a + k - 1} \right)$ and ${\left( a \right)_0}=1$. ${C_{\bs\kappa} }\left( {\bf{X}} \right)$ is defined as the complex zonal polynomial of ${\bf X} \in {\mathbb C}^{N_r \times N_r}$ associated with ${\bs\kappa}$, and is given by \cite{ratnarajah2003complex,ratnarajah2005complex}
\begin{equation}\label{eqn:czonal}
{C_{\bs\kappa} }\left( {\bf{X}} \right) = {\chi _{\left[ {\bs\kappa}  \right]}}\left( 1 \right){\chi _{\left[ {\bs\kappa}  \right]}}\left( {\bf{X}} \right),
\end{equation}
where ${\chi _{\left[ {\bs\kappa}  \right]}}\left( 1 \right)$ is the dimension of the representation $[{\bs\kappa}]$ of the symmetric group on $k$ symbols given by
\begin{equation}\label{eqn:chi_1}
 {\chi _{\left[ {\bs\kappa}  \right]}}\left( 1 \right) = k!\frac{{\prod\nolimits_{i < j}^{{N_r}} {\left( {{k_i} - {k_j} - i + j} \right)} }}{{\prod\nolimits_{i = 1}^{{N_r}} {\left( {{k_i} + m - i} \right)!} }},
\end{equation}
and ${\chi _{\left[ {\bs\kappa}  \right]}}\left( \bf X \right)$ is the character of the representation $\left[ {\bs\kappa}  \right]$ of the linear group given as a symmetric function of the eigenvalues $\mu_1,\cdots,\mu_{N_r}$ of $\bf X$ by
\begin{equation}\label{eqn:zon_charc}
{\chi _{\left[ {\bs\kappa}  \right]}}\left( {\bf{X}} \right) = \frac{{\det \left( {\left\{ {{\mu _i}^{{k_j} + {N_r} - j}} \right\}} \right)}}{{\det \left( {\left\{ {{\mu _i}^{{N_r} - j}} \right\}} \right)}}.
\end{equation}
Particularly, ${C_{\bs\kappa} }\left( {\bf{I}} \right)$ is given by\cite{ratnarajah2005complex}
\begin{equation}\label{eqn:zonal_ident}
{C_{\bs\kappa} }\left( {\bf{I}} \right) = k!\frac{{{{\left[ {\prod\nolimits_{i < j}^{{N_r}} {\left( {{k_i} - {k_j} - i + j} \right)} } \right]}^2}}}{{\prod\nolimits_{i = 1}^{{N_r}} {\left( {{k_i} + {N_r} - i} \right)!} \prod\nolimits_{i = 1}^{{N_r}} {\left( {{N_r} - i} \right)!} }}.
\end{equation}

By using \cite[eq. 5.11]{gross1989total}
\begin{equation}\label{eqn:m_coeff}
{[1]_{\bs\kappa} } = \left\{ {\begin{array}{*{20}{c}}
{{k_1}!,}&{{k_2} =  \cdots  = {k_{{N_r}}} = 0}\\
{0,}&{\rm else}
\end{array}} \right.,
\end{equation}
(\ref{eqn:hyper_gen_01_exp}) can be rewritten as
\begin{equation}\label{eqn:hyper_01_rew}
{}_1{F_0}^{\left( {{N_r}} \right)}\left( {1; - s{{\bf{R}}_r},{\bf{V}}} \right) = \sum\limits_{k = 0}^\infty  {\frac{{{C_{{{\bs\kappa} _k}}}\left( { - s{{\bf{R}}_r}} \right){C_{{{\bs\kappa} _k}}}\left( {\bf{V}} \right)}}{{{C_{{{\bs\kappa} _k}}}\left( {\bf{I}} \right)}}},
\end{equation}
where ${{\bs\kappa} _k} = \left[ {k,0, \cdots ,0} \right]$. Moreover, with (\ref{eqn:chi_1}), ${\chi _{\left[ {{{\bs\kappa} _k}} \right]}}\left( 1 \right) $ can be obtained as ${\chi _{\left[ {{{\bs\kappa} _k}} \right]}}\left( 1 \right)=1$.
%\begin{align}\label{eqn:chi_kappa_1}
%{\chi _{\left[ {{{\bs\kappa} _k}} \right]}}\left( 1 \right) &= k!\frac{{\prod\limits_{i = 1 < j}^{{N_r}} {\left( {{k_i} - {k_j} - i + j} \right)} \prod\limits_{1 < i < j}^{{N_r}} {\left( {{k_i} - {k_j} - i + j} \right)} }}{{\left( {k + {N_r} - 1} \right)!\prod\limits_{i = 2}^{{N_r}} {\left( {{N_r} - i} \right)!} }}\notag\\
%&= \frac{{k!\prod\limits_{j = 2}^{{N_r}} {\left( {k - 1 + j} \right)} \prod\limits_{1 < i < j}^{{N_r}} {\left( {j - i} \right)} }}{{\left( {k - 1 + {N_r}} \right)!\prod\limits_{i = 2}^{{N_r}} {\left( {{N_r} - i} \right)!} }} &\notag\\
%&= \frac{{\prod\limits_{1 < i < j}^{{N_r}} {\left( {j - i} \right)} }}{{\prod\limits_{i = 2}^{{N_r}} {\left( {{N_r} - i} \right)!} }} = 1.
%\end{align}
From (\ref{eqn:zonal_ident}), we have ${C_{{{\bs\kappa} _k}}}\left( {\bf{I}} \right)={\left( {{N_r}} \right)_k}/{k!}$.
%\begin{align}\label{eqn:C_I_sim}
%&{C_{{{\bs\kappa} _k}}}\left( {\bf{I}} \right) \notag\\
%&= k!\frac{{{{\left[ {\prod\limits_{i = 1 < j}^{{N_r}} {\left( {{k_i} - {k_j} - i + j} \right)} \prod\limits_{1 < i < j}^{{N_r}} {\left( {{k_i} - {k_j} - i + j} \right)} } \right]}^2}}}{{\left( {k + {N_r} - 1} \right)!\prod\limits_{i = 2}^{{N_r}} {\left( {{N_r} - i} \right)!} \prod\limits_{i = 1}^{{N_r}} {\left( {{N_r} - i} \right)!} }} \notag\\
% &= k!\frac{{{{\left[ {\prod\limits_{i = 1 < j}^{{N_r}} {\left( {k - 1 + j} \right)} \prod\limits_{1 < i < j}^{{N_r}} {\left( {j - i} \right)} } \right]}^2}}}{{\left( {k + {N_r} - 1} \right)!\prod\limits_{i = 2}^{{N_r}} {\left( {{N_r} - i} \right)!} \prod\limits_{i = 1}^{{N_r}} {\left( {{N_r} - i} \right)!} }} \notag\\
% &= k!\frac{{{{\left[ {\prod\limits_{j = 2}^{{N_r}} {\left( {k - 1 + j} \right)} } \right]}^2}}}{{\left( {k + {N_r} - 1} \right)!}(N_r-1)!} \notag\\
% &= \frac{\left( {k + {N_r} - 1} \right)!}{k!(N_r-1)!}=\frac{\left( {{N_r}} \right)_k}{k!}.
%\end{align}
Accordingly, \eqref{eqn:hyper_01_rew} collapses to
\begin{multline}\label{eqn:F_chi_kapp1}
{}_1{F_0}^{\left( {{N_r}} \right)}\left( {1; - s{{\bf{R}}_r},{\bf{V}}} \right) = \sum\limits_{k = 0}^\infty  {\frac{{{\chi _{\left[ {{{\bs\kappa} _k}} \right]}}\left( { - s{{\bf{R}}_r}} \right){\chi _{\left[ {{{\bs\kappa} _k}} \right]}}\left( {\bf{V}} \right)}}{{{C_{{{\bs\kappa} _k}}}\left( {\bf{I}} \right)}}}\\
 = \sum\limits_{k = 0}^\infty  {{{\left( { - s} \right)}^k}\frac{{k!}{{\chi _{\left[ {{{\bs\kappa} _k}} \right]}}\left( {{{\bf{R}}_r}} \right){\chi _{\left[ {{{\bs\kappa} _k}} \right]}}\left( {\bf{V}} \right)}}{{\left( {{N_r} } \right)_k}}} ,
\end{multline}
where the last step holds by using ${\chi _{\left[ {\bs\kappa}  \right]}}\left( {a{\bf{X}}} \right) = {a^k}{\chi _{\left[ {\bs\kappa}  \right]}}\left( {\bf{X}} \right)$. By using \cite[Proposition 3]{ghaderipoor2012application}, ${\chi _{\left[ {{{\bs\kappa} _k}} \right]}}\left( {\bf{V}} \right)$ can be simplified as
\begin{align}\label{eqn:chi_x_kappa_k_V}
{\chi _{\left[ {{{\bs\kappa} _k}} \right]}}\left( {\bf{V}} \right) &= {\chi _{\left[ {{\bs\kappa}_k'} \right]}}\left( {{{\bf{V}}_{\left( {1:{N_r} - M + 1} \right) \times \left( {1:{N_r} - M + 1} \right)}}} \right) \notag\\
&= \frac{{{C_{{{\bs\kappa} _k'}}}\left( {\bf{I}} \right)}}{{{\chi _{\left[ {{{\bs\kappa} _k'}} \right]}}\left( 1 \right)}}=\frac{\left( {{{N_r} - M+1}} \right)_k}{k!}.
\end{align}
where the sub-matrix ${{{\bf{V}}_{\left( {1:{N_r} - M + 1} \right) \times \left( {1:{N_r} - M + 1} \right)}}}$ is formed by deleting the last $M-1$ rows and the last $M-1$ columns of $\bf{V}$ and ${{{\bs\kappa} _k'}}$ is an $(N_r-M+1) \times 1$ vector of ones. Substituting (\ref{eqn:chi_x_kappa_k_V}) into (\ref{eqn:F_chi_kapp1}) gives
\begin{multline}\label{eqn:F_chi_kapp1_sub}
{}_1{F_0}^{\left( {{N_r}} \right)}\left( {1; - s{{\bf{R}}_r},{\bf{V}}} \right) =
\sum\limits_{k = 0}^\infty  {{{\left( { - s} \right)}^k}} \\
\times\frac{{\left( {{N_r} - M+1} \right)_k}}{{\left( {{N_r}} \right)_k}}{\chi _{\left[ {{{\bs\kappa} _k}} \right]}}\left( {{{\bf{R}}_r}} \right)\\
 = \sum\limits_{k = 0}^\infty  {{{\left( { - s} \right)}^k}\frac{{\left( {{N_r} - M+1} \right)_k}}{{\left( {{N_r}} \right)_k}}\frac{{\det \left( {\left\{ {{\lambda _i}^{{k_j} + {N_r} - j}} \right\}} \right)}}{{\det \left( {\left\{ {{\lambda _i}^{{N_r} - j}} \right\}} \right)}}}
,
\end{multline}
where $k_1=k$ and $k_2=\cdots=k_{N_r}=0$. After some algebraic manipulations, (\ref{eqn:F_chi_kapp1_sub}) can be simplified as (\ref{eqn:F_chi_kapp1_subsim}) at the top of the next page.
\begin{figure*}[!t]
\begin{multline}\label{eqn:F_chi_kapp1_subsim}
{}_1{F_0}^{\left( {{N_r}} \right)}\left( {1; - s{{\bf{R}}_r},{\bf{V}}} \right)
%&= \frac{1}{{\det \left( {\left( {{\lambda _i}^{{N_r} - j}} \right)} \right)}}\sum\limits_{k = 0}^\infty  {{{\left( { - s} \right)}^k}\frac{{\left( {{N_r} - M+1} \right)_k}}{{\left( {{N_r}} \right)_k}}\left| {\begin{array}{*{20}{c}}
%{{\lambda _1}^{k + {N_r} - 1}}&{{\lambda _1}^{{N_r} - 2}}& \cdots &1\\
%{{\lambda _2}^{k + {N_r} - 1}}&{{\lambda _2}^{{N_r} - 2}}& \cdots &1\\
% \vdots & \vdots & \ddots & \vdots \\
%{{\lambda _{{N_r}}}^{k + {N_r} - 1}}&{{\lambda _{{N_r}}}^{{N_r} - 2}}& \cdots &1
%\end{array}} \right|}
=\frac{1}{{\det \left( {\left\{ {{\lambda _i}^{{N_r} - j}} \right\}} \right)}}\sum\limits_{k = 0}^\infty  {{{\left( { - s} \right)}^k}\frac{{{{\left( {{N_r} - M + 1} \right)}_k}}}{{{{\left( {{N_r}} \right)}_k}}}\det \left( {{{\left\{ {{{\left( {{\lambda _i}} \right)}^{k + {N_r} - 1}}} \right\}}_{\scriptstyle 1 \le i \le {N_r}\hfill\atop
\scriptstyle j = 1\hfill}},{{\left\{ {{\lambda _i}^{{N_r} - j}} \right\}}_{\scriptstyle 1 \le i \le {N_r}\hfill\atop
\scriptstyle 2 \le j \le {N_r}\hfill}}} \right)} \\
%&= \frac{1}{{\det \left( {\left( {{\lambda _i}^{{N_r} - j}} \right)} \right)}}\left| {\begin{array}{*{20}{c}}
%{{{\left( {{\lambda _1}} \right)}^{{N_r} - 1}}\sum\limits_{k = 0}^\infty  {\frac{{{{\left( {{N_r} - M + 1} \right)}_k}}}{{{{\left( {{N_r}} \right)}_k}}}} {{\left( { - s{\lambda _1}} \right)}^k}}&{{\lambda _1}^{{N_r} - 2}}& \cdots &1\\
%{{{\left( {{\lambda _2}} \right)}^{{N_r} - 1}}\sum\limits_{k = 0}^\infty  {\frac{{{{\left( {{N_r} - M + 1} \right)}_k}}}{{{{\left( {{N_r}} \right)}_k}}}} {{\left( { - s{\lambda _2}} \right)}^k}}&{{\lambda _2}^{{N_r} - 2}}& \cdots &1\\
% \vdots & \vdots & \ddots & \vdots \\
%{{{\left( {{\lambda _{{N_r}}}} \right)}^{{N_r} - 1}}\sum\limits_{k = 0}^\infty  {\frac{{{{\left( {{N_r} - M + 1} \right)}_k}}}{{{{\left( {{N_r}} \right)}_k}}}} {{\left( { - s{\lambda _{{N_r}}}} \right)}^k}}&{{\lambda _{{N_r}}}^{{N_r} - 2}}& \cdots &1
%\end{array}} \right|
=\frac{1}{{\det \left( {\left\{ {{\lambda _i}^{{N_r} - j}} \right\}} \right)}}\det \left( {{{\left\{ {{\lambda _i}^{{N_r} - 1}\sum\limits_{k = 0}^\infty  {\frac{{{{\left( {{N_r} - M + 1} \right)}_k}}}{{{{\left( {{N_r}} \right)}_k}}}} {{\left( { - s{\lambda _i}} \right)}^k}} \right\}}_{\scriptstyle1 \le i \le {N_r}\hfill\atop
\scriptstyle j = 1\hfill}},{{\left\{ {{\lambda _i}^{{N_r} - j}} \right\}}_{\scriptstyle1 \le i \le {N_r}\hfill\atop
\scriptstyle2 \le j \le {N_r}\hfill}}} \right),
\end{multline}
%\hrulefill
%\vspace*{4pt}
\end{figure*}
With the definition of hypergeometric function, (\ref{eqn:F_chi_kapp1_subsim}) can be further expressed as (\ref{eqn:F_two_matrix}).

\section{Proof of \eqref{eqn:out_mKconst_corr_spefina}}\label{app:der_meijerG}
By virtue of the generalized L'H\^opital's rule \cite[Theorem 1.2.4]{hua1963harmonic}, $p_{m,k}^{out}$ can be obtained as
\begin{align}\label{eqn:out_mKconst_corr_spe}
&p_{m,k}^{out}= \frac{1}{{\Gamma \left( {{N_r} - M + 1} \right)}}{\frac{{{\partial ^{{N_r} - 1}}}}{{\partial {\lambda }^{{N_r} - 1}}}\Biggl[ {{\lambda }^{{N_r} - 1}}\times }\notag\\
&\,\,\,\,\,{ {G_{2,3}^{2,1}\left( {\left. {\begin{array}{*{20}{c}}
{1,{N_r}}\\
{1,{N_r} - M + 1,0}
\end{array}} \right|\frac{{{{\left[ {{{\bf{R}}_{t'}}^{ - 1}} \right]}_{mm}}}}{{{\lambda }\bar \gamma {\theta _{m,k}}\ell ({d_k})}}} \right)} \Biggr]_{{\lambda } = {\lambda _0}}}.
%&{\left. {\frac{{{\partial ^{{N_r} - 1}}{{{{\lambda _1}} }^{{N_r} - 1}}G_{2,3}^{2,1}\left( {\left. {\begin{array}{*{20}{c}}
%{1,{N_r}}\\
%{1,{N_r} - M + 1,0}
%\end{array}} \right|\frac{{{{\left[ {{{\bf{R}}_{t'}}^{ - 1}} \right]}_{mm}}}}{{{\lambda _1}{{\bar \gamma }}{\theta _{m,k}}\ell ({d_k})}}} \right)}}{{\partial {\lambda _1}^{{N_r} - 1}}}} \right|_{{\lambda _1} = {\lambda _0}}}.
\end{align}
%\hrulefill
%\vspace*{4pt}
%\end{figure*}
The $({{N_r} - 1})$-order partial derivative of Meijer G-function in (\ref{eqn:out_mKconst_corr_spe}) can be derived as \eqref{eqn:deri_meijerg_proof}, as shown at the top of the following page, wherein the first and last steps hold by using \cite[eq.9.31.5]{gradshteyn1965table}, and the second equality holds by using \cite[eq.07.34.20.0013.02]{wolframe2010math}.
\begin{figure*}[!t]
\begin{align}\label{eqn:deri_meijerg_proof}
%&\frac{{{\partial ^{{N_r} - 1}}{{{{\lambda _1}}}^{{N_r} - 1}}G_{2,3}^{2,1}\left( {\left. {\begin{array}{*{20}{c}}
%{1,{N_r}}\\
%{1,{N_r} - M + 1,0}
%\end{array}} \right|\frac{{{{\left[ {{{\bf{R}}_{t'}}^{ - 1}} \right]}_{mm}}}}{{{\lambda _1}\bar \gamma {\theta _{m,k}}\ell ({d_k})}}} \right)}}{{\partial {\lambda _1}^{{N_r} - 1}}} \notag\\
%&  = {\left( {\frac{{{{\left[ {{{\bf{R}}_{t'}}^{ - 1}} \right]}_{mm}}}}{{{{\bar \gamma }}{\theta _{m,k}}\ell ({d_k})}}} \right)^{{N_r} - 1}}\frac{{{\partial ^{{N_r} - 1}}}}{{\partial {\lambda _1}^{{N_r} - 1}}}{\left[ {{{\left( {\frac{{{{\left[ {{{\bf{R}}_{t'}}^{ - 1}} \right]}_{mm}}}}{{{\lambda _1}{{\bar \gamma }}{\theta _{m,k}}\ell ({d_k})}}} \right)}^{ - {N_r} + 1}}G_{2,3}^{2,1}\left( {\left. {\begin{array}{*{20}{c}}
%{1,{N_r}}\\
%{1,{N_r} - M + 1,0}
%\end{array}} \right|\frac{{{{\left[ {{{\bf{R}}_{t'}}^{ - 1}} \right]}_{mm}}}}{{{\lambda _1}{{\bar \gamma }}{\theta _{m,k}}\ell ({d_k})}}} \right)} \right]_{{\lambda _1} = {\lambda _0}}} \notag\\
\eqref{eqn:deri_meijerg}& \coloneqq {\left( {\frac{{{{\left[ {{{\bf{R}}_{t'}}^{ - 1}} \right]}_{mm}}}}{{\bar \gamma {\theta _{m,k}}\ell ({d_k})}}} \right)^{{N_r} - 1}}\frac{{{\partial ^{{N_r} - 1}}}}{{\partial {\lambda _1}^{{N_r} - 1}}}\left[ {G_{2,3}^{2,1}\left( {\left. {\begin{array}{*{20}{c}}
{2 - {N_r},1}\\
{2 - {N_r},2 - M,1 - {N_r}}
\end{array}} \right|\frac{{{{\left[ {{{\bf{R}}_{t'}}^{ - 1}} \right]}_{mm}}}}{{{\lambda _1}\bar \gamma {\theta _{m,k}}\ell ({d_k})}}} \right)} \right]\notag\\
%&  = {\left[ {\frac{{{\partial ^{{N_r} - 1}}}}{{\partial {z^{{N_r} - 1}}}}G_{2,3}^{2,1}\left( {\left. {\begin{array}{*{20}{c}}
%{1 - {N_r} + 1,{N_r} - {N_r} + 1}\\
%{1 - {N_r} + 1,{N_r} - M + 1 - {N_r} + 1,0 - {N_r} + 1}
%\end{array}} \right|{z^{ - 1}}} \right)} \right]_{z = \frac{{{\lambda _0}{{\bar \gamma }}{\theta _{m,k}}\ell ({d_k})}}{{{{\left[ {{{\bf{R}}_{t'}}^{ - 1}} \right]}_{mm}}}}}} \notag\\
&  = \frac{{{\partial ^{{N_r} - 1}}}}{{\partial {z^{{N_r} - 1}}}}{\left[ {G_{2,3}^{2,1}\left( {\left. {\begin{array}{*{20}{c}}
{2 - {N_r},1}\\
{2 - {N_r},2 - M,1 - {N_r}}
\end{array}} \right|{z^{ - 1}}} \right)} \right]_{z = \frac{{{\lambda _1}{{\bar \gamma }}{\theta _{m,k}}\ell ({d_k})}}{{{{\left[ {{{\bf{R}}_{t'}}^{ - 1}} \right]}_{mm}}}}}}\notag\\
%&  = {\left( {\frac{{{{\left[ {{{\bf{R}}_{t'}}^{ - 1}} \right]}_{mm}}}}{{{\lambda _1}{{\bar \gamma }}{\theta _{m,k}}\ell ({d_k})}}} \right)^{{N_r} - 1}}G_{3,4}^{3,1}\left( {\left. {\begin{array}{*{20}{c}}
%{2 - {N_r},1,2 - {N_r}}\\
%{1,2 - {N_r},2 - M,1 - {N_r}}
%\end{array}} \right|\frac{{{{\left[ {{{\bf{R}}_{t'}}^{ - 1}} \right]}_{mm}}}}{{{\lambda _1}{{\bar \gamma }}{\theta _{m,k}}\ell ({d_k})}}} \right)\notag\\
&  = {\left( {\frac{{{{\left[ {{{\bf{R}}_{t'}}^{ - 1}} \right]}_{mm}}}}{{{\lambda _1}{{\bar \gamma }}{\theta _{m,k}}\ell ({d_k})}}} \right)^{{N_r} - 1}}G_{1,2}^{1,1}\left( {\left. {\begin{array}{*{20}{c}}
{2 - {N_r}}\\
{2 - M,1 - {N_r}}
\end{array}} \right|\frac{{{{\left[ {{{\bf{R}}_{t'}}^{ - 1}} \right]}_{mm}}}}{{{\lambda _1}{{\bar \gamma }}{\theta _{m,k}}\ell ({d_k})}}} \right)\notag\\
& = G_{1,2}^{1,1}\left( {\left. {\begin{array}{*{20}{c}}
1\\
{{N_r} - M + 1,0}
\end{array}} \right|\frac{{{{\left[ {{{\bf{R}}_{t'}}^{ - 1}} \right]}_{mm}}}}{{{\lambda _1}{{\bar \gamma }}{\theta _{m,k}}\ell ({d_k})}}} \right),% = {{\gamma \left( {{N_r} - M + 1,\frac{{{{\left[ {{{\bf{R}}_{t'}}^{ - 1}} \right]}_{mm}}}}{{{\lambda _1}{{\bar \gamma }}{\theta _{m,k}}\ell ({d_k})}}} \right)}},
\end{align}
\hrulefill
%\vspace*{4pt}
\end{figure*}
Applying \cite[eq.06.07.26.0006.01]{wolframe2010math} to \eqref{eqn:deri_meijerg_proof} then leads to % \eqref{eqn:deri_meijerg}.
%The $({{N_r} - 1})$-order partial derivative of Meijer G-function in (\ref{eqn:out_mKconst_corr_spe}) can be rewritten as %(\ref{eqn:deri_meijerg}), shown at the top of the page after the next page,
\begin{align}\label{eqn:deri_meijerg}
&\frac{{{\partial ^{{N_r} - 1}}{{{{\lambda }}}^{{N_r} - 1}}G_{2,3}^{2,1}\left( {\left. {\begin{array}{*{20}{c}}
{1,{N_r}}\\
{1,{N_r} - M + 1,0}
\end{array}} \right|\frac{{{{\left[ {{{\bf{R}}_{t'}}^{ - 1}} \right]}_{mm}}}}{{{\lambda }\bar \gamma {\theta _{m,k}}\ell ({d_k})}}} \right)}}{{\partial {\lambda }^{{N_r} - 1}}} \notag\\
& = {{\gamma \left( {{N_r} - M + 1,\frac{{{{\left[ {{{\bf{R}}_{t'}}^{ - 1}} \right]}_{mm}}}}{{{\lambda }{{\bar \gamma }}{\theta _{m,k}}\ell ({d_k})}}} \right)}},
\end{align}
Substituting (\ref{eqn:deri_meijerg}) into (\ref{eqn:out_mKconst_corr_spe}) eventually gives \eqref{eqn:out_mKconst_corr_spefina}.

\section{Proof of \eqref{eqn:meijer_G_simple}}\label{app:F_tilde_exp}
To start, the asymptotic expression of Meijer G-function $G_{2,3}^{2,1}\left( {\cdot|z} \right)$ for $z \to 0$ is derived first. By using \cite[eq.07.34.06.0045.01]{wolframe2010math}, the Meijer G-function $G_{2,3}^{2,1}\left( {\cdot|z} \right)$ can be written as
%The Meijer G-function can be represented by a Mellin-Barnes contour integral, such that\cite[eq.9.301]{gradshteyn1965table}
%\begin{multline}\label{eqn:meijer_G_brane}
%G_{2,3}^{2,1}\left( {\left. {\begin{array}{*{20}{c}}
%{1,{N_r}}\\
%{1,{N_r} - M + 1,0}
%\end{array}} \right|z} \right) =\\
%\frac{1}{{2\pi {\rm i}}}\int\nolimits_{c - {\rm i}\infty }^{c + {\rm i}\infty } {\frac{{\Gamma \left( {1 - \theta } \right)\Gamma \left( {{N_r} - M + 1 - \theta } \right)\Gamma \left( \theta  \right)}}{{\Gamma \left( {{N_r} - \theta } \right)\Gamma \left( {1 + \theta } \right)}}{z^\theta }} d\theta,
%\end{multline}
%where $c$ is selected as $0 < c < 1$ to separate the poles of $\Gamma \left( {1 - \theta } \right)$, $\Gamma \left( {{N_r} - M + 1 - \theta } \right)$ from the poles of $\Gamma\left( { \theta } \right)$. Note that there are an infinite number of poles on the right-hand side of the integration path, summing up all the residues at the poles $\{1,2,\cdots,\infty\}$ yields
\begin{align}\label{eqn:meijer_G_res}
&G_{2,3}^{2,1}\left( {\left. {\begin{array}{*{20}{c}}
{1,{N_r}}\\
{1,{N_r} - M + 1,0}
\end{array}} \right|z} \right) =\notag\\
&  - \sum\limits_{\vartheta  = 1}^\infty  {{\rm{Res}}\left\{ {\frac{{\Gamma \left( {1 - \theta } \right)\Gamma \left( {{N_r} - M + 1 - \theta } \right)\Gamma \left( \theta  \right)}}{{\Gamma \left( {{N_r} - \theta } \right)\Gamma \left( {1 + \theta } \right)}}{z^\theta },\theta  = \vartheta } \right\}},
\end{align}
where ${\rm Res}(·; v)$ stands for the residue at $s = v$. Excluding the second order poles at $\{N_r-M+1,\cdots,N_r-1\}$, all the other poles are simple. By using residue theorem, (\ref{eqn:meijer_G_res}) can be further expanded as
%written as (\ref{eqn:meijer_g_res_exp}) at the top of this page.
%\begin{figure*}[!t]
%\begin{multline}\label{eqn:meijer_g_res_exp}
%G_{2,3}^{2,1}\left( {\left. {\begin{array}{*{20}{c}}
%{1,{N_r}}\\
%{1,{N_r} - M + 1,0}
%\end{array}} \right|z} \right) =  - \sum\limits_{\vartheta  = 1}^{{N_r} - M} {{{\left[ {\frac{{\left( {\theta  - \vartheta } \right)\Gamma \left( {1 - \theta } \right)\Gamma \left( {{N_r} - M + 1 - \theta } \right)\Gamma \left( \theta  \right)}}{{\Gamma \left( {{N_r} - \theta } \right)\Gamma \left( {1 + \theta } \right)}}{z^\theta }} \right]}_{\theta  = \vartheta }}} \\
% - \sum\limits_{\vartheta  = {N_r} - M + 1}^{{N_r} - 1} {{{\left[ {\frac{d}{{d\theta }}\frac{{{{\left( {\theta  - \vartheta } \right)}^2}\Gamma \left( {1 - \theta } \right)\Gamma \left( {{N_r} - M + 1 - \theta } \right)\Gamma \left( \theta  \right)}}{{\Gamma \left( {{N_r} - \theta } \right)\Gamma \left( {1 + \theta } \right)}}{z^\theta }} \right]}_{\theta  = \vartheta }}} \\
% - \sum\limits_{\vartheta  = {N_r}}^\infty  {{{\left[ {\frac{{\left( {\theta  - \vartheta } \right)\Gamma \left( {1 - \theta } \right)\Gamma \left( {{N_r} - M + 1 - \theta } \right)\Gamma \left( \theta  \right)}}{{\Gamma \left( {{N_r} - \theta } \right)\Gamma \left( {1 + \theta } \right)}}{z^\theta }} \right]}_{\theta  = \vartheta }}}.
%\end{multline}
%\hrulefill
%\vspace*{4pt}
%\end{figure*}
%After some tedious algebraic manipulations, (\ref{eqn:meijer_g_res_exp}) can be finally expanded as
\begin{align}\label{eqn:meig_res_simp}
&G_{2,3}^{2,1}\left( {\left. {\begin{array}{*{20}{c}}
{1,{N_r}}\\
{1,{N_r} - M + 1,0}
\end{array}} \right|z} \right) =\notag \\
&\sum\limits_{\vartheta  = 1}^{{N_r} - M} {\frac{{{{\left( { - 1} \right)}^{\vartheta  - 1}}\left( {{N_r} - M - \vartheta } \right)!}}{{\vartheta !\left( {{N_r} - 1 - \vartheta } \right)!}}{z^\vartheta }}  + \sum\limits_{\vartheta  = {N_r} - M + 1}^{{N_r} - 1} {\varphi \left( \vartheta  \right){z^\vartheta }} \notag\\
& + \sum\limits_{\vartheta  = {N_r} - M + 1}^{{N_r} - 1} {\frac{{{{\left( { - 1} \right)}^{2\vartheta  - {N_r} + M - 1}}{z^\vartheta }\ln z}}{{\vartheta !\left( {\vartheta  - \left( {{N_r} - M + 1} \right)} \right)!\left( {{N_r} - 1 - \vartheta } \right)!}}}\notag\\
& + \sum\limits_{\vartheta  = {N_r}}^\infty  {\frac{{{{\left( { - 1} \right)}^{\vartheta  + M - 2}}\left( {\vartheta  - {N_r}} \right)!}}{{\vartheta !\left( {\vartheta  - {N_r} + M - 1} \right)!}}{z^\vartheta }},
\end{align}
where $\varphi \left( \vartheta  \right)$ is given by
\begin{align}\label{eqn:phi_def}
&\varphi \left( \vartheta  \right) = {\left( { - 1} \right)^{2\vartheta  - {N_r} + M - 1}}\times\notag\\
&{\left[ {\frac{d}{{d\theta }}\frac{{{{\left( {\Gamma \left( {\vartheta  + 1 - \theta } \right)} \right)}^2}}}{{\prod\limits_{j = 0}^{\vartheta  - 1} {\left( {\theta  - j} \right)} \prod\limits_{j = {N_r} - M + 1}^{\vartheta  - 1} {\left( {\theta  - j} \right)} \Gamma \left( {{N_r} - \theta } \right)}}} \right]_{\theta  = \vartheta }}.
\end{align}
Substituting (\ref{eqn:meig_res_simp}) into (\ref{eqn:cdf_X_k}) yields (\ref{eqn:meijer_G_sub}) at the top of the next page.
\begin{figure*}[!t]
\begin{align}\label{eqn:meijer_G_sub}
\tilde F \left( x \right) &= \frac{{\Gamma \left( {{N_r}} \right)}}{{\Gamma \left( {{N_r} - M + 1} \right)\det \left( {\left\{ {{\lambda _i}^{{N_r} - j}} \right\}} \right)}}\sum\limits_{\vartheta  = 1}^{{N_r} - M} {\frac{{{{\left( { - 1} \right)}^{\vartheta  - 1}}\left( {{N_r} - M - \vartheta } \right)!}}{{\vartheta !\left( {{N_r} - 1 - \vartheta } \right)!}}{{x}^\vartheta }} \det \left( \left\{{{{{\lambda _i}}}^{{N_r} - \vartheta  - 1}},\left. {{{{{\lambda _i}}}^{{N_r} - j}}} \right|_{j = 2}^{{N_r}}\right\} \right)\notag\\
& + \frac{{\Gamma \left( {{N_r}} \right)}}{{\Gamma \left( {{N_r} - M + 1} \right)\det \left( {\left\{ {{\lambda _i}^{{N_r} - j}} \right\}} \right)}}\sum\limits_{\vartheta  = {N_r} - M + 1}^{{N_r} - 1} {\varphi \left( \vartheta  \right){{x}^\vartheta }} \det \left( \left\{{{{{\lambda _i}}}^{{N_r} - \vartheta  - 1}},\left. {{{{{\lambda _i}}}^{{N_r} - j}}} \right|_{j = 2}^{{N_r}}\right\} \right)\notag\\
& + \frac{{\Gamma \left( {{N_r}} \right)}}{{\Gamma \left( {{N_r} - M + 1} \right)\det \left( {\left\{ {{\lambda _i}^{{N_r} - j}} \right\}} \right)}}\sum\limits_{\vartheta  = {N_r} - M + 1}^{{N_r} - 1} {\frac{{{{\left( { - 1} \right)}^{2\vartheta  - {N_r} + M - 1}}}{{x}^\vartheta }}{{\vartheta !\left( {\vartheta  - \left( {{N_r} - M + 1} \right)} \right)!\left( {{N_r} - 1 - \vartheta } \right)!}}} \notag\\
%& \times \det \left( \left\{{{{{\lambda _i}}}^{{N_r} - 1}}{{\left( {\frac{x}{{{\lambda _i}}}} \right)}^\vartheta }\ln \left( {\frac{x}{{{\lambda _i}}}} \right),\left. {{{{{\lambda _i}}}^{{N_r} - j}}} \right|_{j = 2}^{{N_r}}\right\} \right)\notag\\
& \times \det \left( \left\{{{{{\lambda _i}}}^{{N_r} - \vartheta - 1}}\ln \left( {\frac{x}{{{\lambda _i}}}} \right),\left. {{{{{\lambda _i}}}^{{N_r} - j}}} \right|_{j = 2}^{{N_r}}\right\} \right)\notag\\
& + \frac{{\Gamma \left( {{N_r}} \right)}}{{\Gamma \left( {{N_r} - M + 1} \right)\det \left( {\left\{ {{\lambda _i}^{{N_r} - j}} \right\}} \right)}}\sum\limits_{\vartheta  = {N_r}}^\infty  {\frac{{{{\left( { - 1} \right)}^{\vartheta  + M - 2}}\left( {\vartheta  - {N_r}} \right)!}}{{\vartheta !\left( {\vartheta  - {N_r} + M - 1} \right)!}}{x^\vartheta }} \det \left( \left\{{{{{\lambda _i}}}^{{N_r} - \vartheta  - 1}},\left. {{{{{\lambda _i}}}^{{N_r} - j}}} \right|_{j = 2}^{{N_r}}\right\} \right).
\end{align}
%\hrulefill
%\vspace*{4pt}
\end{figure*}
Since $\det \left( \left\{{{{{\lambda _i}} }^\nu },\left. {{{{{\lambda _i}} }^{{N_r} - j}}} \right|_{j = 2}^{{N_r}}\right\} \right) = 0$ for $0 \le \nu  \le {N_r} - 2$, (\ref{eqn:meijer_G_sub}) can be further rearranged as \eqref{eqn:meijer_G_simple}.

\section{Proof of Theorem \ref{the:pow}}\label{app:pow}
With regard to (\ref{eqn:opt_refor_rew}), the associated Lagrangian is written as (\ref{eqn:lag}) at the top of the following page,
\begin{figure*}[!t]
\begin{align}\label{eqn:lag}
&{\cal L}\left( {{{\left\{ {{\zeta _{m,k}}} \right\}}_{m \in \left[ {1,M} \right]\hfill\atop
k \in \left[ {1,K} \right]\hfill}},{{\left\{ {{\theta _{m,k}}} \right\}}_{m \in \left[ {1,M} \right]\hfill\atop
k \in \left[ {1,K} \right]\hfill}},{{\left\{ {{\eta _m}} \right\}}_{m \in \left[ {1,M} \right]}},{{\left\{ {{\xi _{m,k}}} \right\}}_{m \in \left[ {1,M} \right]\hfill\atop
k \in \left[ {1,K} \right]\hfill}},{{\left\{ {{\varsigma _{m,k}}} \right\}}_{m \in [1,M]\hfill\atop
k \in [1,K]\hfill}},{{\left\{ {{\omega _{m,k,i}}} \right\}}_{m \in [1,M]\hfill\atop
{k \in \left[ {1,K} \right]\hfill\atop
i \in \left[ {k,K} \right]\hfill}}}} \right)\notag\\
&=\sum\nolimits_{m \in \left[ {1,M} \right]\hfill\atop
k \in \left[ {1,K} \right]\hfill} {{\phi _{m,k}}{R_{m,k}}{\theta _{m,k}}^{ - \left( {{N_r} - M + 1} \right)}}  + \sum\nolimits_{m \in \left[ {1,M} \right]} {{\eta _m}\left( {\sum\nolimits_{k = 1}^K {{\zeta _{m,k}}}  - \frac{1}{M}} \right)}  - \sum\nolimits_{m \in \left[ {1,M} \right]\hfill\atop
k \in \left[ {1,K} \right]\hfill} {{\xi _{m,k}}{\zeta _{m,k}}}  \notag\\
&\quad  - \sum\nolimits_{m \in \left[ {1,M} \right]\hfill\atop
k \in \left[ {1,K} \right]\hfill} {{\varsigma _{m,k}}{\theta _{m,k}}}  + \sum\nolimits_{m \in [1,M]\hfill\atop
{k \in \left[ {1,K} \right]\hfill\atop
i \in \left[ {k,K} \right]\hfill}} {{\omega _{m,k,i}}\left( {{\theta _{m,k}} - \frac{{{\zeta _{m,i}}}}{{{2^{{R_{m,i}}}} - 1}} + \sum\nolimits_{l = 1}^{i - 1} {{\zeta _{m,l}}} } \right)}  ,
\end{align}
%\hrulefill
%\vspace*{4pt}
\end{figure*}
where ${{\eta _m}}$, ${{\xi _{m,k}}}$, ${{\varsigma _{m,k}}}$ and ${{\omega _{m,k,i}}}$ refer to the KKT multipliers. On the basis of (\ref{eqn:lag}), the optimal solution should satisfy the necessary KKT conditions (\ref{eqn:kkt_zetamul})-(\ref{eqn:kkt_firdtheta}) as follows %\eqref{eqn:kkt_firdtheta}, as shown at the top of the next page, and \eqref{eqn:kkt_zetamul}-(\ref{eqn:kkt_powercon}) as follows
\begin{equation}\label{eqn:kkt_zetamul}
{\xi _{m,k}^*}{\zeta _{m,k}^*} = 0,
\end{equation}
\begin{equation}\label{eqn:kkt_thetamul}
{\varsigma _{m,k}^*}{\theta _{m,k}^*} = 0,
\end{equation}
\begin{equation}\label{eqn:kkt_thetareq}
{\omega _{m,k,i}^*}\left( {{\theta _{m,k}^*} - \frac{{{\zeta _{m,i}^*}}}{{{2^{{R_{m,i}}}} - 1}} + \sum\nolimits_{l = 1}^{i - 1} {{\zeta _{m,l}^*}} } \right) = 0,
\end{equation}
\begin{equation}\label{eqn:kkt_powercon}
\sum\nolimits_{k = 1}^K {{\zeta _{m,k}^*}}  = \frac{1}{M},
\end{equation}
where (\ref{eqn:kkt_firDzeta}) and (\ref{eqn:kkt_firdtheta}) are shown at the top of the next page, ${\xi _{m,k}^*} \ge 0$, ${\varsigma _{m,k}^*} \ge 0$ and ${\omega _{m,k,i}^*} \ge 0$. \begin{figure*}[!t]
\begin{equation}\label{eqn:kkt_firDzeta}
{\left. {\frac{{\partial {\cal L}}}{{\partial {\zeta _{m,k}}}}} \right|_{{{\left\{ {\zeta _{m,k}^*} \right\}}_{m \in \left[ {1,M} \right]\hfill\atop
k \in \left[ {1,K} \right]\hfill}},{{\left\{ {\theta _{m,k}^*} \right\}}_{m \in \left[ {1,M} \right]\hfill\atop
k \in \left[ {1,K} \right]\hfill}},{{\left\{ {\eta _m^*} \right\}}_{m \in \left[ {1,M} \right]}},{{\left\{ {\xi _{m,k}^*} \right\}}_{m \in \left[ {1,M} \right]\hfill\atop
k \in \left[ {1,K} \right]\hfill}},{{\left\{ {\varsigma _{m,k}^*} \right\}}_{m \in [1,M]\hfill\atop
k \in [1,K]\hfill}},{{\left\{ {\omega _{m,k,i}^*} \right\}}_{m \in [1,M]\hfill\atop
{k \in \left[ {1,K} \right]\hfill\atop
i \in \left[ {k,K} \right]\hfill}}}}} = 0,
\end{equation}
%\hrulefill
%\vspace*{4pt}
\end{figure*}
\begin{figure*}[!t]
\begin{equation}\label{eqn:kkt_firdtheta}
{\left. {\frac{{\partial {\cal L}}}{{\partial {\theta _{m,k}}}}} \right|_{{{\left\{ {\zeta _{m,k}^*} \right\}}_{m \in \left[ {1,M} \right]\hfill\atop
k \in \left[ {1,K} \right]\hfill}},{{\left\{ {\theta _{m,k}^*} \right\}}_{m \in \left[ {1,M} \right]\hfill\atop
k \in \left[ {1,K} \right]\hfill}},{{\left\{ {\eta _m^*} \right\}}_{m \in \left[ {1,M} \right]}},{{\left\{ {\xi _{m,k}^*} \right\}}_{m \in \left[ {1,M} \right]\hfill\atop
k \in \left[ {1,K} \right]\hfill}},{{\left\{ {\varsigma _{m,k}^*} \right\}}_{m \in [1,M]\hfill\atop
k \in [1,K]\hfill}},{{\left\{ {\omega _{m,k,i}^*} \right\}}_{m \in [1,M]\hfill\atop
{k \in \left[ {1,K} \right]\hfill\atop
i \in \left[ {k,K} \right]\hfill}}}}} = 0,
\end{equation}
\hrulefill
%\vspace*{4pt}
\end{figure*}
By putting (\ref{eqn:lag}) into (\ref{eqn:kkt_firDzeta}) and (\ref{eqn:kkt_firdtheta}), respectively, we reach
\begin{multline}\label{eqn:kkt1_firzeta}
\eta _m^* - \xi _{m,k}^* - \frac{1}{{{2^{{R_{m,k}}}} - 1}}\sum\nolimits_{t \in \left[ {1,k} \right]} {\omega _{m,t,k}^*}  \\
+ \sum\nolimits_{ t \in \left[ {k + 1,K} \right]\hfill\atop
 i \in \left[ {t,K} \right]\hfill} {\omega _{m,t,i}^*}   = 0,
\end{multline}
\begin{multline}\label{eqn:kkt2_firtheta}
 - \left( {{N_r} - M + 1} \right){\phi _{m,k}}{R_{m,k}}{\theta _{m,k}^*}^{ - \left( {{N_r} - M + 2} \right)} \\
 - \varsigma _{m,k}^* + \sum\nolimits_{i \in \left[ {k,K} \right]} {\omega _{m,k,i}^*}  = 0.
\end{multline}
Notice ${{\zeta _{m,k}} > 0}$ and ${{\theta _{m,k}} > 0}$, (\ref{eqn:kkt_zetamul}) and (\ref{eqn:kkt_thetamul}) suggest $\xi _{m,k}^* = 0$ and $\varsigma _{m,k}^* = 0$, respectively. Thus, (\ref{eqn:kkt1_firzeta}) and (\ref{eqn:kkt2_firtheta}) can be rewritten as
\begin{equation}\label{eqn:kkt1_firzetare}
\eta _m^* - \frac{1}{{{2^{{R_{m,k}}}} - 1}}\sum\limits_{t \in \left[ {1,k} \right]} {\omega _{m,t,k}^*}  + \sum\limits_{t \in \left[ {k + 1,K} \right]\hfill\atop
i \in \left[ {t,K} \right]\hfill} {\omega _{m,t,i}^*}  = 0,
\end{equation}
\begin{multline}\label{eqn:kkt2_firthetare}
 - \left( {{N_r} - M + 1} \right){\phi _{m,k}}{R_{m,k}}{\theta _{m,k}^*}^{ - \left( {{N_r} - M + 2} \right)} \\
 + \sum\nolimits_{i \in \left[ {k,K} \right]} {\omega _{m,k,i}^*}  = 0.
\end{multline}
To proceed, we assume that ${i_{m,k}} = \mathop {\arg }\limits_i \min \left\{ {{{\zeta _{m,i}^*}}/{({{2^{{R_{m,i}}}} - 1})} - \sum\nolimits_{l = 1}^{i - 1} {\zeta _{m,l}^*} ,i \in \left[ {k,K} \right]} \right\}$, and the index ${i_{m,k}}$ uniquely exists for simplicity. This assumption will be further demonastrated in Appendix \ref{app:rem}. It is readily proved that ${i_{m,1}} \le  \cdots  \le {i_{m,K}}$, and the uniqueness of ${i_{m,k}}$ means
\begin{equation}\label{eqn:theta_mk_upp}
\theta _{m,k}^* < \frac{{\zeta _{m,i}^*}}{{{2^{{R_{m,i}}}} - 1}} + \sum\nolimits_{l = 1}^{i - 1} {\zeta _{m,l}^*},\,i\in[k,K] \wedge  i \ne {i_{m,k}}.
\end{equation}
Accordingly, (\ref{eqn:kkt_thetareq}) indicates
\begin{equation}\label{eqn:omega_1sp}
  \omega _{m,k,i \ne {i_{m,k}}}^* = 0.
\end{equation}

Besides, by setting $k=K$ in (\ref{eqn:kkt2_firthetare}), we have
\begin{equation}\label{eqn:omega_K_K}
\omega _{m,K,K}^* = \left( {{N_r} - M + 1} \right){\phi _{m,K}}{R_{m,K}}{\theta_{m,K}^*} ^{ - \left( {{N_r} - M + 2} \right)}.
\end{equation}
Since $\theta _{m,k}^*$ is upper bounded from (\ref{eqn:theta_mk_upp}), it follows from (\ref{eqn:omega_K_K}) that $\omega _{m,K,K}^* > 0$. Then setting $k=K$ in (\ref{eqn:kkt1_firzetare}) leads to
\begin{align}\label{eqn:eta_m_bnd}
\eta _m^* &= \frac{1}{{{2^{{R_{m,K}}}} - 1}}\sum\limits_{t \in \left[ {1,K} \right]} {\omega _{m,t,K}^*} \ge \frac{1}{{{2^{{R_{m,K}}}} - 1}}\omega _{m,K,K}^* > 0.
\end{align}
Combining (\ref{eqn:kkt1_firzetare}) and (\ref{eqn:omega_1sp}) results in
\begin{equation}\label{eqn:eta_m_gz}
\eta _m^* = \frac{1}{{{2^{{R_{m,k}}}} - 1}}\sum\limits_{t \in \left[ {1,k} \right]} {\omega _{m,t,k}^*}  - \sum\limits_{t \in \left[ {k + 1,K} \right]} {\omega _{m,t,{i_{m,t}}}^*}.
\end{equation}
By setting $k=1$, it follows from (\ref{eqn:eta_m_gz}) that
\begin{equation}\label{eqn:eta_m_k1}
\eta _m^* = \frac{1}{{{2^{{R_{m,1}}}} - 1}}\omega _{m,1,1}^* - \sum\limits_{t \in \left[ {2,K} \right]} {\omega _{m,t,{i_{m,t}}}^*} .
\end{equation}
With (\ref{eqn:eta_m_bnd}), (\ref{eqn:eta_m_k1}) indicates that $\omega _{m,1,1}^* > 0$. According to (\ref{eqn:omega_1sp}), we deduce that $i_{m,1} = 1$, i.e., $\omega _{m,1,i \ne 1}^* = 0$. Then setting $k=2$ in (\ref{eqn:eta_m_gz}), we have
\begin{equation}\label{eqn:eta_m_k2}
\eta _m^* = \frac{1}{{{2^{{R_{m,2}}}} - 1}}\left( {\omega _{m,1,2}^* + \omega _{m,2,2}^*} \right) - \sum\limits_{t \in \left[ {3,K} \right]} {\omega _{m,t,{i_{m,t}}}^*}.
\end{equation}
Since $\omega _{m,1,2}^* = 0$, (\ref{eqn:eta_m_k2}) suggests that $\omega _{m,2,2}^* > 0$. Accordingly, $i_{m,2} = 2$ and $\omega _{m,2,i \ne 2}^* = 0$ follow from (\ref{eqn:omega_1sp}). Similarly, by successively setting $k=3, \cdots, K$ in (\ref{eqn:eta_m_gz}), we have ${i_{m,k}}=k$. Therefore, from (\ref{eqn:kkt_thetareq}), it follows that
\begin{equation}\label{eqn:theta_fina}
\theta _{m,k}^* = \frac{{\zeta _{m,k}^*}}{{{2^{{R_{m,k}}}} - 1}} - \sum\limits_{l = 1}^{k - 1} {\zeta _{m,l}^*}.
\end{equation}
By combining (\ref{eqn:theta_mk_upp}) with (\ref{eqn:theta_fina}), we arrive at $\theta _{m,1}^* < \cdots < \theta _{m,K}^*$. Meanwhile, (\ref{eqn:eta_m_gz}) and (\ref{eqn:kkt2_firthetare}) respectively reduce to
\begin{equation}\label{eqn:eta_m_rew}
\eta _m^* = \frac{1}{{{2^{{R_{m,k}}}} - 1}}\omega _{m,k,k}^* - \sum\limits_{t \in \left[ {k + 1,K} \right]} {\omega _{m,t,t}^*},
\end{equation}
\begin{multline}\label{eqn:zeta_rew}
 - \left( {{N_r} - M + 1} \right){\phi _{m,k}}{R_{m,k}}{\theta_{m,k}^*} ^{ - \left( {{N_r} - M + 2} \right)} \\
 + \omega _{m,k,k}^* = 0.
\end{multline}
By solving the system of equations constituted by (\ref{eqn:eta_m_rew}) for $k=1, \cdots, K$, ${{{{ \bs\omega }}}_m} = {\left( {{{ \omega }_{m,1,1}^*}, \cdots ,{{ \omega }_{m,K,K}^*}} \right)^{\rm{T}}}$ can be obtained as
\begin{equation}\label{eqn:omega_v}
{{{{ \bs\omega }}}_m} = \eta _m^*{{\bf{U}}_m}^{ - 1}{{\bf{1}}_K},
\end{equation}
where ${\bf 1}_K$ is $K\times 1$ vector of ones and ${{\bf{U}}_m}$ are defined in (\ref{eqn:L_def}). We define ${\bf e}_k$ as a column vector with one as the $k$-th element and zeros elsewhere, and ${{ \omega }_{m,k,k}^*}$ can then be represented by
\begin{equation}\label{eqn:omega_k_eq}
{{ \omega }_{m,k,k}^*} = \eta _m^*{{\bf e}_k}^{\rm T} {{\bf{U}}_m}^{ - 1}{{\bf{1}}_K}.
\end{equation}
 %  Furthermore, it can be proved by contradiction that $\eta _m^*$ could not be zero. Specifically, if $\eta _m^*=0$, (\ref{eqn:kkt1_firzetare}) indicates $\omega _{m,k,i }^* = 0$, then (\ref{eqn:kkt2_firthetare}) reduces to
%\begin{equation}\label{eqn:theta_kktfirt_theta}
% - \left( {{N_r} - M + 1} \right){\phi _{m,k}}{R_{m,k}}{\theta_{m,k}^*} ^{ - \left( {{N_r} - M + 2} \right)} = 0.
%\end{equation}
%The solution to (\ref{eqn:theta_kktfirt_theta}) is therefore given by ${\theta_{m,k}^*} = \infty$, which contradicts (\ref{eqn:theta_mk_upp}). Hence $\eta _m^*>0$, (\ref{eqn:kkt1_firzetare}) can further be written as
%\begin{equation}\label{eqn:eta_omega}
%\eta _m^* = \frac{1}{{{2^{{R_{m,k}}}} - 1}}\omega _{m,k,k}^* - \sum\limits_{k \in \left[ {1,K} \right]\hfill\atop
%i \in \left[ {k + 1,K} \right]\hfill} {\omega _{m,k,i}^*}>0.
%\end{equation}
%With (\ref{eqn:omega_1sp}), we have ${i_{m,k}}=k$ and
%\begin{equation}\label{eqn:eta_omega_re}
%\eta _m^* = \frac{1}{{{2^{{R_{m,k}}}} - 1}}\omega _{m,k,k}^*.
%\end{equation}

Substituting (\ref{eqn:theta_fina}) and (\ref{eqn:omega_k_eq}) into (\ref{eqn:zeta_rew}) gives rise to
\begin{multline}\label{eqn:zeta_eq}
 - \left( {{N_r} - M + 1} \right){\phi _{m,k}}{R_{m,k}}\times\\
 {\left( {\frac{{\zeta _{m,k}^*}}{{{2^{{R_{m,k}}}} - 1}} - \sum\limits_{l = 1}^{k - 1} {\zeta _{m,l}^*} } \right)^{ - \left( {{N_r} - M + 2} \right)}} \\
 + \eta _m^*{{\bf e}_k}^{\rm T} {{\bf{U}}_m}^{ - 1}{{\bf{1}}_K} = 0.
\end{multline}
By defining ${\tilde\zeta _{m,k}^*} = {{ \zeta }_{m,k}^*}{\eta _m^*}^{\frac{1}{{{N_r} - M + 2}}}$, (\ref{eqn:zeta_eq}) can be rewritten as
\begin{multline}\label{eqn:zeta_eq_sys}
\frac{{\tilde \zeta _{m,k}^*}}{{{2^{{R_{m,k}}}} - 1}} - \sum\limits_{l = 1}^{k - 1} {\tilde \zeta _{m,l}^*}  \\
= {\left( {\frac{{\left( {{N_r} - M + 1} \right){\phi _{m,k}}{R_{m,k}}}}{{{\bf e}_k}^{\rm T} {{\bf{U}}_m}^{ - 1}{{\bf{1}}_K}}} \right)^{\frac{1}{{{N_r} - M + 2}}}}.
\end{multline}
By solving the system of equations in (\ref{eqn:zeta_eq_sys}), ${{{\tilde{ \bs\zeta }}}_m} = {\left( {{{\tilde \zeta }_{m,1}^*}, \cdots ,{{\tilde \zeta }_{m,K}^*}} \right)^{\rm{T}}}$ can be derived as
\begin{equation}\label{eqn:zeta_L_b}
{{{\tilde{ \bs\zeta }}}_m} = {{\bf{L}}_m}^{ - 1}{{\bf{b}}_m},
\end{equation}
where ${{\bf{b}}_m}$ and ${{\bf{L}}_m}$ are defined in (\ref{eqn:b_def}) and (\ref{eqn:L_def}), respectively.

Substituting (\ref{eqn:zeta_L_b}) into (\ref{eqn:kkt_powercon}) along with the definition of ${{\tilde \zeta }_{m,k}^*}$, ${\eta _m^*}^{\frac{1}{{{N_r} - M + 2}}}$ is given by
\begin{equation}\label{eqn:eta_fin}
{\eta _m^*}^{\frac{1}{{{N_r} - M + 2}}} = {{M{{\bf{1}}_K}^{\rm{T}}{{\bf{L}}_m}^{ - 1}{{\bf{b}}_m}}},
\end{equation}
Defining by ${{{{ \bs\zeta }}}_m^*} = {\left( {{{ \zeta }_{m,1}^*}, \cdots ,{{ \zeta }_{m,K}^*}} \right)^{\rm{T}}}$ the vector of power allocation coefficients, ${{{{ \bs\zeta }}}_m^*}$ is finally expressed as (\ref{eqn:zeta_fina}). With (\ref{eqn:theta_fina}), ${\theta _{m,k}^*}$ can be expressed as (\ref{eqn:theta_fin}).

\section{Proof of Remark \ref{the:rem}}\label{app:rem}
${{\bf{L}}_m}$ is a lower triangular matrix and its matrix inverse is still a lower triangular matrix. Specifically, ${{\bf{L}}_m}^{ - 1}$ can be derived as
\begin{equation}\label{eqn:L_inv}
{{\bf{L}}_m}^{ - 1} = \left( {{{\left[ {\prod\limits_{t \in \left\{ {i,j} \right\}} {\left( {{2^{{R_{m,t}}}} - 1} \right)} \prod\limits_{t = j + 1}^{i - 1} {{2^{{R_{m,t}}}}} } \right]}_{K \ge i \ge j \ge 1}}} \right),
\end{equation}
where all entries of ${{\bf{L}}_m}^{ - 1}$ with $i<j$ are zero. By using (\ref{eqn:L_inv}), ${{\bf{L}}_m}^{ - 1}{{\bf{b}}_m}$
 can be explicitly expressed as (\ref{eqn:Lminvb}) at the top of the next page,
 \begin{figure*}[!t]
\begin{equation}\label{eqn:Lminvb}
{{\bf{L}}_m}^{ - 1}{{\bf{b}}_m} = \left( {\begin{array}{*{20}{c}}
 \vdots \\
{\underbrace {\sum\nolimits_{j = 1}^i {\prod\nolimits_{t \in \left\{ {i,j} \right\}} {\left( {{2^{{R_{m,t}}}} - 1} \right)} \prod\nolimits_{t = j + 1}^{i - 1} {{2^{{R_{m,t}}}}} } {{\left( {\frac{{\left( {{N_r} - M + 1} \right){\phi _{m,j}}{R_{m,j}}}}{{{{\bf e}_j}^{\rm T} {{\bf{U}}_m}^{ - 1}{{\bf{1}}_K}}}} \right)}^{\frac{1}{{{N_r} - M + 2}}}}}_{{\beta_{m,i}},1 \le i \le K}}\\
 \vdots
\end{array}} \right),%\left( {{{\left[ {\sum\limits_{j = 1}^i {\prod\limits_{t \in \left\{ {i,j} \right\}} {\left( {{2^{{R_{m,t}}}} - 1} \right)} \prod\limits_{t = j + 1}^{i - 1} {{2^{{R_{m,t}}}}} } {{\left( {\frac{{\left( {{N_r} - M + 1} \right){\phi _{m,j}}{R_{m,j}}}}{{{2^{{R_{m,j}}}} - 1}}} \right)}^{\frac{1}{{{N_r} - M + 2}}}}} \right]}_{1 \le i \le K}}} \right),
\end{equation}
%\hrulefill
%\vspace*{4pt}
\end{figure*}
where $\beta_{m,i}>0$. With (\ref{eqn:Lminvb}), ${{{\bf{1}}_K}^{\rm{T}}{{\bf{L}}_m}^{ - 1}{{\bf{b}}_m}}$ and $ {{{\bf{a}}_{m,k}}^{\rm{T}}{{\bf{L}}_m}^{ - 1}{{\bf{b}}_m}}$ can be written as
\begin{equation}\label{eqn:power_num}
{{{\bf{1}}_K}^{\rm{T}}{{\bf{L}}_m}^{ - 1}{{\bf{b}}_m}}=\sum\limits_{i = 1}^K {{\beta _{m,i}}},
\end{equation}
\begin{equation}\label{eqn:dom_out}
 {{{\bf{a}}_{m,k}}^{\rm{T}}{{\bf{L}}_m}^{ - 1}{{\bf{b}}_m}}=  \frac{{{\beta_{m,k}}}}{{{2^{{R_{m,k}}}} - 1}} - \sum\limits_{i = 1}^{k - 1} {{\beta_{m,i}}}. %\frac{1}{{{2^{{R_{m,k}}}} - 1}}\sum\limits_{j = 1}^k {\prod\limits_{t \in \left\{ {k,j} \right\}} {\left( {{2^{{R_{m,t}}}} - 1} \right)} \prod\limits_{t = j + 1}^{k - 1} {{2^{{R_{m,t}}}}} } {\left( {\frac{{\left( {{N_r} - M + 1} \right){\phi _{m,j}}{R_{m,j}}}}{{{2^{{R_{m,j}}}} - 1}}} \right)^{\frac{1}{{{N_r} - M + 2}}}}  \\
%- \sum\limits_{i = 1}^{k - 1} {\sum\limits_{j = 1}^i {\prod\limits_{t \in \left\{ {i,j} \right\}} {\left( {{2^{{R_{m,t}}}} - 1} \right)} \prod\limits_{t = j + 1}^{i - 1} {{2^{{R_{m,t}}}}} } {{\left( {\frac{{\left( {{N_r} - M + 1} \right){\phi _{m,j}}{R_{m,j}}}}{{{2^{{R_{m,j}}}} - 1}}} \right)}^{\frac{1}{{{N_r} - M + 2}}}}},
\end{equation}
In order to prove that the sequence $\theta _{m,1}^*,\cdots,\theta _{m,K}^*$ is in an ascending order of $k$, i.e., $\theta _{m,k+1}^* > \theta _{m,k}^*$. From (\ref{eqn:theta_fin}), it suffices to show that  %the following inequality
\begin{equation}\label{eqn:theta_ine_proof}
\frac{{{\beta_{m,k + 1}}}}{{{2^{{R_{m,k + 1}}}} - 1}} - \sum\limits_{i = 1}^k {{\beta_{m,i}}}  > \frac{{{\beta_{m,k}}}}{{{2^{{R_{m,k}}}} - 1}} - \sum\limits_{i = 1}^{k - 1} {{\beta_{m,i}}} .
\end{equation}
This amounts to proving ${{{\beta_{m,k + 1}}}}/{({{2^{{R_{m,k + 1}}}} - 1})} > {{{2^{{R_{m,k}}}}{\beta_{m,k}}}}/{({{2^{{R_{m,k}}}} - 1})}$, which can be justified by substituting the definition of $\beta_{m,k}$ into (\ref{eqn:Lminvb}). This result is consistent with the assumption of (\ref{eqn:theta_mk_upp}). By using (\ref{eqn:theta_fina}) together with $\theta _{m,1}^* < \cdots < \theta _{m,K}^*$, it follows that
\begin{equation}\label{eqn:power_normal}
\frac{{\zeta _{m,k}^*}}{{{2^{{R_{m,k}}}} - 1}} < \frac{{\zeta _{m,k + 1}^*}}{{{2^{{R_{m,k + 1}}}} - 1}} - \zeta _{m,k}^* < \frac{{\zeta _{m,k + 1}^*}}{{{2^{{R_{m,k + 1}}}} - 1}}.
\end{equation}
This leads to (\ref{eqn:normal_power}).
\section{Proof of Remark \ref{the:remout}}\label{app:remout}
By using (\ref{eqn:power_num}) and (\ref{eqn:dom_out}), we have
\begin{align}\label{eqn:out_bound}
\frac{{{{\bf{1}}_K}^{\rm{T}}{{\bf{L}}_m}^{ - 1}{{\bf{b}}_m}}}{{{{\bf{a}}_{m,k}}^{\rm{T}}{{\bf{L}}_m}^{ - 1}{{\bf{b}}_m}}} &= \frac{{\sum\nolimits_{i = 1}^K {{d_{m,i}}} }}{{\frac{{{d_{m,k}}}}{{{2^{{R_{m,k}}}} - 1}} - \sum\nolimits_{i = 1}^{k - 1} {{d_{m,i}}} }} \ge \frac{{\sum\nolimits_{i = 1}^K {{d_{m,i}}} }}{{\frac{{{d_{m,k}}}}{{{2^{{R_{m,k}}}} - 1}}}}\notag\\
 &\ge \frac{{{d_{m,k}}}}{{\frac{{{d_{m,k}}}}{{{2^{{R_{m,k}}}} - 1}}}} = {2^{{R_{m,k}}}} - 1.
\end{align}
Hence, the asymptotic outage probability in (\ref{eqn:out_tar}) is lower bounded by (\ref{eqn:out_opt_lower}).

\section{Proof of Theorem \ref{the:rate_opt}}\label{app:rate_opt}
The associated Lagrangian is given by (\ref{eqn:Lag_bar}) at the top of this page.
\begin{figure*}[!t]
\begin{multline}\label{eqn:Lag_bar}
\bar {\mathcal L}\left( {{{\left\{ {{R_{m,k}}} \right\}}_{m \in \left[ {1,M} \right]\hfill\atop
k \in \left[ {1,K} \right]\hfill}},{{\left\{ {{\theta _{m,k}}} \right\}}_{m \in \left[ {1,M} \right]\hfill\atop
k \in \left[ {1,K} \right]\hfill}},{{\left\{ {{\mathchar'26\mkern-10mu\lambda _{m,k}}} \right\}}_{m \in \left[ {1,M} \right]\hfill\atop
k \in \left[ {1,K} \right]\hfill}},{{\left\{ {{\iota _{m,k}}} \right\}}_{m \in \left[ {1,M} \right]\hfill\atop
k \in \left[ {1,K} \right]\hfill}},{{\left\{ {{\varpi _{m,k,i}}} \right\}}_{m \in \left[ {1,M} \right]\hfill\atop
{k \in \left[ {1,K} \right]\hfill\atop
i \in \left[ {k,K} \right]\hfill}}}} \right)\\
 =  - \sum\nolimits_{m \in \left[ {1,M} \right]\hfill\atop
k \in \left[ {1,K} \right]\hfill} {\left( {1 - {\phi _{m,k}}{\theta _{m,k}}^{ - \left( {{N_r} - M + 1} \right)}} \right){R_{m,k}}}  - \sum\nolimits_{m \in \left[ {1,M} \right]\hfill\atop
k \in \left[ {1,K} \right]\hfill} {{\mathchar'26\mkern-10mu\lambda _{m,k}}{R_{m,k}}}  - \sum\nolimits_{m \in \left[ {1,M} \right]\hfill\atop
k \in \left[ {1,K} \right]\hfill} {{\iota _{m,k}}{\theta _{m,k}}}  \\
+ \sum\nolimits_{m \in \left[ {1,M} \right]\hfill\atop
{k \in \left[ {1,K} \right]\hfill\atop
i \in \left[ {k,K} \right]\hfill}} {{\varpi _{m,k,i}}\left( {{R_{m,i}} - {{\log }_2}\left( {1 + \frac{{{\zeta _{m,i}}}}{{{\theta _{m,k}} + \sum\nolimits_{l = 1}^{i - 1} {{\zeta _{m,l}}} }}} \right)} \right)},
\end{multline}
\hrulefill
%\vspace*{4pt}
\end{figure*}
The necessary KKT conditions are given as %\eqref{eqn:lamba_R_kkt}-\eqref{eqn:kkt_bar_theta},
\begin{equation}\label{eqn:lamba_R_kkt}
\mathchar'26\mkern-10mu\lambda _{m,k}^*R_{m,k}^* = 0,
\end{equation}
\begin{equation}\label{eqn:iota_theta_kkt}
\iota _{m,k}^*\theta _{m,k}^* = 0,
\end{equation}
\begin{equation}\label{eqn:vapi_kkt}
\varpi _{m,k,i}^*\left( {R_{m,i}^* - {{\log }_2}\left( {1 + \frac{{{\zeta _{m,i}}}}{{\theta _{m,k}^* + \sum\limits_{l = 1}^{i - 1} {{\zeta _{m,l}}} }}} \right)} \right) = 0,
\end{equation}
\eqref{eqn:kkt_bar_R} and \eqref{eqn:kkt_bar_theta} at the top of the next page.
\begin{figure*}[!t]
\begin{equation}\label{eqn:kkt_bar_R}
{\left. {\frac{{\partial \bar {\mathcal L}}}{{\partial {R_{m,k}}}}} \right|_{{{\left\{ {R_{m,k}^*} \right\}}_{m \in \left[ {1,M} \right]\hfill\atop
k \in \left[ {1,K} \right]\hfill}},{{\left\{ {\theta _{m,k}^*} \right\}}_{m \in \left[ {1,M} \right]\hfill\atop
k \in \left[ {1,K} \right]\hfill}},{{\left\{ {\mathchar'26\mkern-10mu\lambda _{m,k}^*} \right\}}_{m \in \left[ {1,M} \right]\hfill\atop
k \in \left[ {1,K} \right]\hfill}},{{\left\{ {\iota _{m,k}^*} \right\}}_{m \in \left[ {1,M} \right]\hfill\atop
k \in \left[ {1,K} \right]\hfill}},{{\left\{ {\varpi _{m,k,i}^*} \right\}}_{m \in \left[ {1,M} \right]\hfill\atop
{k \in \left[ {1,K} \right]\hfill\atop
i \in \left[ {k,K} \right]\hfill}}}}} = 0,
\end{equation}
%\hrulefill
%\vspace*{4pt}
\end{figure*}
\begin{figure*}[!t]
\begin{equation}\label{eqn:kkt_bar_theta}
{\left. {\frac{{\partial \bar {\mathcal L}}}{{\partial {\theta _{m,k}}}}} \right|_{{{\left\{ {R_{m,k}^*} \right\}}_{m \in \left[ {1,M} \right]\hfill\atop
k \in \left[ {1,K} \right]\hfill}},{{\left\{ {\theta _{m,k}^*} \right\}}_{m \in \left[ {1,M} \right]\hfill\atop
k \in \left[ {1,K} \right]\hfill}},{{\left\{ {\mathchar'26\mkern-10mu\lambda _{m,k}^*} \right\}}_{m \in \left[ {1,M} \right]\hfill\atop
k \in \left[ {1,K} \right]\hfill}},{{\left\{ {\iota _{m,k}^*} \right\}}_{m \in \left[ {1,M} \right]\hfill\atop
k \in \left[ {1,K} \right]\hfill}},{{\left\{ {\varpi _{m,k,i}^*} \right\}}_{m \in \left[ {1,M} \right]\hfill\atop
{k \in \left[ {1,K} \right]\hfill\atop
i \in \left[ {k,K} \right]\hfill}}}}} = 0,
\end{equation}
\hrulefill
%\vspace*{4pt}
\end{figure*}
Substituting (\ref{eqn:Lag_bar}) into (\ref{eqn:kkt_bar_R}) and (\ref{eqn:kkt_bar_theta}) respectively leads to
\begin{multline}\label{eqn:barR_kkt_sim}
 -  \left({1 - \phi _{m,k}\theta {{_{m,k}^*}^{ - \left( {{N_r} - M + 1} \right)}}}\right)  - \mathchar'26\mkern-10mu\lambda _{m,k}^* \\
 + \sum\nolimits_{t \in \left[ {1,k} \right]} {\varpi _{m,t,k}^*}  = 0,
\end{multline}
\begin{align}\label{eqn:theta_kkt_bar_sim}
- \left( {{N_r} - M + 1} \right)\phi _{m,k}{\theta _{m,k}^*}^{ - \left( {{N_r} - M + 2} \right)}R_{m,k}^* - \iota _{m,k}^* + \notag\\
  \sum\limits_{i \in \left[ {k,K} \right]} {\frac{{{\varpi _{m,k,i}^*\zeta _{m,i}}}}{{\left( {\theta _{m,k}^* + \sum\limits_{l = 1}^{i - 1} {{\zeta _{m,l}}} } \right)\left( {\theta _{m,k}^* + \sum\limits_{l = 1}^i {{\zeta _{m,l}}} } \right)\ln2}}} \notag\\
   = 0.
\end{align}
Note that $R_{m,k}^*, \theta _{m,k}^* >0$, it follows from (\ref{eqn:lamba_R_kkt}) and (\ref{eqn:iota_theta_kkt}) that $\mathchar'26\mkern-10mu\lambda _{m,k}^*$ and $\iota _{m,k}^*=0$. Thus, (\ref{eqn:barR_kkt_sim}) and (\ref{eqn:theta_kkt_bar_sim}) can be rewritten as
\begin{equation}\label{eqn:barR_kkt_simre}
 {1 - \phi _{m,k}\theta {{_{m,k}^*}^{ - \left( {{N_r} - M + 1} \right)}}}  = \sum\limits_{t \in \left[ {1,k} \right]} {\varpi _{m,t,k}^*},
\end{equation}
\begin{align}\label{eqn:theta_kkt_bar_simre}
 \left( {{N_r} - M + 1} \right)\phi _{m,k}{\theta _{m,k}^*}^{ - \left( {{N_r} - M + 2} \right)}R_{m,k}^* = \notag\\
  \sum\limits_{i \in \left[ {k,K} \right]} {\frac{{{\varpi _{m,k,i}^*\zeta _{m,i}}}}{{\left( {\theta _{m,k}^* + \sum\limits_{l = 1}^{i - 1} {{\zeta _{m,l}}} } \right)\left( {\theta _{m,k}^* + \sum\limits_{l = 1}^i {{\zeta _{m,l}}} } \right)\ln2}}}.
\end{align}
For further simplification, we define ${{\bar i}_{m,k}} = \mathop {\arg }\limits_i \min \left\{ {{{\log }_2}\left( {1 + {{{\zeta _{m,i}}}}/{({\theta _{m,k}^* + \sum\nolimits_{l = 1}^{i - 1} {{\zeta _{m,l}}} })}} \right),i \in \left[ {k,K} \right]} \right\}$ and we assume that ${{\bar i}_{m,k}}$ uniquely exists. Hence, we have
\begin{equation}\label{eqn:R_ineq}
{R_{m,i}^* < {{\log }_2}\left( {1 + \frac{{{\zeta _{m,i}}}}{{\theta _{m,k}^* + \sum\limits_{l = 1}^{i - 1} {{\zeta _{m,l}}} }}} \right)},\, i\in[k,K] \wedge i \ne {{\bar i}_{m,k}}.
\end{equation}
According to (\ref{eqn:vapi_kkt}), we deduce that $\varpi _{m,k,i \ne {{\bar i}_{m,k}}}^* = 0$. Moreover, by setting $k=1$ in (\ref{eqn:barR_kkt_simre}), we have
\begin{equation}\label{eqn:barR_kkt_simre1}
 {1 - \phi _{m,1}\theta {{_{m,1}^*}^{ - \left( {{N_r} - M + 1} \right)}}}  = {\varpi _{m,1,1}^*},
\end{equation}
It is well known that the outage probability without perfect CSI at the BS is nonzero unless no information is conveyed. Hence, it is reasonable to assume that $ p_{m,1}^{out\_asy} = {1 - \phi _{m,1}\theta {{_{m,1}^*}^{ - \left( {{N_r} - M + 1} \right)}}} > 0$. Accordingly, it follows from (\ref{eqn:barR_kkt_simre1}) that ${\varpi _{m,1,1}^*} > 0$, ${{\bar i}_{m,1}} = 1$ and $\varpi _{m,1,i \ne 1}^* = 0$. In an analogous way, it is readily proved that ${{\bar i}_{m,k}} = k$ by successively setting $k=2,3,\cdots$ in (\ref{eqn:barR_kkt_simre1}), i.e., $\varpi _{m,k,i \ne k}^* = 0$ and $\varpi _{m,k,k}^* > 0$. With the results, (\ref{eqn:barR_kkt_simre}) and (\ref{eqn:theta_kkt_bar_simre}) respectively degenerates to
\begin{equation}\label{eqn:barR_kkt_simrefur}
{1 - \phi _{m,k}\theta {{_{m,k}^*}^{ - \left( {{N_r} - M + 1} \right)}}} = {\varpi _{m,k,k}^*},
\end{equation}
\begin{align}\label{eqn:theta_kkt_bar_simrefur}
 \left( {{N_r} - M + 1} \right)\phi _{m,k}{\theta _{m,k}^*}^{ - \left( {{N_r} - M + 2} \right)}R_{m,k}^* = \notag\\
   {\frac{{{\varpi _{m,k,k}^*\zeta _{m,k}}}}{{\left( {\theta _{m,k}^* + \sum\limits_{l = 1}^{k - 1} {{\zeta _{m,l}}} } \right)\left( {\theta _{m,k}^* + \sum\limits_{l = 1}^k {{\zeta _{m,l}}} } \right)\ln2}}}.
\end{align}
Moreover, combining (\ref{eqn:vapi_kkt}) with $\varpi _{m,k,k}^* > 0$ yields (\ref{eqn:R_opt_theta}). Then putting (\ref{eqn:R_opt_theta}) and  (\ref{eqn:barR_kkt_simrefur}) into (\ref{eqn:theta_kkt_bar_simrefur}) together with some algebraic manipulations produces
\begin{equation}\label{eqn:theta_eq}
%\left( {{N_r} - M + 1} \right)\phi _{m,k}{\theta_{m,k}^*} ^{ - \left( {{N_r} - M + 2} \right)}R_{m,k}^* \\
%=\left( {1 - \phi _{m,k}{\theta_{m,k}^*} ^{ - \left( {{N_r} - M + 1} \right)}} \right)\times \\
%\frac{{{\zeta _{m,k}}}}{{\left( {\theta _{m,k}^* + \sum\limits_{l = 1}^{k - 1} {{\zeta _{m,l}}} } \right)\left( {\theta _{m,k}^* + \sum\limits_{l = 1}^k {{\zeta _{m,l}}} } \right)}}
%\frac{{{\zeta _{m,k}}{\theta {_{m,k}^*}}^{{N_r} - M + 2}\left( {1 - \phi _{m,k}{\theta {_{m,k}^*}}^{ - \left( {{N_r} - M + 1} \right)}} \right)}}{{\left( {{N_r} - M + 1} \right)\phi _{m,k}{{\log }_2}\left( {1 + \frac{{{\zeta _{m,k}}}}{{\theta _{m,k}^* + \sum\limits_{l = 1}^{k - 1} {{\zeta _{m,l}}} }}} \right)}} =\\
% \left( {\theta _{m,k}^* + \sum\limits_{l = 1}^{k - 1} {{\zeta _{m,l}}} } \right)\left( {\theta _{m,k}^* + \sum\limits_{l = 1}^k {{\zeta _{m,l}}} } \right).
 \varrho(\theta{_{m,k}^*}) = \ell(\theta{_{m,k}^*}),
\end{equation}
where $\varrho(x)$ and $\ell(x)$ are defined in (\ref{eqn:varrho_def}) and (\ref{eqn:ell_def}), respectively.

Clearly from (\ref{eqn:theta_eq}), the intersection of two curves ${\varrho(x)}$ and ${\ell(x)}$ gives the optimal value $\theta _{m,k}^*$. The curves of ${\varrho(x)}$ and ${\ell(x)}$ are plotted in Fig. \ref{fig:opt_rate}. Clearly, ${\varrho(x)}$ and ${\ell(x)}$ are both increasing functions. Additionally, it is not hard to prove that ${\varrho(x)}$ increases much faster than ${\ell(x)}$ as $x \to \infty$. Hence, the zero point of (\ref{eqn:theta_eq}), i.e., $\theta _{m,k}^*$, exists, and $\theta _{m,k}^*$ can be numerically solved by using bisection method. After determining $\theta _{m,k}^*$, $R_{m,k}^*$ can be calculated by using (\ref{eqn:R_opt_theta}).
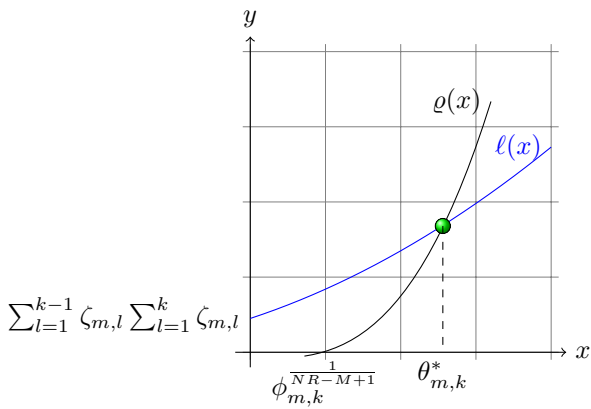
\begin{figure}
  \centering
  \begin{tikzpicture}[domain=0:5]
      \draw[very thin,color=gray] (-0.1,-0.1) grid (4.1,4.1);
      \draw[->] (-0.2,0) -- (4.2,0) node[right] {$x$};
      \draw[->] (0,-0.2) -- (0,4.2) node[above] {$y$};
      \draw[color=blue,scale=0.5,domain=0:8] plot (\x,{0.03*(\x +5)*(\x+6)}) node[left] {$\ell(x)$};
      \draw[color=black,scale=0.4,domain=1.8:8] plot (\x,{0.07*(\x^2*(1-\x^-1)/log2(1+3/(\x+0.01)))-0.2}) node[left] {$\varrho(x)$};
	%\draw[color=red,ultra thick,->] (1.91,1.33) -- (3,1.35) node[right] {$\theta _{m,k}^*$};
	  \draw[ball color=green] (2.56,1.68) circle (.1);
	 \draw[-,dashed] (2.56,1.68) -- (2.56,0) node[below] {$\theta _{m,k}^*$};
	 	\draw (1,0) node[below] { $\phi_{m,k}^{\frac{1}{NR-M+1}}$};
	\draw (0,0.5) node[left] {$\sum\nolimits_{l = 1}^{k - 1} {{\zeta _{m,l}}} \sum\nolimits_{l = 1}^k {{\zeta _{m,l}}} $};
  \end{tikzpicture}
  \caption{The increasing monotonicity of ${\varrho(x)}$ and ${\ell(x)}$.}\label{fig:opt_rate}
\end{figure}
By combining (\ref{eqn:R_opt_theta}) and (\ref{eqn:R_ineq}), the same conclusion can be drawn as Remark \ref{the:rem}, that is, the optimal transmission rates satisfy the inequalities in (\ref{eqn:normal_rate}). The proof of Theorem \ref{the:rate_opt} is thus accomplished.

\bibliographystyle{ieeetran}
\bibliography{mimo_zf_noma}
\begin{IEEEbiography}[{\includegraphics[width=1in,height=1.25in,clip,keepaspectratio]{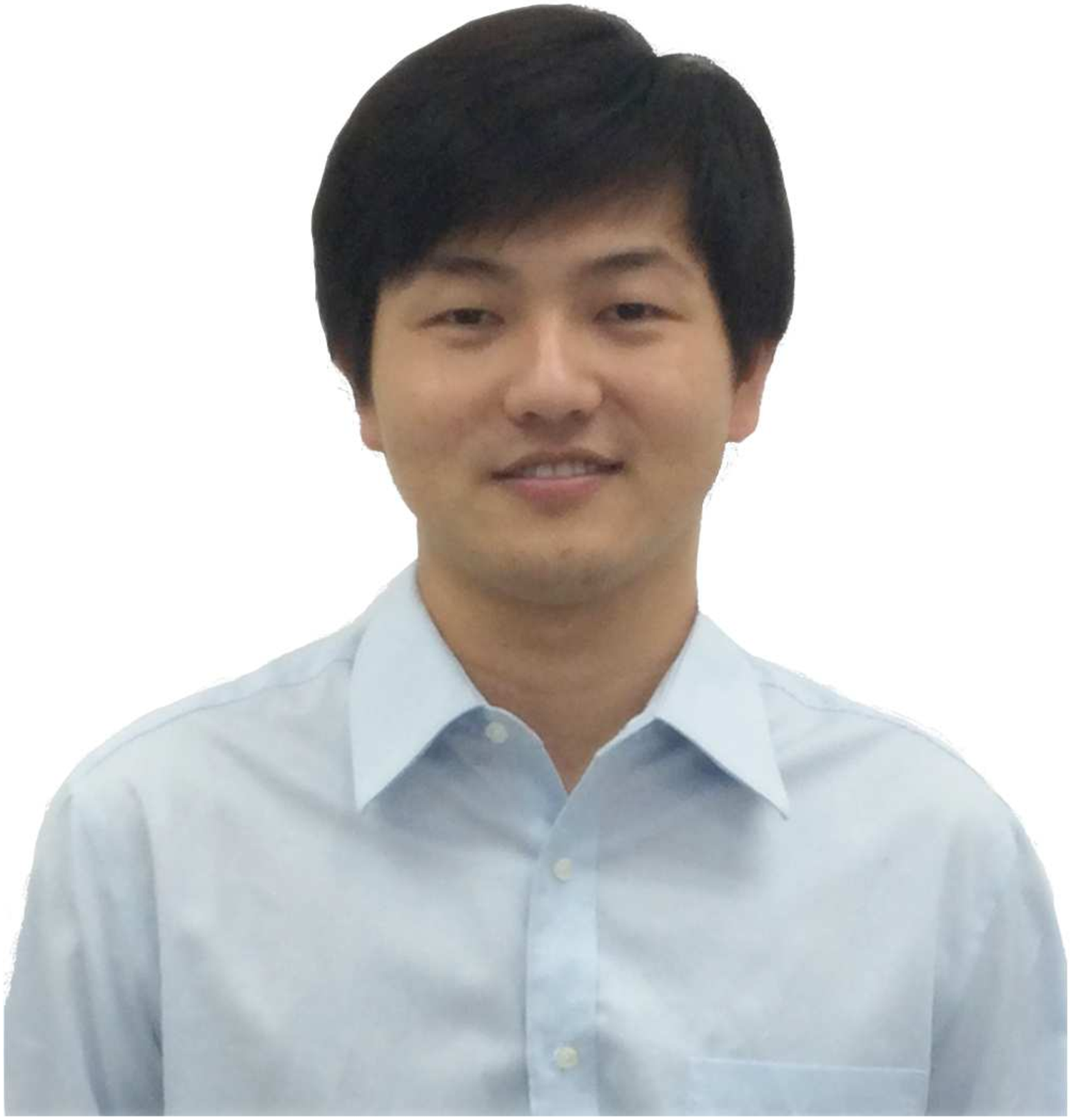}}]{Zheng Shi} received his B.S. degree in communication engineering from Anhui Normal University, China, in 2010 and his M.S. degree in communication and information system from Nanjing University of Posts and Telecommunications (NUPT), China, in 2013. He obtained his Ph.D. degree in Electrical and Computer Engineering from University of Macau, Macao, in 2017. Since 2017, he has been with the School of Intelligent Systems Science and Engineering, Jinan University and is now an Assistant Professor there. His current research interests include hybrid automatic repeat request, non-orthogonal multiple access, short-packet communications and IoT.
\end{IEEEbiography}

\begin{IEEEbiography}[{\includegraphics[width=1in,height=1.25in,clip,keepaspectratio]{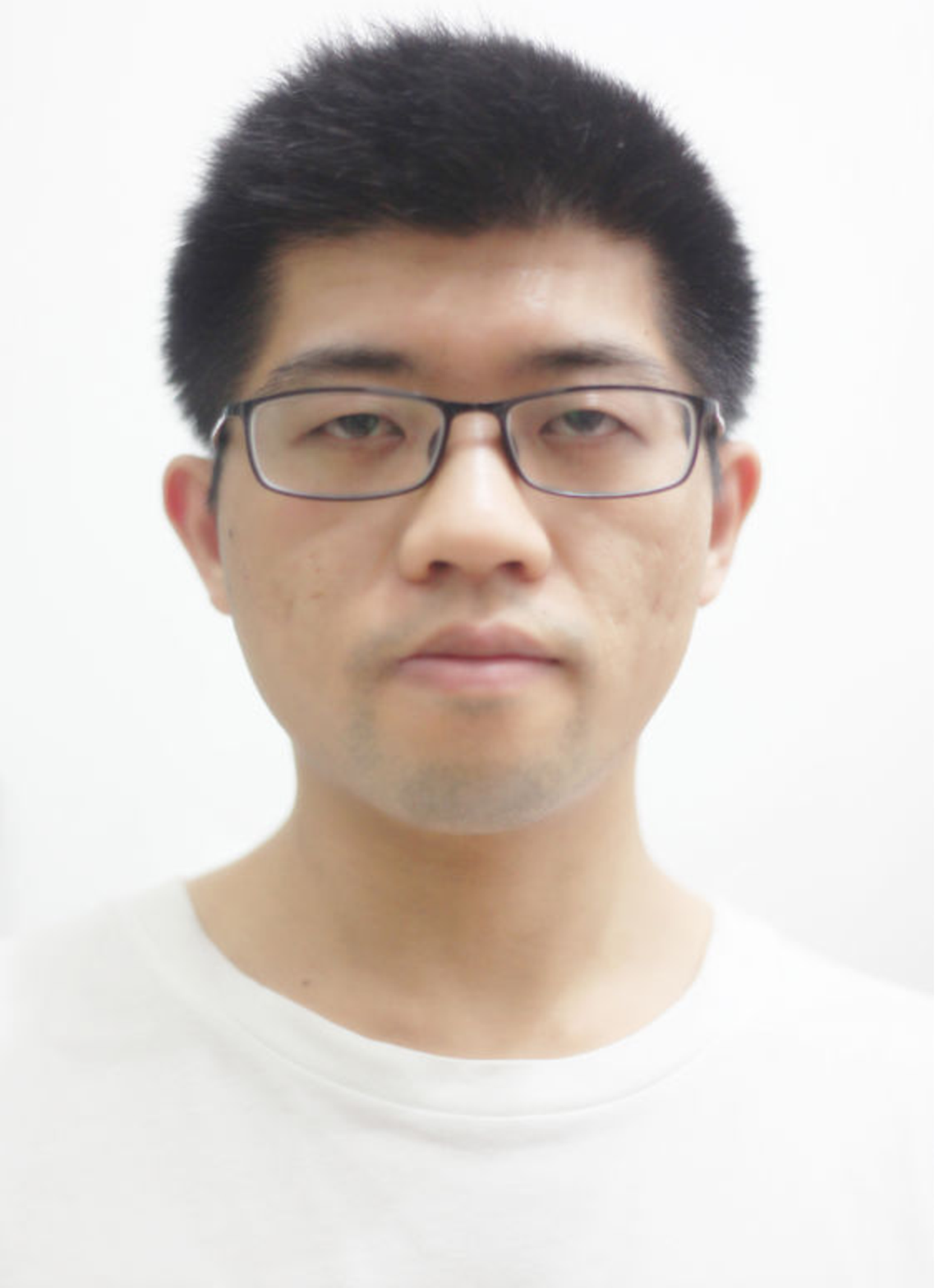}}]{Hong Wang}
received the B.S. degree in communication engineering from Jiangsu University, Zhenjiang, China, in 2011, and the Ph.D. degree in information and communication engineering from Nanjing University of Posts and Telecommunications (NUPT), Nanjing, China, in 2016. From 2014 to 2015, he was a Research Assistant with the Department of Electronic Engineering, City University of Hong Kong, Hong Kong. From 2016 to 2018, he was a Senior Research Associate with the State Key Laboratory of Millimeter Waves and Department of Electronic Engineering, City University of Hong Kong. Since July 2019, he has been an Associate Professor with the Department of Communication Engineering, NUPT. His research interests include broadband wireless communications, particularly in interference management in HetNets, non-orthogonal multiple access, and interference alignment.
\end{IEEEbiography}

\begin{IEEEbiography}[{\includegraphics[width=1in,height=1.25in,clip,keepaspectratio]{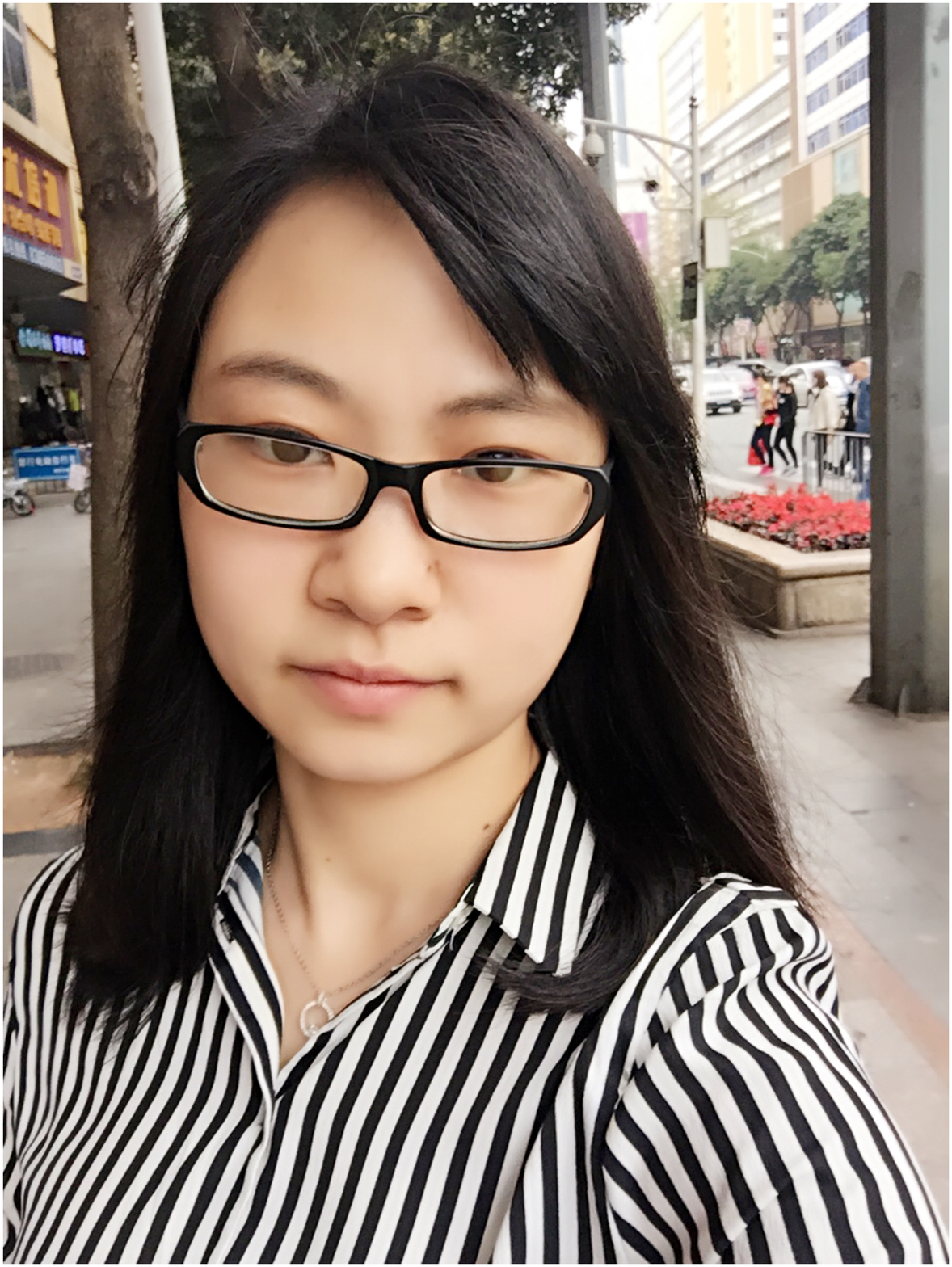}}]{Yaru Fu}
received the B.S. degree (Hons.) in Telecommunication Engineering from Northeast Electric Power University (NEEPU), China, in 2011, the M.S. degree (Hons.) in Communication and Information System from Nanjing University of Posts and Telecommunications (NUPT), China, in 2014 and the Ph.D in Electronic Engineering from City University of Hong Kong (CityU), Hong Kong SAR, China, in 2018, respectively. Then, she joined in the Institute of Network Coding (INC), The Chinese University of Hong Kong (CUHK) as a postdoc research assistant. Since Sep. 2018, she started her journey as a research fellow in Information Systems Technology and Design (ISTD) Pillar, Singapore University of Technology and Design (SUTD). Her current research interests include wireless communications and network, machine learning, recommendation system and IoT.
\end{IEEEbiography}

\begin{IEEEbiography}[{\includegraphics[width=1in,height=1.25in,clip,keepaspectratio]{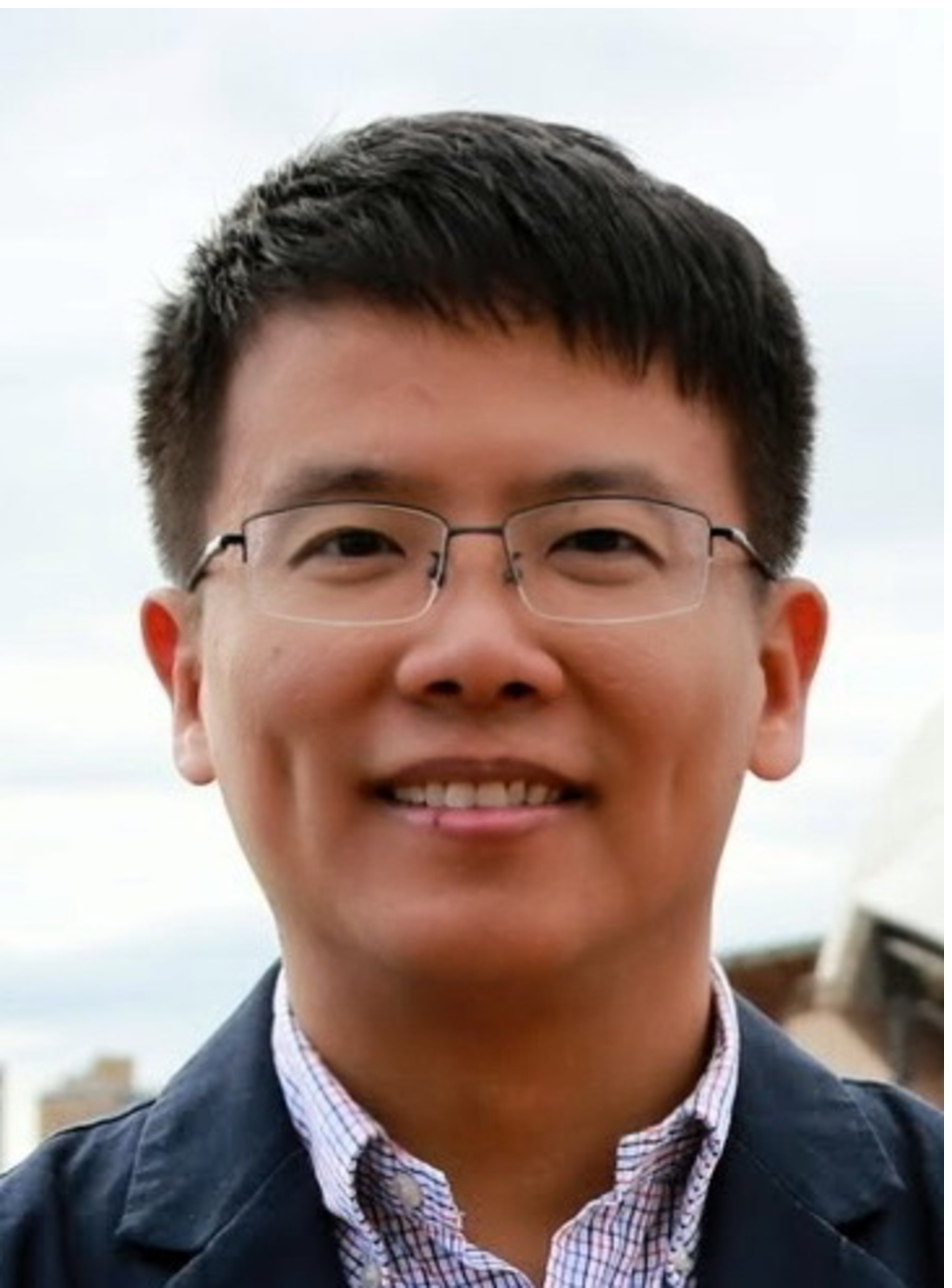}}]{Guanghua Yang}
(S'02, M'07, SM'19) received his Ph.D. degree in electrical and electronic engineering from the University of Hong Kong, Hong Kong, in 2006. From 2006 to 2013, he served as post-doctoral fellow, research associate, and project manager in the University of Hong Kong. Since April 2017, he is an Associate Professor and Associate Dean with the Institute of Physical Internet, Jinan University, Guangdong, China. His research interests are in the general areas of communications, networking and multimedia.
%received his Ph.D. degree in electrical and electronic engineering from the University of Hong Kong, Hong Kong, in 2006. From 2006 to 2013, he served as post-doctoral fellow, research associate, and project manager in the University of Hong Kong. Since April 2017, he is an Associate Professor and Associate Dean with the Institute of Physical Internet, Jinan University.  His research interests are in the general areas of communications, networking and multimedia.
\end{IEEEbiography}

\begin{IEEEbiography}[{\includegraphics[width=1in,height=1.25in,clip,keepaspectratio]{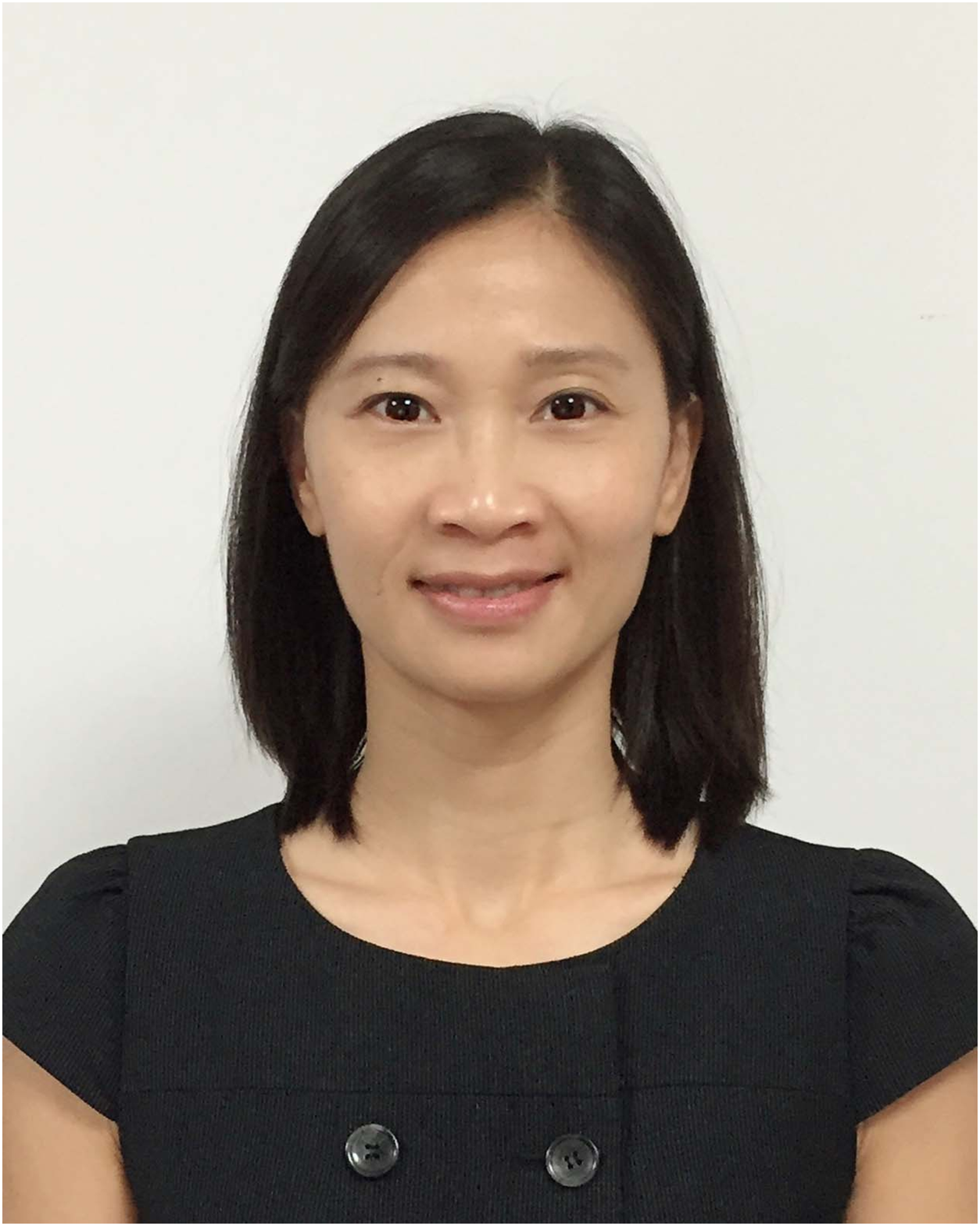}}]{Shaodan Ma}
received her double Bachelor degrees in Science and Economics, and her Master degree of Engineering, from Nankai University, Tianjin, China. She obtained her Ph. D. degree in electrical and electronic engineering from the University of Hong Kong, Hong Kong, in 2006. From 2006 to 2011, she was a Postdoctoral Fellow in the University of Hong Kong. Since August 2011, she has been with the University of Macau and is now an Associate Professor there. She was a visiting scholar in Princeton University in 2010 and is currently an Honorary Assistant Professor in the University of Hong Kong. Her research interests are in the general areas of signal processing and communications, particularly, transceiver design, resource allocation and performance analysis.
\end{IEEEbiography}

\begin{IEEEbiography}[{\includegraphics[width=1in,height=1.25in,clip,keepaspectratio]{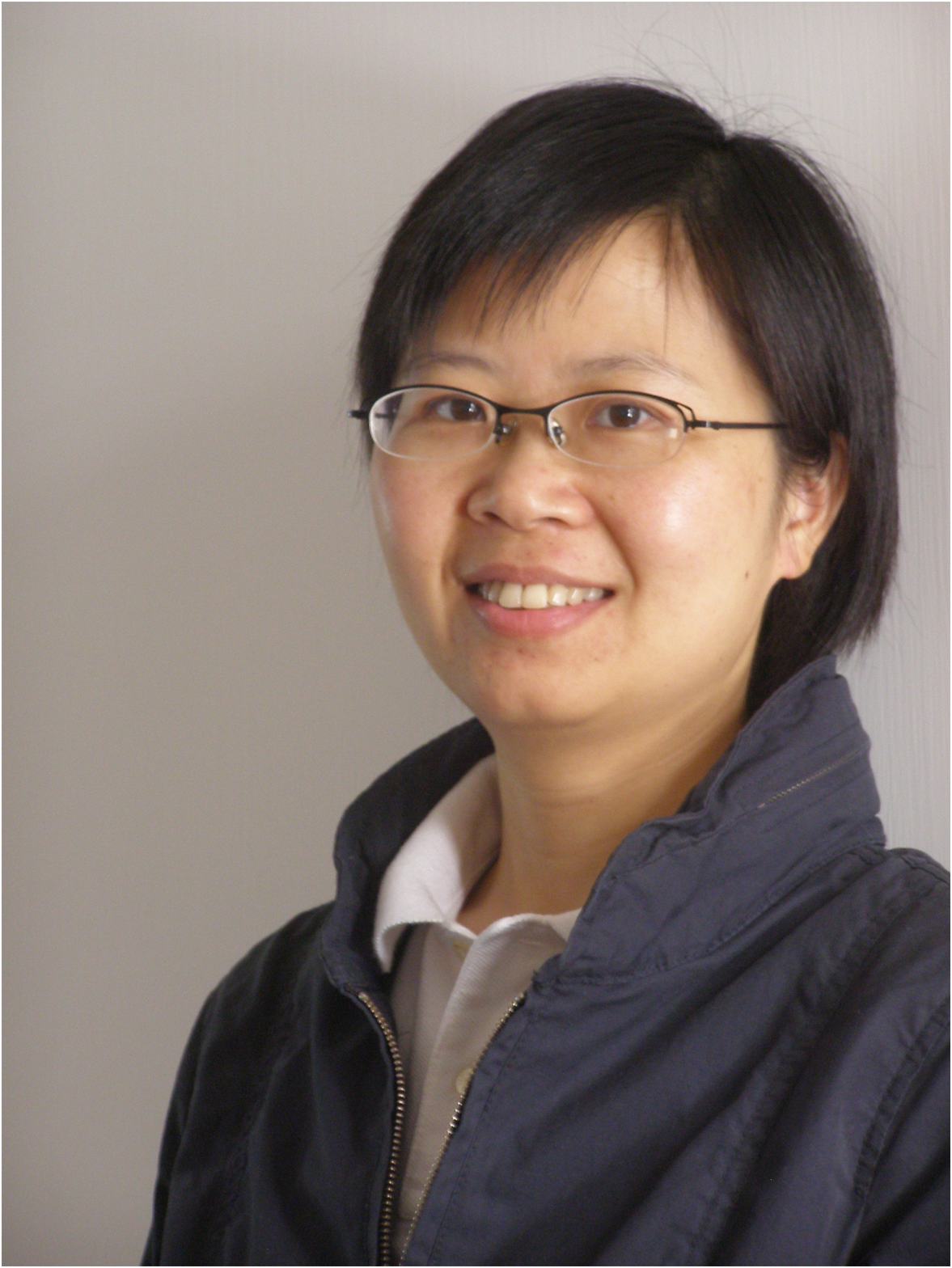}}]{Fen Hou}
received the Ph.D. degree in electrical and computer engineering from the University of Waterloo, Waterloo, Canada, in 2008. She is currently an Associate Professor with the State Sky Laboratory of Internet of Things for Smart City and the Department of Electrical and Computer Engineering at the University of Macau. Her research interests include resource allocation and scheduling in broadband wireless networks, mechanism design and optimal user behavior in mobile crowd sensing networks, and mobile data offloading. She was a recipient of the IEEE Globecom Best Paper Award in 2010 and the Distinguished Service Award in the IEEE MMTC in 2011. She served as the Co-Chair of the INFOCOM 2014 Workshop on Green Cognitive Communications and Computing Networks, the IEEE Globecom Workshop on Cloud Computing System, Networks, and Application 2013 and 2014, the ICCC 2015 Selected Topics in Communications Symposium, and the ICC 2016 Communication Software Services and Multimedia Application Symposium. She currently serves as the Director of Award Board in IEEE ComSoc Multimedia Communications Technical Committee. She also serves as an Associate Editor of IET Communications.
\end{IEEEbiography}

\begin{IEEEbiography}[{\includegraphics[width=1in,height=1.25in,clip,keepaspectratio]{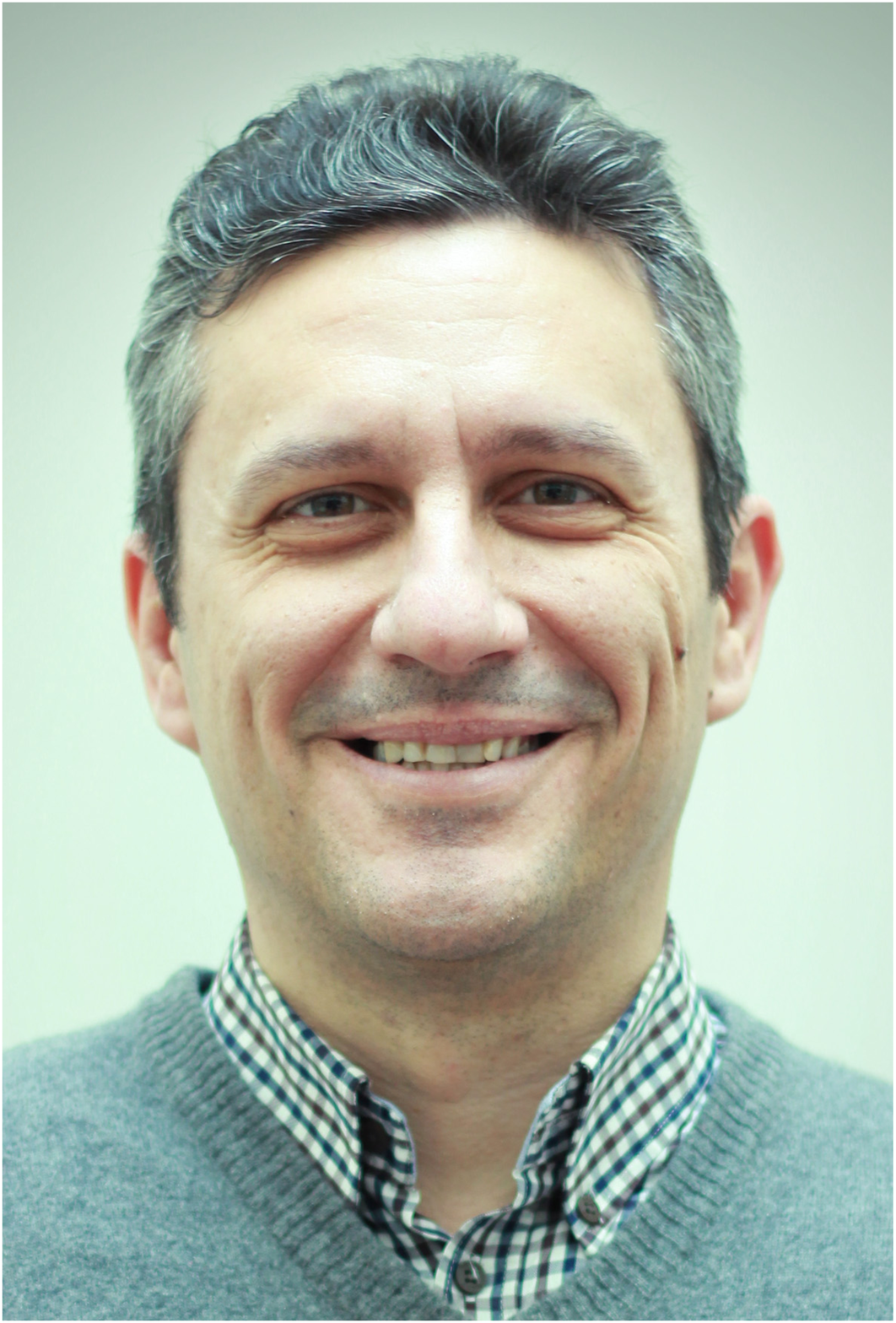}}]
{Theodoros A. Tsiftsis} (S'02, M'04, SM'10) was born in Lamia, Greece, in 1970. He received the B.Sc. degree in physics from the Aristotle University of Thessaloniki, Greece, in 1993, the M.Sc. degree in digital systems engineering from the Heriot-Watt University, Edinburgh, U.K., in 1995, the M.Sc. degree in decision sciences from the Athens University of Economics and Business, in 2000, and the Ph.D. degree in electrical engineering from the University of Patras, Greece, in 2006. He is currently a Professor in the School of Intelligent Systems Science \& Engineering at Jinan University, Zhuhai Campus, Zhuhai, China. He has authored or co-authored over 180 technical papers in scientific journals and international conferences. His research interests include the broad areas of cognitive radio, communication theory, wireless powered communication systems, optical wireless communication, and ultra-reliable low-latency communication.

Prof. Tsiftsis has been appointed to a 2-year term as an IEEE Vehicular Technology Society Distinguished Lecturer (IEEE VTS DL), Class 2018. Dr. Tsiftsis acts as a Reviewer for several international journals and conferences. He has served as Senior or Associate Editor in the Editorial Boards of \textsc{IEEE Transactions on Vehicular Technology}, \textsc{IEEE Communications Letters},  \textsc{IET Communications}, and \textsc{IEICE Transactions on Communications}. He is currently an Area Editor for Wireless Communications II of the \textsc{IEEE Transactions on Communications} and an Associate Editor of the \textsc{IEEE Transactions on Mobile Computing}.
\end{IEEEbiography}
\end{document}